\documentclass[12pt]{article}

\usepackage[margin=1in]{geometry}
\usepackage{setspace}
\usepackage[T1]{fontenc}
\usepackage[utf8]{inputenc}

\usepackage{amsmath,amssymb,amsfonts,amsthm,mathtools}
\usepackage{natbib}
\usepackage{graphicx}
\usepackage{booktabs}
\usepackage{multirow}
\usepackage{algorithm}
\usepackage{algorithmic}
\usepackage{enumitem}
\usepackage{bbm, bm}
\usepackage{multirow}
\usepackage{booktabs}
\usepackage{mathtools}
\usepackage{commath}
\usepackage{algorithm}
\usepackage{algorithmic}
\newtheorem{prop}{Proposition}
\newtheorem{lem}{Lemma}
\allowdisplaybreaks

\newcommand{\on}{{(1)}} %

\newcommand{\sn}{\sum_{i=1}^{n}} %
\newcommand{\E}{\mathbb{E}} %
\newcommand{\scj}{\sum_{j=1}^{K-1}}
\newcommand{\scl}{\sum_{l=1}^{K-1}}
\newcommand{\onj}{{(j)}}
\newcommand{\onjc}{{(j^c)}}

\usepackage[
    colorlinks=true,
    citecolor=blue,
    linkcolor=blue,
    urlcolor=blue
]{hyperref}

\newtheorem{theorem}{Theorem}[section]

\theoremstyle{definition}

\title{\bfseries Fast high-dimensional mean testing via logistic regression}

\author{
Sayan Das\thanks{Department of Statistics and Data Science, Washington University in St. Louis, St. Louis, MO, USA}
\and
Debraj Das\thanks{Department of Mathematics, Indian Institute of Technology Bombay, Mumbai, India}
\and
Subhajit Dutta\thanks{Applied Statistics Unit, Indian Statistical Institute, Kolkata, India}
}

\date{}

\begin{document}

\maketitle

\vspace{0.4cm}

\begin{abstract}
We propose computationally efficient tests for equality of mean vectors of two or more high-dimensional populations. Central to our approach is an equivalence between equality of means and a zero population logistic regression parameter. We establish this equivalence for independently distributed observations without imposing common distributional assumptions across populations. Our procedure uses logistic Lasso to screen informative variables and an unpenalized logistic refit for inference in the reduced dimension, yielding asymptotically correct size and consistency. For a specified two-sample Gaussian submodel and sparse discriminative class, the test also attains the minimax separation rate. The framework extends to multiple populations through multi-class logistic regression. Simulations demonstrate accurate size control, strong power, and favorable computational scaling compared with existing tests under unbalanced designs and variance heterogeneity. Applications to gene-expression data with more than twenty-two thousand variables illustrate the practical scalability of the proposed procedures.
\end{abstract}

\vspace{0.5em}


\section{Introduction}\label{sec:intro}

Testing equality of mean vectors from two or more populations is a fundamental problem in high-dimensional statistics with wide applications. Traditional procedures such as Hotelling's $T^2$ test may not be applicable when the data dimension is comparable to or larger than the sample size, and can also be numerically unstable when the sample covariance matrix is nearly singular. This has led to a rich literature comprising new methods and theory for such data settings. One key idea underlying many existing tests is to measure a distance between suitably standardized sample means. Although such procedures can be effective, many of them either do not fully utilize dependence among the variables or require estimation of a high-dimensional covariance or precision matrix. Projection-based procedures offer an alternative route to improved power, but identifying and estimating a projection is often computationally challenging in high dimensions.

Our main idea is to transform the mean-testing problem into a logistic regression problem. By using population membership as the response, we show that the population means are equal if and only if the corresponding population logistic parameter is zero. This equivalence makes it possible to use logistic Lasso to screen informative variables and then construct the test in a reduced dimension, without estimating or inverting the full covariance matrix of the observations. We assume that observations are independent and identically distributed within each population, while allowing the population distributions to differ beyond their means.

\subsection{Existing literature}

The classical Hotelling's $T^2$ test for testing equality of the means of two populations was extended to high dimensions by \cite{bai1996effect}. Other influential contributions to the two-sample problem include \cite{chen2010two} and \cite{tony2014two}. More recently, \cite{yang2024new} proposed a correlation-aware test based on regularized estimation of the precision matrix under a linear-structure assumption. \cite{feng2024asymptotic} developed a test based on the asymptotic independence between the sum and maximum of suitable statistics. \cite{Huang2022} provide an overview and numerical comparison of a broad collection of high-dimensional two-sample mean tests.

The assumptions governing the relationship between the population distributions vary considerably across this literature. \cite{chen2010two} allow unequal covariance matrices through a common latent factor or linear transformation representation, but the populations retain a shared standardized latent structure. The non-Gaussian formulation of \cite{Cai2014} assumes a common centered distribution and covariance matrix, so that the populations differ through their locations. For the optimal-projection test of \cite{huang2015projection}, the non-Gaussian theory likewise retains a common covariance matrix, whereas its unequal-covariance extension is developed for Gaussian populations.

Some existing methods do accommodate heterogeneous populations. \cite{Xue2020} developed a bootstrap-based sup-norm test under coordinate-wise moment and exponential-tail conditions. On the contrary, we require a weaker sub-Weibull condition that includes the sub-exponential setting as a special case and permits heavier tails when the sub-Weibull parameter is below one. \cite{Kong2022} proposed consistent permutation procedures under exponential-tail conditions; their theory for the Hotelling-\(T^2\)-based procedure additionally assumes uniformly bounded eigenvalues of the full population covariance matrices. By contrast, our power analysis does not impose a global eigenvalue condition on the full data covariance matrices: its invertibility and non-degeneracy requirements concern only an active-set block of the expected logistic information matrix, whose dimension is determined by the sparse population logistic parameter.

Our numerical comparisons focus on procedures most directly related to the proposed tests. In the two-sample setting, we compare with the tests of \cite{chen2010two} (CQ), \cite{Xue2020} (XY), \cite{huang2015projection} (OP), and \cite{Kong2022} (ES). For the multi-sample case, \cite{Chakraborty2023} develop a bootstrap-based sup-norm test for testing linear hypotheses for high-dimensional means (say, the CS test), while \cite{Li2023} propose tests based on a generalized Hotelling’s $T^2$ statistic (referred to as the HDT test). We use these tests for numerical comparisons in Section \ref{sec:simulation}.

\subsection{Our contributions}

We first describe the two-sample formulation. Let $X_1,\ldots,X_{n_1}$ be independent and identically distributed (iid) from a non-degenerate $p$-dimensional distribution $F_1$ with mean vector $\mu_1$, and let $Y_1,\ldots,Y_{n_2}$ be iid samples from a non-degenerate $p$-dimensional distribution $F_2$ with mean vector $\mu_2$. The hypothesis of interest is

\begin{equation*}
    H_0:\mu_1=\mu_2
    \text{ vs. }
    H_1:\mu_1\ne\mu_2.
\end{equation*}

Pool the observations using their population prior probabilities, let the population indicator be the response in a logistic regression, and include the class-prior log-odds as an offset. We define $\beta$ through the population logistic score equation, equivalently as the pseudo-true minimizer of the population logistic risk; therefore, the conditional class probability need not follow a correctly specified linear logistic model. Our first key result establishes that
\begin{equation*}
    \mu_1=\mu_2
    \quad\Longleftrightarrow\quad
    \beta=0.
\end{equation*}
Consequently, the original mean-testing problem is equivalent to testing $H_0':\beta=0$ against $H_1':\beta\ne0$. This transformation is the main methodological device of the paper: it converts a high-dimensional mean problem into a regression-coefficient testing problem for which sparse logistic regression provides a computationally tractable route to an informative low-dimensional representation.

We construct the test using logistic Lasso (see \cite{tib1996}) followed by an unpenalized logistic refit on the selected variables. The Lasso step screens the relevant components of $\beta$ even when the ambient dimension is exponentially large relative to the sample size, while the post-Lasso step yields an estimator suitable for inference. The resulting procedure performs inference only in the selected dimension and avoids inversion of a full $p$-dimensional covariance or precision matrix. Under uniform sub-Weibull tails and the stated sparsity and design conditions, we establish asymptotically correct size and consistency; see Theorems \ref{theo:size} and \ref{theo:power}.

The sparsity assumption is placed on the discriminative parameter $\beta$, rather than directly on the mean difference $\mu_1-\mu_2$. A sparse discriminative set need not correspond to a sparse mean difference, so the proposed procedure can remain effective under some dense mean alternatives. For the two-sample Gaussian model and the sparse discriminative class of Section~\ref{sec:optimality}, we derive the minimax separation rate for the active components of $\beta$ and show that the proposed test attains the rate.

The same principle extends naturally to more than two populations. For $K\ge2$ populations, we prove that equality of all $K$ mean vectors is equivalent to a zero population multi-class logistic parameter. We then construct a multi-sample test using multi-class logistic Lasso and establish its asymptotic size control and consistency under sparsity of the multi-class discriminative set; see Theorem~\ref{theo:multi}. The two- and multi-sample procedures therefore share the same transformation, screening, and reduced-dimensional inference framework.

Computational efficiency is another central feature of the proposed approach. The basic procedure D3 uses logistic Lasso for screening and then performs inference in the selected lower-dimensional space, substantially reducing computational cost relative to most competing tests. D3op additionally uses label permutations to calibrate the Lasso penalty through directly computed maximal-penalty thresholds. Figure~\ref{fig:test} shows that across the two- and three-sample settings, D3 remains fast and relatively stable as the dimension increases, while D3op retains a practical balance between adaptive penalty selection and computation time. Section~\ref{sec:simulation} reports the full simulation timings with the current competitors, together with the size and power results.

\begin{figure}
\centering
\includegraphics[width=0.8\linewidth]{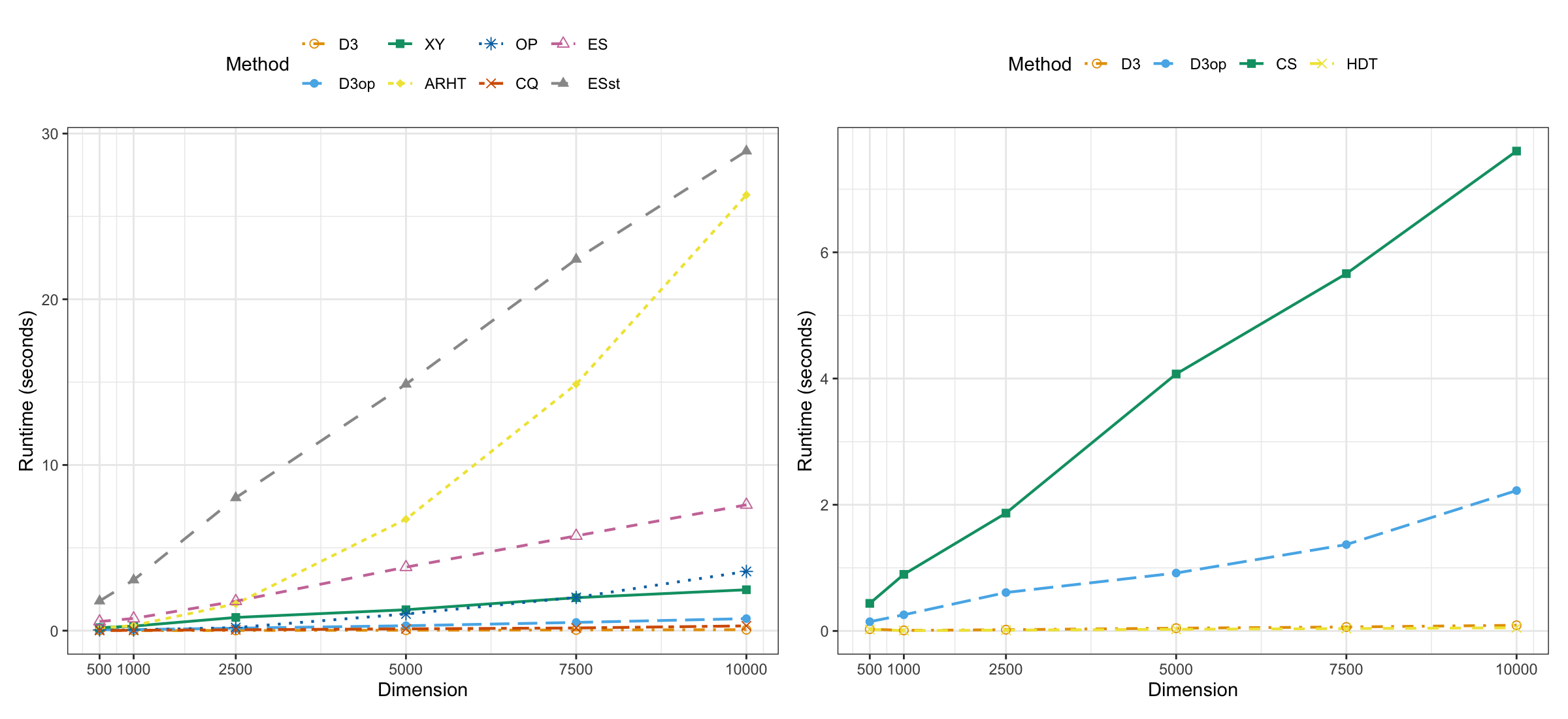}
\caption{Computing time of some existing tests for $p=500, 1000, 2500, 5000, 7500$ and $10000$ with 75 observations from each population. The left panel shows two-sample tests, while the right panel shows three-sample tests.}
\label{fig:test}
\end{figure}

Simulations further demonstrate accurate size control under balanced and unbalanced designs and strong power across a range of alternatives, including coordinate-wise variance heterogeneity. Applications to gene-expression data with more than twenty-two thousand variables illustrate that the proposed multi-sample procedures can detect meaningful group differences while remaining computationally feasible. Together, the theory and numerical results show how the logistic-regression transformation combines sparse discriminative modeling, reduced-dimensional inference, and scalability in a unified way for two- and multi-sample mean testing.

\subsection{Organization of the paper}

The two-sample scenario is considered in Section~\ref{sec:ttsm}. The transformation of the two-sample mean-testing problem into a test for the population logistic parameter is discussed in Section~\ref{sec: transform}. The proposed test is constructed in Section~\ref{sec:test stat}, and its asymptotic analysis is presented in Section~\ref{sec:tpower}. Minimax rate optimality under Gaussianity is established in Section~\ref{sec:optimality}. The extension to the multi-sample setup is presented in Section~\ref{sec:test stat multi}. Section~\ref{sec:simulation} contains an extensive simulation study comparing the finite-sample performance and computational cost of the proposed procedures with existing tests. Section~\ref{sec:real} presents applications to gene-expression data. Proofs of the main results and auxiliary lemmas, together with additional simulation studies and data analyses, are relegated to the Supplementary Material.

\section{Two-sample mean testing} \label{sec:ttsm}

Let $X_1,\dots, X_{n_1}$ be a random sample from a $p$-dimensional distribution $F_1$ and $Y_1,\dots, Y_{n_2}$ be from another $p$-dimensional distribution $F_2$. Assume that $F_1$ and $F_2$ may differ in their mean vectors.
For ease of presentation, we unify the two samples and define a random vector $W$ which can be generated either from $F_1$ or $F_2$ with nonzero probabilities $\pi_1$ and $\pi_2$, respectively, with $\pi_1+\pi_2 =1$. Here, $\pi_1$ and $\pi_2$ can be thought of as the corresponding prior probabilities for the underlying two populations. Defining $n = n_1+n_2$, $n_k/n$ is clearly a natural estimator of $\pi_k$ for $k=1,2$. In the next subsection, we translate the aforementioned testing problem into a testing problem for the parameter vector $\bm{\beta}$ in the logistic regression framework.

\subsection{Transformation to logistic regression} \label{sec: transform}

Consider the logistic regression (LR) based on the two samples $X_1,\dots, X_{n_1}$ and $Y_1,\dots, Y_{n_2}$ using a random variable $z$ where $z = 0$ or $1$ according to whether $W \sim F_1$ or $W \sim F_2$. Thus, we can devise an LR setup with $z_i = 0$ for $i=1,\dots, n_1$ and $z_i=1$ for $i=n_1+1,\dots,n$ being responses corresponding to the covariates $W_1,\dots, W_n$, where $W_i = X_i$ for $i=1,\dots, n_1$ and $W_i=Y_{i-n_1}$ for $i=(n_1+1),\dots,n$. As a consequence, $(0, X_1^\top)^\top, \dots, (0, X_{n_1}^\top)^\top,$ $(1, Y_1^\top)^\top,\dots, (1, Y_{n_2}^\top)^\top$ can be seen as independent and identically distributed (iid) realizations of $(z, W^\top)^\top$.

In LR, the logarithm of the odds (or logit) is modeled as a linear function of the covariates and hence, 
the logistic regression model is given by 
\begin{align*}
\log\bigg[\dfrac{P(z=1|W)}{1-P(z=1|W)}\bigg]= \log\big(\pi_2/\pi_1\big) + \beta_{0} + \bm{\beta}^{(1)\top}W,
\end{align*}
where $\bm{\beta}=(\beta_0, \bm{\beta}^{(1)\top})^\top = (\beta_{0}, \beta_1,$ $\dots, \beta_p)^\top$ is the $(p+1)$-dimensional regression parameter vector. Clearly, $\log\big(\pi_2/\pi_1\big)$ captures the discrepancy in the observed sample sizes and disappears when $\pi_1 = \pi_2$ (representing the scenario when the sample sizes $n_1$ and $n_2$ are \textit{equal}). 
Since the logistic function is strictly increasing, regression parameter $\bm{\beta}$ can equivalently be defined as the unique solution of 
\begin{equation}\label{eqn:pdef}
\mathbb{E}\Big[ \left(z - p(\bm t|W)  \right)  ( 1^\top, W^\top )^\top
   \Big] = 0, \text{ where } p(\bm t|W) =\dfrac{\pi_2e^{t_0 + W^{\top}\bm{t}^{(1)}}}{\pi_1+\pi_2e^{t_0 +W^{\top}\bm{t}^{(1)}}}. 
\end{equation} 

It is quite natural to expect that $\mu \equiv \mu_1-\mu_2$ and $\bm{\beta}$ are related. In fact, if the underlying populations are Gaussian, it is known~that
\begin{equation} \label{minimax_beta}
\beta_0 = -\frac{1}{2} (\mu_2+\mu_1)^\top\Sigma^{-1}(\mu_2-\mu_1)
\text{ and } \bm\beta^{(1)} = \Sigma^{-1}(\mu_2-\mu_1),
\end{equation} 
where $\Sigma$ is the common non-singular covariance matrix (see, e.g., \cite{hastie2009elements}). Clearly, $\mu_1=\mu_2$ if and only if $\bm{\beta} = \bm{0}$ under Gaussianity. For general distributions, an exact relation like equation (\ref{minimax_beta}) may not hold. 
For example, for general elliptical distributions (other than the Gaussian), this relation (\ref{minimax_beta}) is true only in some special cases.
Hence, the equivalence of $\mu_1=\mu_2$ with $\bm{\beta} = \bm{0}$ cannot be claimed directly. However, the following proposition shows that this fact is indeed true for general non-degenerate probability~distributions.

\begin{prop}\label{prop:transformation}
$\mu_1=\mu_2$ if and only if $\bm{\beta} = \bm{0}$.
\end{prop}

The above equivalence is pivotal in developing our approach to test the equality of high-dimensional means. So, we present a brief outline of the proof of this proposition under the simpler setup when the distributions of $X_1$ and $Y_1$ are symmetric and absolutely continuous and $\pi_1 = \pi_2 = 1/2$. The general case (with more than two populations) is presented as Proposition \ref{prop:multitransformation} and its proof is relegated to the Supplementary Material. First, assume that $\bm\beta = \bm{0}$. 
Using equation (\ref{eqn:pdef}), we have
\begin{align*}
    0 & =
    \mathbb{E}\Big((z- 1/2)W\Big) \\
    & = \mathbb{E}\Big((z- 1/2)W|W \sim F_1\Big) P(W \sim F_1) + \mathbb{E}\Big((z- 1/2)W|W \sim F_2\Big) P(W \sim F_2) \\
    & = 1/2\Big[\mathbb{E}\Big(-\dfrac{X_1}{2}\Big) + \mathbb{E}\Big(\dfrac{Y_1}{2}\Big)\Big] = 1/4(\mu_2 - \mu_1),
\end{align*}
which implies $\mu_1 = \mu_2$. 
Next, assume that $\mu_1 = \mu_2$. From (\ref{eqn:pdef}), we have
\begin{equation}
    \label{eqn:star1}
    \mathbb{E} \left[ \left(\frac{ -\exp(\beta_0+X_1^\top{\bm\beta^{(1)}}) }{1+\exp(\beta_0+X_1^\top{\bm\beta^{(1)}})}\right) \begin{pmatrix} 1 \\ X_1 \end{pmatrix} + \left(\frac{ 1}{1+\exp(\beta_0+Y_1^\top{\bm\beta^{(1)}})}\right) \begin{pmatrix} 1 \\ Y_1 \end{pmatrix}\right] = 0. 
\end{equation}
Clearly, $(\beta_0, {{\bm\beta^{(1)}}}^\top)^\top = \bm{0}$ is a solution of (\ref{eqn:star1}). Moreover, if ${\bm\beta^{(1)}} = \bm{0}$, then $\beta_0$ must be $0$. If possible, let ${\bm\beta^{(1)}} \neq \bm{0}$. Then $(X_1-\mu_1)^\top{{\bm\beta^{(1)}}}$ and $(Y_1-\mu_1)^\top{{\bm\beta^{(1)}}}$ (as $\mu_1 = \mu_2$) have absolutely continuous distributions symmetric around $0$. 
Denote the density of $(X_1-\mu_1)^\top {{\bm\beta^{(1)}}}$ by $f_1$ and density of $(Y_1-\mu_1)^\top {{\bm\beta^{(1)}}}$ by $f_2$ and their average by $f_{1,2} \equiv (f_1+f_2)/2.$ Therefore, from equation (\ref{eqn:star1}), we obtain
\begin{align}
     \int_{-\infty}^{\infty} \begin{pmatrix} 1 \\ y \end{pmatrix}\Bigg[ \frac{f_2 (y-\beta_0- {{\bm\beta^{(1)}}}^\top \mu_1) - \exp(y) f_1 (y-\beta_0- {{\bm\beta^{(1)}}}^\top \mu_1)}{1+\exp(y)}\Bigg] \; dy = 0. \label{eqn:star2}
\end{align} 
Now for any density $f$ and for any $a \in \mathbb{R}$, $\Big[1 -\int_{-\infty}^{\infty} \frac{ f(z-a) \exp(z)}{1+\exp(z)}\;dz\Big] = \int_{-\infty}^{\infty} \frac{f(z-a)}{1+\exp(z)}\;dz.$ Thus, from (\ref{eqn:star2}), we have
\begin{align}
    &\int_{-\infty}^{\infty} \frac{1-\exp (y)}{1+\exp(y)}f_{1,2} (y-\beta_0- {{\bm\beta^{(1)}}}^\top \mu_1) \; dy = 0. \label{eqn:star200}
\end{align}
Since $\frac{1-\exp(\cdot)}{1+\exp(\cdot)}$ is an odd function and $f_{1,2}(\cdot)$ is symmetric around $0$,
we get $\beta_0 + {{\bm\beta^{(1)}}}^\top \mu_1=0$ from (\ref{eqn:star200}). Thus, from (\ref{eqn:star2}), we have
\begin{align}
     &\int_{-\infty}^{\infty} y\Bigg[ \frac{f_2 (y) - \exp(y) f_1 (y)}{1+\exp(y)}\Bigg] \; dy = 0. \label{eqn:star2000}
     \end{align}
Now, for any density $f$ satisfying $\int_{-\infty}^{\infty}zf(z) = 0$, we have
$\int_{-\infty}^{\infty} \frac{ zf(z)}{1+\exp(z)}\;dz = -\int_{-\infty}^{\infty} \frac{z\exp(z)f(z)}{1+\exp(z)}\;dz.$
Therefore, (\ref{eqn:star2000}) reduces to
$$\int_{-\infty}^{\infty} y\frac{1-\exp(y)}{1+\exp(y)}f_{1,2}(y)\;dy = 0,$$
which cannot be true since the function $z~\frac{1-\exp(z)}{1+\exp(z)}$ is a non-positive even function and $f_{1,2}$ is a density function. Therefore, we arrive at a contradiction. Hence, the assumption that ${\bm\beta^{(1)}} \neq \bm{0}$ cannot be true and the proof of the special case is complete.

In the next subsection, we construct an estimator of $\bm{\beta}$ which we then utilize to construct a procedure to test an equivalent hypothesis
\begin{equation} \label{hyp_beta}
H_0^\prime: \bm{\beta} = \bm{0} \text{ vs. } H_1^\prime: \bm{\beta} \neq \bm{0}.    
\end{equation}

\subsection{Construction of the test using logistic Lasso regression}\label{sec:test stat}
In this subsection, we construct a testing procedure based on a suitable estimator of the LR parameter $\bm{\beta}$. Since the dimension of the underlying populations may be large compared to the sample size $n$, the regression parameter $\bm{\beta}$ may be high-dimensional. Hence, the maximum likelihood estimator (MLE) of $\bm{\beta}$, the most natural estimator of $\bm{\beta}$, may be sub-optimal and may not even exist, as observed by \cite{candes2020phase} when $p \ge n$. A natural way out is to assume that the parameter $\bm{\beta}$ lies in a lower-dimensional space and accordingly, one may consider a penalized version of the MLE as a potential estimator of $\bm{\beta}$. 

Several penalized estimators have been introduced for the linear regression setting (more generally, for generalized linear models). Among the different penalized estimators, the most well-known method is the Lasso introduced by \cite{tib1996}. 
One reason behind the popularity of Lasso is its computational feasibility in high-dimensional regression problems (see, e.g., \cite{Efron2004, Friedman2007, Fu1998, Osborne2000}). 
A natural estimator of $\pi_k$ is $n_k/n$ for $k=1, 2$. So, the logistic Lasso estimator can naturally be defined as 
\begin{equation*}
\begin{split}
    \tilde{\bm{\beta}}_n = \underset{(t_0, \bm{t}^\top)^\top \in \mathbb{R}^{p+1}}{\arg\min}& \Big[-\sum_{i=1}^n z_i (t_0 + W_i^\top \bm{t}) + \sum_{i=1}^{n} \log (n_1+n_2\exp(t_0 + W_i^\top\bm{t})) +\lambda_n \sum_{i=0}^{p} |t_j|\Big], 
\end{split}
\end{equation*}
which is the $\ell_1$-penalized log-likelihood of the LR model based on $\{(z_i, W_i)\}_{i=1}^n$. 

Define $\mathcal{A}_n = \{0\leq j\leq p: \beta_{j}\neq 0\}$, the set of true non-zero regression coefficients and let
$\tilde{\mathcal{A}}_n = \{0\leq j\leq p: \tilde{\beta}_{j,n}\neq 0\}$ is the estimator of $\mathcal{A}_n$ based on $\tilde{\bm{\beta}}_n$. 
For a suitable choice of the penalty parameter $\lambda_n$, the logistic Lasso estimator $\tilde{\bm{\beta}}_n$ can be shown to be variable selection consistent (VSC), i.e., 
$\mathbbm{P}(\tilde{\mathcal{A}}_n=\mathcal{A}_n)\rightarrow 1 \text{ as } n \to \infty.$
For such choices of $\lambda_n$, the components of $\tilde{\bm{\beta}}_n$ can be shown to be not $\sqrt{n}-$ consistent. 
See \cite{lahiri2021} and \cite{Chak26} for this duality result of the Lasso for linear regression.
However, the VSC property and $\sqrt{n}$-consistency are both important to construct an effective testing procedure for the hypothesis stated in \eqref{hyp_beta}. In view of these reasons, we will construct our test based on the post-Lasso logistic estimator $\hat{\bm{\beta}}_n = \Big(\big(\hat{\bm{\beta}}^{\tilde{\mathcal{A}}_n}_n\big)^\top, \bm{0}^\top\Big)^\top$. Define $\dot{W}_i = (1, W_i^\top)^\top$ for $1\leq i \leq n$. Then, $\hat{\bm{\beta}}^{\tilde{\mathcal{A}}_n}_n$ is defined as
\begin{align*}
\bm{\hat\beta}^{\tilde{\mathcal{A}}_n}_n = \underset{\bm{t} \in \mathbb{R}^{|\tilde{\mathcal{A}}_n|}}{\arg\min}  &\bigg[-\sum_{i=1}^n z_i \big(\dot{W}_{i}^{\tilde{\mathcal{A}}_n}\big)^\top \bm{t}
+ \sum_{i=1}^n \log \Big(n_1+n_2\exp\Big(\big(\dot{W}_{i}^{\tilde{\mathcal{A}}_n}\big)^\top\bm{t}\Big)\Big)\bigg],
\end{align*}
where $\bm{x}^{\tilde{\mathcal{A}}_n}$ is the sub-vector of a vector $\bm{x}$ containing the components in $\tilde{\mathcal{A}}_n$.
Note that $\hat{\bm{\beta}}_n$ is a two-step estimator. First one needs to compute the logistic Lasso estimator $\tilde{\bm{\beta}}_n$ and identify the estimated active set $\tilde{\mathcal{A}}_n$ of covariates, and then one defines $\hat{\bm{\beta}}_n$ in terms of the LR estimator computed using the covariates present in $\tilde{\mathcal{A}}_n$.
The above definition of $\hat{\bm{\beta}}_n$ is motivated by \cite{Belloni2013}, who introduced the post-Lasso OLS estimator in linear regression. When $\tilde{\bm{\beta}}_n$ is VSC, $\hat{\bm{\beta}}_n$ trivially satisfies the VSC property due to its definition and the components of the active part $\bm{\hat\beta}^{\tilde{\mathcal{A}}_n}_n$ are obviously $\sqrt{n}$-consistent for $\bm{\beta}^{\mathcal{A}_n}$, since $\bm{\hat\beta}^{\tilde{\mathcal{A}}_n}_n$ is actually the LR estimator of $\bm{\beta}^{\mathcal{A}_n}$ on the active set $\mathcal{A}_n$ with probability tending to~$1$.

Under $H_0^\prime: \bm{\beta} = \bm{0}$, the estimator $\hat{\bm{\beta}}_n$ will be $\bm{0}$ with probability tending to $1$ due to the VSC property. To achieve the given significance level $\alpha$  for some $\alpha \in (0, 1)$, we introduce the $(p+1)$-dimensional standard Gaussian random vector $\bm{Z}_1$ and subsequently, 
define the test statistic as
$\bm{T}_n^{(1)} =
   \sqrt{n}\hat{\bm{\beta}}_n 
   + b_1(p, n)\bm{Z}_1,$
where, $b_1(p, n)$ is a positive sequence of $p$ and $n$ tending to zero that serves as a scaling factor for $\bm{Z}_1$ since we are in a high-dimensional setting.
Now, define the quantity $m_{n,p,\gamma}^{(1)}> 0$ such that $\mathbb{P}\left( \underset{j}{\max}\; |Z_{1j}| \leq \frac{m_{n,p,\gamma}^{(1)}}{b_1(p, n)} \right) = \gamma \in  (0,1)$. 
Let $\mathbbm{I}(\cdot)$ be the indicator function and $T_{nj}^{(1)}$ be the $j$th component of $T_{n}^{(1)}$. Then at level $\alpha$, we define our test function  
$$\psi_{n, \alpha}^{(1)} = \mathbbm{I} \left \{ \underset{j}{\max}\; |T_{nj}^{(1)}| > m_{n,p,1-\alpha}^{(1)} \right \}.$$

To study the size of the test $\psi_{n, \alpha}^{(1)}$ in high dimensions, we need some tail condition on the underlying population distribution. For the remainder of this section, we shall assume that the distribution $F_i$ is sub-Weibull (see, e.g., \cite{kuchi22}). In particular, we assume $\underset{1\leq i \leq n}{\max}\; \underset{1\leq j \leq p}{\max}\|W_{ij}\|_{\zeta_\kappa} < C.$
The class of sub-Weibull distributions clearly generalizes the classes of sub-Gaussian and sub-Exponential distributions
(which correspond to $\kappa = 2$ or $1$, respectively). We are now ready to state the first theoretical result:

\begin{theorem}\label{theo:size}
Suppose that $|n_k - n\pi_k|= o\big(n^{1/2}\big)$ for  $k=1,2$ and $$\dfrac{\lambda_n}{\sqrt{n}}>C_1 \max\left\{\sqrt{\log(np)},\frac{(\log(n))^{1/\kappa}(\log(np))^{1/\min\{1,\kappa\}}}{\sqrt{n}} \right\}$$ for some constant $C_1 \geq  1$. Then, for any $\alpha \in (0, 1)$, we have $\mathbb{P}_{H_0}\big(\tilde{\bm{A}}_n = \emptyset\big)\rightarrow 1$ as $n\rightarrow \infty$ 
and hence, $$\mathbb{E}_{H_0}\big(\psi_{n, \alpha}^{(1)}\big) \rightarrow \alpha\;\; \text{as}\;\; n \rightarrow \infty.$$ 
\end{theorem}
\vspace*{-0.4cm}
Theorem \ref{theo:size} establishes that for any $\alpha \in (0,1)$, the test $\psi_{n, \alpha}^{(1)}$ achieves the nominal size asymptotically. This essentially shows that the test $\psi_{n, \alpha}^{(1)}$ is a valid test for the hypothesis \eqref{hyp_beta}
at any given significance level $\alpha \in (0, 1)$. 
Before moving to the power analysis of the test $\psi_{n, \alpha}^{(1)}$, we would like to point out that one may construct a test based on other estimators (not necessarily the post-Lasso estimator) of $\bm{\beta}$. For example, in low dimensions, i.e., when the dimension $p$ is fixed or grows like a fractional power of $n$, then a test statistic may be defined based on the original LR estimator. In high dimensions, one existing alternative to the post-Lasso estimator is the debiased Lasso estimator proposed by {Ma et al. (2020)}. However, the testing procedure based on the debiased Lasso estimator would be computationally quite slow, for reasons similar to the testing procedure developed by \cite{tony2014two}.

\subsection{Asymptotic power analysis}\label{sec:tpower}

In this subsection, we analyze the power of the test $\psi_{n, \alpha}^{(1)}$ in high dimensions. We carry out power analysis for the alternative corresponding to 
\begin{equation} \label{hyp_beta_local}
H_0^\prime: \bm{\beta} = \bm{0}\text{ vs. } H_1^\prime: \bm{\beta} \in  \Theta^\prime(p_0)= \{\bm{\beta} \in \mathbb{R}^{p+1}: \|\bm{\beta}\|_0 = p_0+1 \},
\end{equation}
for some $p_0\le p.$ Here, $\|\cdot\|_0$ denotes the $\ell_0$ metric, and $\Theta^\prime(p_0)$ is essentially the informative set which dictates why $H_0^\prime$ may be false. 
The structure of $H_1^{\prime}$ is motivated by the LR classifier utilized to transform $H_0: \mu_1 =\mu_2$ to $H_0^\prime: \bm{\beta}= \bm{0}$ (see Section \ref{sec: transform} for more~details).

Moreover, when the populations are Gaussian, the LR classifier is the optimal discriminant method (see, e.g., \cite{hastie2009elements}), implying that $\Theta^\prime(p_0)$ is actually the discriminative set. In high dimensions (in a large-$p$, small-$n$ setting), $p_0$ must be substantially smaller than $p$, since without such sparsity a discriminant method may fail (see, e.g., \cite{fan2008high} and \cite{mai2012}). However, the sparsity of the discriminative set does not necessarily induce sparsity in the signal set $\{\mu_1- \mu_2: \|\bm{\beta}\|_0 = p_0 + 1\}$, and therefore our proposed test may work even under a dense alternative in terms of $\mu_1-\mu_2$. One can see Proposition 1 of \cite{mai2012} for the connection between the discriminative set and the signal set under Gaussian.  

Now, we are ready to state the regularity conditions required to develop asymptotic power analysis of our test $\psi_{n, \alpha}^{(1)}$ under $H_1$. However, first we define some additional notation. 
For any matrix $\bm{M}$, let $\bm{M}_{j\cdot}^\top$ denote the $j$th row of $\bm{M}$ and $\|\bm{M}\|_2$, $\|\bm{M}\|_1$ and $\|\bm{M}\|_\infty$ denote the spectral norm, $\ell_1$ norm and $\ell_\infty$ norm, respectively, of $\bm{M}$. For any vector $\bm{l}\in \mathbb{R}^{m}$, $\bm{l}_j$ denotes the $j$th component of $\bm{l}$. For a real number $k\geq 1$, $\|\bm{l}\|_k = \Big(m^{-1}\sum_{j=1}^{m}|l_j|^k\Big)^{1/k}$ and 
$\|\bm{l}\|_{\infty}=\max\{|l_{m}|: 1\leq j \leq m\}$ denote the $\ell_k$ and the sup norm, respectively. When $k=2$, we denote $\|\bm{l}\|_2$ simply as $\|\bm{l}\|$. Define $sgn(x) =-1, 0 ,1$ according as $x<0$, $x=0$, $x>0$, respectively. Recall that $n = n_1+n_2$. Without loss of generality, assume that under $H_1^\prime$, $\mathcal{A}_n \equiv \mathcal{A}_n(p_0)=\{j: \beta_{j}\neq 0\} = \{0,\dots,p_0\}$ is the set of active indices in the discriminative set $\Theta^\prime(p_0)$. Now, recall that $\dot{W}_i = (1, W_i^\top)^\top$ for $1\leq i \leq n$. Define $(p+1)\times (p+1)$ matrix $\bm{L}_n$ as 
$\bm{L}_n  = 
\frac{1}{n}\sum_{i=1}^n \dot{W}_i\dot{W}_i^\top\frac{\pi_2\exp\big(\dot{W}_i ^\top\bm{\beta}\big)}{\pi_1+\pi_2\exp\big(\dot{W}_i ^\top\bm{\beta}\big)}.$
Consider the partition of  $\bm{L}_n$ with respect to $\mathcal{A}_n$ as follows:
\begin{align*}
\bm{L}_{n} = \begin{bmatrix}
\bm{L}_{11,n}\;\;\;\bm{L}_{12,n}\\
\bm{L}_{21,n}\;\;\; \bm{L}_{22,n}
\end{bmatrix}.
\end{align*}
Under VSC, $\bm{L}_{11, n}$ is the variance of $\hat{\bm{\beta}}_n^{\tilde{\mathcal{A}}_n}$. Now, for some $\delta_1, a_1 \in (0, 1]$ and $ \tau_1 \in (0, 1)$, consider the following regularity conditions whenever $n \geq \delta_1^{-1}$. 
Under $H_1^\prime$ and for all $\bm{\beta} \in \Theta^\prime(p_0)$, let us consider the following:
\begin{enumerate}
    \item[(A.1)] $\underset{p_0+1\leq j\leq p}{\max}\left\{|\big(\mathbb{E}\bm{L}_{21, n}\big)_{j \cdot}^\top \big(\mathbb{E} \bm{L}_{11, n}\big)^{-1} sgn\big(\bm{\beta}_n^{(1)}\big) |\right\} \leq 1-\tau_1$.
    \item[(A.2)]
    \begin{enumerate}
    \item[(i)] $\|(\mathbb{E}\bm{L}_{11, n})^{-1}\|_\infty \leq \delta_1^{-1} n^{a_1}$. 
    \item[(ii)]
    $\Big[\underset{0 \leq j \leq p_0}{\min}\Big(\big(\mathbb{E}\bm{L}_{11, n}\big)^{-1}\Big)_{jj}
    \Big] \geq \delta_1$.
    \item[(iii)] $\underset{0\leq j\leq p}{\max} \Big\|  \mathbb{E}\Big(|\dot{W}_{1j}| \dot{W}^{(1)}_{1} \dot{W}^{(1)\top}_{1} \Big) \Big\| \leq \delta_1^{-1}$
    \end{enumerate}
    \item[(A.3)] $|n_k - n\pi_k| = o\big(n^{1/2-a_1}\big)\; \text{for} \; k=1,2$.
    \item[(A.4)]\begin{enumerate}
    \item[(i)]$\dfrac{\lambda_n}{\sqrt{n}} \geq \delta_1^{-1} \max\left\{\sqrt{\log (np)}, p_0n^{a_1}\sqrt{\log{(np_0)}}\right\}$
    \item[(ii)]
    $\dfrac{\lambda_n}{n}n^{3a_1} p_0^2 = o(1)$.
    \item[(iii)]
    $\log p \leq \delta_1 \big(n^{{\kappa/3}} (\log(n))^{-1}\big).$
    \end{enumerate}
\end{enumerate}
Condition (A.1) is the strong irrepresentable condition which has appeared routinely in the literature of Lasso (see, e.g., \cite{Zhao2006}, \cite{Wainwright2009}, \citet{lahiri2021} and \cite{mai2012} among others). This condition is sufficient and almost necessary for 
identifying the relevant set of covariates in Lasso. One can drop this irrepresentable condition if one considers weighted $\ell_1$-penalty like the adaptive Lasso (\cite{zou2006}) or UniLasso (\cite{chatterjee2025}). However, we stick to the usual Lasso, since the adaptive Lasso and UniLasso estimators generally require a two-stage approach and hence are computationally costlier.
Condition (A.2)(i) is on the weighted design matrix $\bm{L}_{11, n}$ corresponding to the discriminative set under $H_1^\prime$. This type of condition is quite common in the literature of logistic regression and guarantees that the underlying MLE exists uniquely under $H_1^\prime$. On the other hand, condition (A.2)(ii) is required to ensure that the distribution of the estimator of individual regression coefficients under $H_1^\prime$ is non-degenerate. Assumption (A.2)(iii) is a technical moment condition required to handle the Karush–Kuhn–Tucker (KKT) conditions corresponding to the Lasso estimator. 
Condition (A.3) is on how far the sample proportions can be from the actual prior probabilities of the underlying distributions. When $a_1 = 0$ in (A.2)(i), (A.3) indicates that the sample proportions cannot be different from the population proportions by the magnitude of the order of the error, which is $o(n^{-1/2})$. 
When the sample sizes $n_1$ and $n_2$ are close, $\pi_1 = \pi_2 = 1/2$ is enough to ensure (A.3). Finally, condition (A.4) specifies how the penalty parameter $\lambda_n$ should be chosen based on the intrinsic properties of the underlying testing problem. 
Note that while $p_0$, the sparsity index of the discriminative set $\Theta^\prime(p_0)$ under $H_1^\prime$, can be of order $o(n^{1/6})$, the actual dimension $p$ can be exponentially large compared to $n$. For example, when $\kappa = 2$, i.e., when the distribution $F$ is sub-Gaussian, then $\log p$ can grow like $o\big(n^{2/3}\big)$ up to a logarithmic factor of $n$, whereas if the distribution $F$ is sub-exponential (i.e., $\kappa = 1$), then $\log p$ can grow like $o\big(n^{1/3}\big)$. We would like to point out that if, in addition to the above regularity conditions, we also have 
$\underset{(p_0+1)\leq j\leq p}{\max}\|\big(\mathbb{E}\bm{L}_{21, n}\big)_{j \cdot}\|_1  \leq \delta_1^{-1},$
then $p_0$ can be made as large as $o(n^{1/2})$. 
We are now ready to state the result on consistency of the test function $\psi_{n, \alpha}^{(1)}$ for testing the hypothesis stated in \eqref{hyp_beta}.

\begin{theorem}\label{theo:power}
Suppose that the regularity conditions (A.1)--(A.4) hold 
for the alternative $H_1: (\mu_1^\top, \mu_2^\top)^\top \in \Theta(p_0)$. Then for any $\alpha \in (0, 1)$, we have
$$\mathbb{E}_{H_1}\big(\psi_{n, \alpha}^{(1)}\big) \rightarrow 1\;\; \text{as}\;\; n \rightarrow \infty,$$ provided $\underset{j \in \mathcal{A}_n(p_0)}{\min}|\beta_{j}|  \geq C_2\dfrac{\lambda_n}{n}n^{a_1}$ for some $C_2 > 1$.
\end{theorem}

Theorems \ref{theo:size} and \ref{theo:power} together imply that the test $\psi_{n, \alpha}^{(1)}$ that we constructed above 
is a valid and consistent test for testing the hypothesis stated in \eqref{hyp_beta} for high dimensions under sparsity of the discriminative set. Moreover, as discussed in Section \ref{sec:intro}, the test is computationally fast compared to existing tests in the literature. The \textit{beta-min} condition in Theorem \ref{theo:power} specifies the minimum magnitude of the active components of $\bm{\beta}$ required for the test $\psi_{n,\alpha}^{(1)}$ to be consistent. In the high-dimensional case, if we choose the penalty $\lambda_n$ such that $\lambda_n \sim \sqrt{n\log p}$ and $a_1 = 0$, then for the consistency of the test function $\psi_{n, \alpha}^{(1)}$, it is necessary to have 
\begin{align}\label{eqn:minseppsi}
\underset{j \in \mathcal{A}_n(p_0)}{\min}|\beta_{j}|  \geq C_2\sqrt{\log p/n}.
\end{align}
In the next subsection, we establish that the separation in equation (\ref{eqn:minseppsi}) is the best possible in the minimax sense for testing the hypothesis stated in \eqref{hyp_beta_local} when the underlying populations are Gaussian. 

\subsection{Minimax separation distance and optimality} \label{sec:optimality}

In this section, 
we study the minimax separation distance of the components of $\bm{\beta}$ from $0$ required for any level $\alpha$ test to achieve power arbitrarily close to $1$. We again concentrate on the alternative $H_1: \mu = (\mu_1^\top,\mu_2^\top)^\top \in \Theta(p_0)$ and consider the underlying populations to be Gaussian. It turns out that the resulting minimax separation distance is of the same order as the \textit{beta-min} condition mentioned in equation (\ref{eqn:minseppsi}). We denote the law of the samples by $\mathbbm{P}_{\mu}$ (or, by $\mathbbm{P}_{\bm \beta}$) for a particular choice of $\mu$ (or, $\bm \beta$) under $H_1$. 
To fix ideas, for a level $\alpha\in(0,1)$ and a Type II error probability $\delta \in(0,1)$, define the $\delta$-separation distance of a level $\alpha$ test $\psi_\alpha$ for testing the hypothesis stated in \eqref{hyp_beta_local} as follows:
\begin{align*}
\rho(\phi_\alpha, \delta) &= \inf\Big\{ \rho>0: \inf_{\mu\in \Theta^\prime(p_0): \min_{j: |\beta_j| \ge \rho}}  \mathbbm{P}_{\mu} (\phi_\alpha=1) \ge 1 - \delta  \Big\}\\
&= \inf\Big\{ \rho>0: \sup_{\mu\in \Theta^\prime(p_0): \min_{j: |\beta_j| \ge \rho}}  \mathbbm{P}_{\mu} (\phi_\alpha=0) \le \delta  \Big\}.
\end{align*}
Subsequently, define the $(\alpha, \delta)$-minimax separation distance for testing the hypothesis stated in \eqref{hyp_beta_local} to be

\begin{equation*}
    \rho^* \equiv \rho^*(\alpha, \delta) = \inf_{\phi_\alpha} \rho(\phi_\alpha, \delta).
\end{equation*}
This quantity measures the minimum magnitude of the components of $\bm{\beta}$ required for any level $\alpha$ test for testing the hypothesis stated in \eqref{hyp_beta_local}
to have power at least $1-\delta$. Under Gaussianity of the underlying populations, the following theorem gives a lower bound on $\rho^*$.

\begin{theorem} \label{theo:minimax}
Let $X_1, \dots, X_n \stackrel{iid}{\sim} N(\mu_1, \Sigma)$ and $Y_1,\dots,Y_n \stackrel{iid}{\sim} N(\mu_2,\Sigma)$, with the two samples independent.
Suppose that $(\alpha, \delta)\in (0, 1)^2$ such that $\alpha+\delta<1$, the minimum eigenvalue of $\Sigma$, $\Sigma_{\min} \ge p_0^{-1}$, $\| \Sigma \|_{L_1} \leq M$, and $p_0 \leq Mp^{1/4}$, for some constant $M>0$. Then, for sufficiently large $n$ and $p$, we have
\begin{equation*}
    \rho^* \ge c \sqrt{\log p/n},
\end{equation*}
for some positive constant $c$.
\end{theorem}

Theorem \ref{theo:minimax} shows that if $c$ is sufficiently small, then there does not exist any level $\alpha$ test which can reject the null hypothesis $H_0^\prime: \bm{\beta} = \bm{0}$ uniformly over the discriminative set $ \Theta^\prime(p_0) \cap \big\{{\min_j}|\beta_{j}|  \geq c\sqrt{\log p/n}\big\}$ with probability tending to~$1$. Clearly, the \textit{beta-min} condition stated in Theorem \ref{theo:power} implies that the lower bound on $\rho^*$ (obtained in Theorem \ref{theo:minimax}) is attainable by our proposed test $\psi_{n, \alpha}^{(1)}$ in the Gaussian setting and hence, our proposed test is minimax rate optimal for testing the hypothesis stated in \eqref{hyp_beta_local}.

\section{An extension to the multi-sample case} \label{sec:test stat multi}
We now extend our testing procedure for comparing high-dimensional two-sample means to the multi-sample setup. 
Similar to the two-sample setup, here we also transform the testing problem to a classification problem using multi-class logistic regression. This technique may be thought of as an alternative to high-dimensional ANOVA. More precisely, suppose that there are $K\geq 2$ populations with means $\mu_1,\ldots, \mu_K$ and we would like to test
\begin{equation}
H_{0K}: \mu_1=\mu_2=\dots=\mu_K \text{ vs. } H_{1K}: \mu_k \neq \mu_l \;\text{for some}\; k \neq l, 
\end{equation}
based on the $p$-dimensional observed samples $\big\{X^{(k)}_1,\dots,X^{(k)}_{n_k}\big\}_{k=1}^{K}$ \big(with $n = \sum_{k=1}^{K}n_k$\big). Throughout this section, we assume that $X^{(k)}_1,\dots,X^{(k)}_{n_k} \stackrel{{iid}}{\sim} F_k$ a non-degenerate probability distribution for $k = 1,\dots, K$.
Our aim here is to extend the testing procedure developed in Section \ref{sec:test stat} to more than two populations. 

We can utilize the multi-class LR to construct a valid test function for testing $H_{0}^K$ vs. $H_{1}^K$. Define a random vector $U$ such that $\mathbb{P}(U = X^{(k)}) = \pi_k$ for $k=1,2,\dots,K$ with $\sum_{k=1}^{K}\pi_k = 1$, where $\pi_1, \dots, \pi_K$ are prior probabilities corresponding to the $K$ populations. 
Next, define a $(K-1)$-dimensional random vector $\bm{y}=(y_1\dots,y_{K-1})^\top$ such that $y_{k}$ is 1 and the remaining elements of $\bm{y}$ are all $0$ when $U = X^{(k)}$ for $k=1,2,\dots,(K-1)$, and $\bm{y} = \bm{0}$ when $U = X^{(K)}$. Based on the observed samples, we can view the data as arising from a multi-class logistic regression setup, where $\bm{y}$ is the response vector with observed values $\{\bm{y}_1, \dots, \bm{y}_n\}$ and $U$ is the vector of covariates with $\{U_1, \dots, U_n\}$ being the corresponding observed values. Subsequently, the regression parameter vector $\bm\beta = ((\beta_{1, 0},\bm\beta_{1}^\top)^\top,\dots,\bm(\beta_{(K-1), 0},\bm\beta_{(K-1)}^\top)^\top)^\top$ is defined as the solution of
\begin{align*}
\mathbb{E}\left[ \left(y_{k} - p_k(\bm t_k|U)\right) (1,U^\top)^\top \right] = 0,
\end{align*}
where for $k \in \{1,\dots, (K-1)\}$,
\[p_k(\bm t_k|U) = \frac{ \pi_k\exp(t_{k, 0}+\bm t_{k}^\top U) }{\pi_K+\sum_{l=1}^{K-1}\pi_l\exp(t_{l, 0}+\bm t_{l}^\top U)}.\]
Now, similar to Proposition \ref{prop:transformation}, $H_{0}^K$ vs. $H_{1}^K$ can be can be expressed in terms of $\bm{\beta}$. The following proposition summarizes the~result.
\begin{prop}\label{prop:multitransformation} 
$\mu_1= \cdots = \mu_K$ if and only if $\bm{\beta} = \bm{0}$.
\end{prop}
The above proposition indicates that it is enough to test $H_{0K}^\prime: \bm{\beta} = \bm{0}$ vs. $H_{1K}^\prime: \bm{\beta} \neq \bm{0}$ to conduct the multi-sample test of means. Thus, we need to construct a suitable test statistic based on an estimator of $\bm{\beta}$ and then use it to define a test function.

The inability of the Lasso to achieve both VSC and $\sqrt{n}$-consistency simultaneously motivates us to consider a two-step estimator of the multi-class LR parameter $\bm{\beta}$. In the first step, we consider the Lasso estimator to identify the non-zero components. In the second step, we consider the multi-class post-Lasso LR estimator of the selected active set to build our testing procedure. 
To that end, define the multi-class logistic Lasso estimator $\tilde{\bm{\beta}}_n$ of $\bm{\beta}$ as follows:
\begin{align*}
\tilde{\bm{\beta}}_n =& \big(\tilde{\bm{\beta}}_{n,1}^\top, \dots, \tilde{\bm{\beta}}_{n, (K-1)}^\top\big)^\top = \Big(\Big(\tilde{\beta}_{n,1,0}, \tilde{\bm{\beta}}_{n,1}^{\top}\Big), \dots, \Big(\tilde{\beta}_{n,(K-1),0}, \tilde{\bm{\beta}}_{n,(K-1)}^{\top}\Big)\Big)^\top \nonumber \\
= & \underset{\Big(\big(t_{1,0}, \bm{t}_{1}^{\top}\big), \dots, \big(t_{(K-1),0}, \bm{t}_{(K-1)}^{\top}\big)\Big)^\top \in \mathbb{R}^{(p+1)(K-1)}}{\arg\min} \Bigg[-\sum_{i=1}^{n}\sum_{k=1}^{K-1} y_{ik}\big(t_{k,0}+\bm{t}_{k}^{(1)\top} U_i\big)\nonumber\\
+&\sum_{i=1}^n\log \Big(n_K+\sum_{k=1}^{K-1}n_k\exp\Big(t_{k,0}+\bm{t}_{k}^{\top} U_i\Big)\Big) +\lambda_n\sum_{k=1}^{K-1}\Big(|t_{k,0}|+\|\bm{t}_k\|_1\Big)\Bigg]. 
\end{align*}
For each class $k \in \{1,\dots, (K-1)\}$, we assume $\tilde{\bm{\beta}}_{n, k} = \big(\tilde{\beta}_{n, k,0}, \dots, \tilde{\beta}_{n, k,p}\big)^\top$ and define the set of selected covariates as $\tilde{\mathcal{A}}_{nk}= \{l: \tilde{\beta}_{n, k,l}\neq 0\}$. Subsequently, we can define the multi-class post-Lasso logistic estimator $\hat{\bm{\beta}}_n^{\tilde{\mathcal{A}}_n}$ corresponding to the selected active set $\tilde{\mathcal{A}}_n = \cup_{k=1}^{K-1}\tilde{\mathcal{A}}_{nk}$ as
\begin{align*}
\bar{\bm{\beta}}_n^{\tilde{\mathcal{A}}_n} =&\Big(\big(\bar{\bm{\beta}}_n^{\tilde{\mathcal{A}}_{n1}}\big)^\top, \dots, \big(\bar{\bm{\beta}}_n^{\tilde{\mathcal{A}}_{n(K-1)}}\big)^\top\Big)^\top 
= \underset{\Big(\bm{t}_1^\top, \dots, \bm{t}_{K-1}^\top\Big)^\top \in \mathbb{R}^{\sum_{k=1}^{K-1}|\tilde{\mathcal{A}}_{nk}|}}{\arg\min}\nonumber \\ &\Bigg[-\sum_{i=1}^{n}\sum_{k=1}^{K-1} y_{ik}\bm{t}_{k}^\top \dot{U}_i^{\tilde{\mathcal{A}}_{nk}}+\sum_{i=1}^n\log \Big(n_K+\sum_{k=1}^{K-1}n_k\exp\Big(\bm{t}_{k}^\top \dot{U}_i^{\tilde{\mathcal{A}}_{nk}}\Big)\Big)\Bigg],
\end{align*}
where $\dot{U}_i = \big(1, U_i^\top\big)^\top$ for $i \in \{1,\dots, n\}$. We can define the test statistic for testing $H_{0K}^\prime$ vs. $H_{1K}^\prime$ as 
\begin{align*}
\bm{T}_{n}^{(2)} =
\sqrt{n}\begin{pmatrix}  \bm{\bar\beta}^{\tilde{\mathcal{A}}_n}_n\\ \bm{0} \end{pmatrix} + b_2(p, n)\bm{Z}_{2},
\end{align*}
where $\bm{Z}_2 \sim N\big(\bm{0}, I_{(p+1)(K-1)}\big)$ and $b_2(p, n)$ is a positive scaling factor (tending to zero) for $\bm{Z}_2$ to accommodate the increasing $p$ scenario. 
For any $\alpha \in (0, 1)$, define the quantity $m_{n,p,\alpha}^{(2)}> 0$ such that $\mathbb{P}\left( \underset{j}{\max}\; |Z_{2j}| \leq \frac{m_{n,p,\alpha}^{(2)}}{b_2(p, n)} \right) = \alpha$. At level $\alpha$, we define the test function for testing $H_{0K}^\prime: \bm{\beta} = \bm{0}$ vs. $H_{1K}^\prime: \bm{\beta} \neq \bm{0}$ as follows
$$\psi_{n, \alpha}^{(2)} = \mathbbm{I} \left \{ \underset{j}{\max}\; |T_{nj}^{(2)}| > m_{n,p,1-\alpha}^{(2)} 
\right \}.$$
We now state some assumptions that we require for size and power analyses of the test $\psi_{n, \alpha}^{(2)}$.
Similar to the two-sample case, here also throughout we assume that the distribution of $U$ is sub-Weibull, or in other words the components of $U$ are uniformly sub-Weibull, i.e., for some constant $C \in (0, \infty)$ and $\kappa \in (0, 2]$,
 $\underset{1\leq i \leq n}{\max}\underset{1\leq j \leq p}{\max}\|U_{ij}\|_{\zeta_\kappa} < C,$
where  $\|\cdot\|_{\zeta_{\kappa}}$ is the  norm as defined before.
We are going to analyze the power of the test $\psi_{n, \alpha}^{(2)}$ over the~alternative 
$$H_{1K}: \mu^{(K)}= (\mu_1^\top,\dots, \mu_{K}^\top) \in \Theta_K(p_0) = \{\mu^{(K)} \in \mathbf{R}^{pK}: \|\bm{\beta}\|_0 = (p_0+K-1)\}.$$ Therefore, for the power analysis, we essentially focus on the testing problem
$$H_{0K}^\prime: \bm{\beta} = 0 \text{ vs. } H_{1K}^\prime: \bm{\beta}\in \Theta_K^\top(p_0),$$ where $\Theta_K^\top(p_0) = \{\bm{\beta} \in \mathbf{R}^{(p+1)(K-1)}: \|\bm{\beta}\|_0 = (p_0+K-1)\}$ is the discriminative set under the alternative with sparsity index $(p_0+K-1)$. Without loss of generality, assume that under $H_{1K}^\prime$, $\mathcal{A}_{nk} =\{l: \beta_{k, l}\neq 0\}$ 
$\text{for all}\; k=1,\dots,(K-1)$ with $\sum_{k=1}^{K-1}|\mathcal{A}_{nk}| = p_0 +K -1$ and hence, $\mathcal{A}_n = \cup_{k=1}^{K-1}\mathcal{A}_{nk}$ is the set of active indices in the discriminative set $\Theta_K^\top(p_0)$. For any $k\in \{1,\dots, (K-1)\}$, define $\dot{U}_i^{(k)} = \dot{U}_i^{\mathcal{A}_{nk}}$, $\dot{U}_i^{(k^c)} = \dot{U}_i^{\mathcal{A}_{nk}^c}$ and $\bm{\beta}_k^{(1)} = \bm{\beta}^{\mathcal{A}_{nk}}$. Moreover, for any $k, l\in \{1,\dots, (K-1)\}$, we define
\begin{align*}
& w_i^{(k,l)} = \begin{cases}   \dfrac{\pi_k(\pi_K+\tilde{\Sigma}_{i,-k})\exp\left(\dot{U}_i^{(k)\top}\bm{\beta}_k^\on\right)}{(\pi_K+\tilde{\Sigma}_i)^2} \;\;\text{ if } k=l,\\
- \dfrac{\pi_k \pi_{l}\exp\left(\dot{U}_i^{(k)\top}\bm{\beta}_k^\on\right)\exp\left(\dot{U}_i^{(l)\top}\bm{\beta}_{l}^\on\right)}{(\pi_K+\tilde{\Sigma}_i)^2} \;\; \text{ if } k\neq l,   \end{cases}
\end{align*}
where \begin{equation*}
\tilde{\Sigma}_i = \sum_{l=1}^{K-1} \pi_l\exp\left(\dot{U}_i^{(l)\top}\bm{\beta}^\on_l\right) \text{ and }
\tilde{\Sigma}_{i, -k} = \underset{l\neq k}{\sum_{l=1 }^{K-1}} \pi_l\exp\left(\dot{U}_i^{(l)\top}\bm{\beta}^\on_l\right).    
\end{equation*}
Further, define the $(p+1)(K-1) \times (p+1)(K-1)$ symmetric matrix 
\begin{align*}
\bm{L}_{n} = \begin{bmatrix}
\bm{L}_{11,n}\;\;\;\bm{L}_{12,n}\\
\bm{L}_{21,n}\;\;\; \bm{L}_{22,n}
\end{bmatrix},
\end{align*}
where $\bm{L}_{11,n}$ is the $(p_0+K-1)\times (p_0+K-1)$ block matrix with $(l,l')$th block being $\bm{L}_{11,n}^{(l,l')} = \frac{1}{n} \sn w_i^{(l,l')}\dot{U}_i^{(l)} \dot{U}_i^{(l')\top}$ and $\bm{L}_{21,n}$ is the $(p(K-1) -p_0)\times (p_0+K-1)$ block matrix with $(k,l)$th block being $\bm{L}_{21,n}^{(l,l')} = \frac{1}{n} \sn w_i^{(l,l')}\dot{U}_i^{(l^c)} \dot{U}_i^{(l')\top}$ and $\bm{L}_{22,n}$ is the $(p(K-1) -p_0)\times (p(K-1) -p_0)$ block matrix with $(k,l)$th block being $\frac{1}{n} \sn w_i^{(k,l)}\dot{U}_i^{(k^c)} \dot{U}_i^{(l^c)\top}$. Note that under VSC, $\bm{L}_{11,n}$ is the variance of $\bar{\bm{\beta}}_n^{\tilde{\mathcal{A}}_n}$. Now, we are ready to state some regularity conditions, in addition to the sub-Weibull nature of $U$, required for the power analysis of $\psi_{n, \alpha}^{(2)}$. For some $\delta_2, a_2 \in (0, 1]$ and $ \tau_2 \in (0, 1)$, consider the following regularity conditions whenever $n \geq \delta_2^{-1}$. 
\begin{enumerate}
    \item[(C.1)] $\underset{1\leq k\leq (K-1)}{\max}\;\underset{l_k \in \mathcal{A}_{nk}^c}{\max}\left\{|\big(\mathbb{E}\bm{L}_{21, n}\big)_{l_k \cdot}^\top \big(\mathbb{E} \bm{L}_{11, n}\big)^{-1} sgn\big(\bm{\beta}^{\mathcal{A}_n}\big) |\right\} \leq 1-\tau_2$.
    \item[(C.2)] 
    \begin{enumerate}
    \item[(i)]$\|\big(\mathbb{E}\bm{L}_{11, n}\big)^{-1}\|_{\infty} \leq \delta_2^{-1} n^{a_2}$. 
    \item[(ii)] 
    $\Big[\underset{1 \leq k \leq (K-1)}{\min}\;\underset{l_k \in \mathcal{A}_{nk}}{\min}\Big(\big(\mathbb{E}\big(\bm{L}_{11, n}\big)\big)^{-1}\Big)_{l_k l_k}
    \Big] \geq \delta_2$.
    \item[(iii)] $\underset{1\leq k\leq (K-1)}{\max}\; \underset{l_k \in \mathcal{A}_{nk}}{\max} \Big\| \mathbb{E}\Big(|\dot{U}_{1l_k}| \dot{U}^{(k)}_{1} \dot{U}^{(k)\top}_{1} \Big) \Big\| \leq \delta_2^{-1}.$
    \end{enumerate}
    \item[(C.3)] $|n_k - n\pi_k| = o\big(n^{1/2-a_2}\big),\; \text{for} \; k=1,2,\dots,(K-1)$.
    \item[(C.4)]\begin{enumerate}
    \item[(i)]$\dfrac{\lambda_n}{\sqrt{n}} \geq \delta_2^{-1} \max\left\{\sqrt{\log (np)}, p_0n^{a_2}\sqrt{\log{(np_0)}}\right\}.$
    \item[(ii)]
    $\dfrac{\lambda_n}{n}n^{3a_2} p_0^2 = o(1)$.
    \item[(iii)] 
    $\log p \leq \delta_2\big(n^{{\kappa/3}} (\log(n))^{-1}\big).$
    \end{enumerate}
\end{enumerate}

The regularity conditions stated above are in the same spirit as the regularity conditions (A.1)--(A.4) considered in the two-sample case. So, all discussions on the regularity conditions (A.1)--(A.4) in the two-sample case (including those on the growth of the original dimension $p$) and the active dimension $p_0$ with respect to the sample size $n$ are valid here as well. 
Now, we have the following theorem on the test $\psi_{n, \alpha}^{(2)}$ for testing $H_{0K}: \mu_1=\dots=\mu_K$ vs. $H_{1K}: \mu^{(K)}= (\mu_1^\top,\dots, \mu_{K}^\top) \in \Theta_K(p_0)$.

\begin{theorem} \label{theo:multi}
\begin{enumerate}
\item[(a)] Suppose that $|n_k - n\pi_k|= o\big(n^{1/2}\big),\;$  for  $k=1,\dots,(K-1),$ and $\dfrac{\lambda_n}{\sqrt{n}}>C_3 \max\left\{\sqrt{\log(np)},\frac{(\log(n))^{1/\kappa}(\log(np))^{1/\min\{1,\kappa\}}}{\sqrt{n}} \right\},$ for some constant $C_3 \geq~1$. Then, for any $\alpha \in (0, 1)$, we have 
$$\mathbb{E}_{H_{0K}}\big(\psi_{n, \alpha}^{(2)}\big) \rightarrow \alpha\;\; \text{as}\;\; n \rightarrow \infty.$$ 
\item[(b)] Suppose that the regularity conditions (C.1)--(C.4) hold. Then, 
we have  $$\mathbb{E}_{H_{1K}}\big(\psi_{n, \alpha}^{(2)}\big) \rightarrow 1\;\; \text{as}\;\; n \rightarrow \infty,$$ provided \label{mminb}$\underset{1\leq j\leq (K-1)}{\min}\;\underset{l_j \in \mathcal{A}_{nj}}{\min}|\beta_{j,l_j}| \geq C_4\dfrac{\lambda_n}{n}n^{a_2}$ with some $C_4 > 1$.
\end{enumerate}
\end{theorem}

\noindent
Theorem \ref{theo:multi} implies that the test $\psi_{n, \alpha}^{(2)}$ constructed based on the multi-class logistic regression classifier is a valid and consistent test for testing $H_{0K}: \mu_1=\mu_2=\dots=\mu_K$ vs. $H_{1K}: \mu_k \neq \mu_l \;\text{for some}\; k \neq l$, under sparsity of the discriminative set even when the dimension $p$ can grow exponentially with the sample size $n$. As in the two-sample case, the multi-sample procedure requires neither normality nor bounded observations.

\section{Simulation studies} \label{sec:simulation}

We now conduct simulation studies to evaluate the performance of the proposed methods against several competing procedures. Our primary focus is on methods that remain computationally feasible in high-dimensional settings. We consider both two- and three-sample scenarios. 

Throughout this section, the dimension is fixed at $p = 5{,}000$ and the significance level at $\alpha = 0.05$. For sample sizes, we consider both balanced and unbalanced settings: $(75,75)$ and $(50,100)$ in the two-sample case, and $(75,75,75)$ and $(50,75,100)$ in the three-sample case. All results are based on $3000$ Monte Carlo replications. The implementation is available on \href{https://github.com/dsayan97/fast-test.git}{GitHub}.

\subsection{Data-driven choice of the penalty parameter}

Our primary goal is to select the Lasso penalty parameter $\lambda_n$ that satisfies the VSC property. In theory, this requires $\lambda_n$ to scale as $C \sqrt{\log p / n}$ for some constant $C>0$ (the scaling factor $n$ differs from our assumptions due to the different objective function used in the \texttt{glmnet} function available in R). However, determining an appropriate value of $C$ in practice is quite challenging as it can vary substantially for different datasets. For example, in our proposed method D3, we used $C = 0.8$, but this choice is not universally applicable. The standard cross-validation approach selects the Lasso penalty parameter to optimize prediction performance and does not necessarily provide the sparsity required for VSC (see \cite{Chak26}). 

We instead calibrate the penalty parameter under the null hypothesis, with the aim of favoring the selection of no variables when there is no association between the covariates and the response. To approximate this setting, we randomly permute the group labels, thereby breaking the association between the covariates and the response, while preserving the observed covariate structure and group sizes. For each permuted dataset, we calculate the threshold penalty at which the active set produced by Lasso first becomes empty. This threshold, usually denoted by $\lambda_{\max}$, can be obtained directly without fitting the Lasso over a grid of penalty values (see, e.g., \citet[Section~2.5]{glmnet}). Calibrating the penalty parameter under these permutations is intended to promote VSC under the null model. We repeat this calculation for $500$ independent random permutations and use the empirical $99.5$th percentile of the resulting threshold values as the penalty parameter. This choice favors sparsity under the permutation-based null without unduly compromising the selection of informative variables in the original data.

After selecting the penalty, we proceed similarly to the method D3. This data-adaptive procedure is called D3op. Unlike conventional Lasso implementations that rely on fitting models for a sequence of candidate penalties, D3op computes the permutation-specific maximal penalties directly, making the proposed calibration computationally quite efficient.

\subsection{Two-sample setting}
We compare D3 and D3op with several existing methods introduced earlier. The CQ test is implemented using the \texttt{PEtests} R package \citep{petest}. Implementations of XY \citep{Xue2020}, OP \citep{huang2015projection}, and ES \citep{Kong2022} are based on code provided by the respective authors. In view of computational constraints (also see Figure \ref{fig:test}), the XY test is implemented with $1{,}000$ bootstrap resamples (instead of the recommended $10{,}000$). For the ES test, we use the non-studentized version of the infinity-norm statistic to reduce the computational cost.

Data are generated as follows. Let
\[
X_i = \mu_1 + (x_{i,1}, \ldots, x_{i,p})^\top, \quad 
Y_j = \mu_2 + (y_{j,1}, \ldots, y_{j,p})^\top,
\]
where $\{x_{i,k}\}$ and $\{y_{j,k}\}$ follow stationary autoregressive processes
\[
x_{i,k+1} = a_i x_{i,k} + \xi_k, \quad 
y_{j,k+1} = b_j y_{j,k} + \eta_k,
\]
with $\xi_k, \eta_k \overset{\text{iid}}{\sim} N(0,1)$ and $a_i, b_j \overset{\text{iid}}{\sim} \text{Unif}(0,0.95)$. We consider two models:
\begin{itemize}
\item[(i)] homoscedastic case (both populations have identical marginal variances),
\item[(ii)] heteroscedastic case (even-indexed components are multiplied by two, leading to unequal variances across populations).
\end{itemize}

For size analysis, we set $\mu_1 = \mu_2 = 0$. For power analysis, we fix $\mu_1 = 0$ and define $\mu_2$ as follows:
\[
\mu = 2\delta \sqrt{\frac{\log p}{n}} \text{ and } 
\mu_2 = \big( \mu \cdot \mathbf{1}_{\lceil 0.2\sqrt{p} \rceil}^\top, \, 0_{p - \lceil 0.2\sqrt{p} \rceil}^\top \big)^\top,
\]
with $n=n_1+n_2$. The signal strength parameter is varied over $\delta\in \{0,1,2,3\}.$

\begin{table}[t]
\centering
\caption{Empirical size and power for the two-sample setting, together with average computation times in seconds and their standard errors in parentheses.}
\label{tab:sim_2sample}
\scriptsize
\begin{tabular}{llcccccc}
\toprule
$(n_1,n_2)$ & Setting & D3 & D3op & CQ & XY & OP & ES \\
\midrule
\multicolumn{8}{c}{Model 1: Empirical size ($\delta=0$)} \\ \midrule
$(75,75)$   &  & 0.053 & 0.056 & 0.052 & 0.024 & 0.049 & 0.045 \\
$(50,100)$  &  & 0.043 & 0.046 & 0.055 & 0.020 & 0.049 & 0.053 \\
\midrule
\multicolumn{8}{c}{Model 1: Empirical power} \\ \midrule
$(75,75)$   & $\delta=1$ & 0.112 & 0.101 & 0.105 & 0.077 & 0.069 & 0.119 \\
            & $\delta=2$ & 0.972 & 0.969 & 0.420 & 0.940 & 0.383 & 0.965 \\
            & $\delta=3$ & 1.000 & 1.000 & 0.915 & 1.000 & 0.926 & 1.000 \\
$(50,100)$  & $\delta=1$ & 0.073 & 0.099 & 0.100 & 0.061 & 0.065 & 0.116   \\
            & $\delta=2$ & 0.847 & 0.930 & 0.378 & 0.864 & 0.335 & 0.918   \\
            & $\delta=3$ & 1.000 & 1.000 & 0.866 & 1.000 & 0.869 & 1.000  \\
\midrule
\multicolumn{8}{c}{Model 2: Empirical size ($\delta=0$)} \\ \midrule
$(75,75)$   &  & 0.050 & 0.058 & 0.052 & 0.029 & 0.045 & 0.048 \\
$(50,100)$  &  & 0.050 & 0.058 & 0.043 & 0.024 & 0.044 & 0.048 \\
\midrule
\multicolumn{8}{c}{Model 2: Empirical power} \\ \midrule
$(75,75)$   & $\delta=1$ & 0.097 & 0.100 & 0.067 & 0.030 & 0.056 & 0.052 \\
            & $\delta=2$ & 0.952 & 0.950 & 0.148 & 0.074 & 0.207 & 0.108 \\
            & $\delta=3$ & 1.000 & 1.000 & 0.364 & 0.441 & 0.688 & 0.516 \\
$(50,100)$  & $\delta=1$ & 0.059 & 0.078 & 0.068 & 0.021 & 0.056 & 0.046 \\
            & $\delta=2$ & 0.771 & 0.876 & 0.153 & 0.057 & 0.200 & 0.096 \\
            & $\delta=3$ & 1.000 & 1.000 & 0.328 & 0.318 & 0.593 & 0.428 \\
\midrule
\multicolumn{8}{c}{Average computation time in seconds over 50 iterations} \\ \midrule
$(75,75)$   &  & 0.032 (0.008) & 0.301 (0.023) & 0.114 (0.013) & 1.461 (0.103) & 1.157 (0.045) & 4.143 (0.167) \\
$(50,100)$  &  & 0.027 (0.004) & 0.289 (0.008) & 0.111 (0.001) & 1.291 (0.076) & 1.101 (0.026) & 1.912 (0.098)  \\
\bottomrule
\end{tabular}
\end{table}
Table~\ref{tab:sim_2sample} reports the empirical size, power, and computation time of the two-sample procedures. Under the null hypothesis, D3 and D3op are well calibrated across covariance models and sample-size configurations, with rejection probabilities between $0.043$ and $0.058$. CQ, OP, and ES also remain close to the nominal level, whereas XY is conservative, with empirical sizes between $0.020$ and $0.029$.

Under the homoscedastic setting of Model 1, all procedures have relatively low power at the weakest signal level. The proposed procedures become competitive as the signal increases. At $\delta=2$, D3 and D3op attain powers of $0.972$ and $0.969$ in the balanced design, comparable to ES at $0.965$ and XY at $0.940$. Under sample-size imbalance, D3op has the highest power, $0.930$, followed by ES at $0.918$. At $\delta=3$, D3, D3op, XY, and ES attain power one, while CQ and OP remain less powerful.

The distinction between the procedures is more pronounced under the heteroscedastic setting of Model 2. At $\delta=2$, the powers of D3 and D3op are $0.952$ and $0.950$ in the balanced design and $0.771$ and $0.876$ in the unbalanced design; none of the competing procedures exceeds $0.207$ in these configurations. At $\delta=3$, both proposed procedures attain power one, whereas the strongest competing procedure, OP, has powers of $0.688$ and $0.593$ in the balanced and unbalanced designs, respectively. D3op is particularly advantageous over D3 under variance heterogeneity combined with sample-size imbalance, although both D3 and D3op perform similarly otherwise.

D3 is also the fastest procedure, requiring approximately $0.03$ seconds per replication. D3op requires approximately $0.30$ seconds, reflecting the cost of its data-adaptive penalty calibration, but remains substantially faster than XY, OP, and ES. Thus, D3 provides the most computationally efficient implementation, while D3op can yield an appreciable power gain in the unbalanced and heteroscedastic settings.

\subsection{Three-sample setting}

We extend the numerical analysis to the multi-sample case with three populations. We compare D3 and D3op with several existing methods. Implementations of CS \citep{Chakraborty2023} and HDT \citep{Li2023} are based on codes provided by the respective authors. Under the null hypothesis, we set $\mu_1 = \mu_2 = \mu_3 = 0$. Under the alternative, we take $\mu_1 = 0$, $\mu_2$ as defined above, and construct $\mu_3$ by randomly permuting the coordinates of $\mu_2$. This preserves the signal strength while introducing heterogeneity across populations. The covariance structure is kept identical across all three groups.
All the other settings, including data generation, sample sizes, and number of replications, remain the same as in the two-sample case. 

\begin{table}[t]
\centering
\caption{Empirical size and power for the three-sample setting, together with average computation times in seconds and their standard errors in parentheses.}
\label{tab:sim_3sample}
\scriptsize
\begin{tabular}{llcccc}
\toprule
$(n_1,n_2,n_3)$ & Setting & D3 & D3op & CS & HDT \\
\midrule
\multicolumn{6}{c}{Model 1: Empirical size ($\delta=0$)} \\ \midrule
$(75,75,75)$   &  & 0.055 & 0.050 & 0.020 & 0.049 \\
$(50,75,100)$  &  & 0.051 & 0.055 & 0.019 & 0.108 \\
\midrule
\multicolumn{6}{c}{Model 1: Empirical power} \\ \midrule
$(75,75,75)$   & $\delta=1$ & 0.094 & 0.101 & 0.067 & 0.090 \\
               & $\delta=2$ & 0.980 & 0.988 & 0.959 & 0.336 \\
               & $\delta=3$ & 1.000 & 1.000 & 1.000 & 0.859 \\
$(50,75,100)$  & $\delta=1$ & 0.080 & 0.073 & 0.042 & 0.134 \\
               & $\delta=2$ & 0.904 & 0.890 & 0.770 & 0.325 \\
               & $\delta=3$ & 1.000 & 1.000 & 1.000 & 0.694 \\
\midrule
\multicolumn{6}{c}{Model 2: Empirical size ($\delta=0$)} \\ \midrule
$(75,75,75)$   &  & 0.056 & 0.052 & 0.026 & 0.056 \\
$(50,75,100)$  &  & 0.054 & 0.053 & 0.018 & 0.100 \\
\midrule
\multicolumn{6}{c}{Model 2: Empirical power} \\ \midrule
$(75,75,75)$   & $\delta=1$ & 0.072 & 0.081 & 0.027 & 0.058 \\
               & $\delta=2$ & 0.929 & 0.941 & 0.052 & 0.120 \\
               & $\delta=3$ & 1.000 & 1.000 & 0.352 & 0.299 \\
$(50,75,100)$  & $\delta=1$ & 0.072 & 0.078 & 0.022 & 0.110 \\
               & $\delta=2$ & 0.840 & 0.817 & 0.033 & 0.167 \\
               & $\delta=3$ & 1.000 & 1.000 & 0.186 & 0.281 \\
\midrule
\multicolumn{6}{c}{Average computation time in seconds over 50 iterations} \\ \midrule
$(75,75,75)$   &  &  0.043 (0.005) & 1.020 (0.102) & 4.307 (0.175) & 0.028 (0.005)  \\
$(50,75,100)$  &  &  0.036 (0.002) & 0.991 (0.081) & 4.053 (0.172) & 0.013 (0.001)  \\
\bottomrule
\end{tabular}
\end{table}

Table~\ref{tab:sim_3sample} presents the results for the three-sample setting. D3 and D3op maintain rejection probabilities close to $0.05$ under both covariance models and both sample-size configurations. CS is conservative, with empirical sizes between $0.018$ and $0.026$. HDT is well calibrated in the balanced design but exhibits substantial size inflation under sample-size imbalance, with empirical sizes of $0.108$ and $0.100$ under Models 1 and 2, respectively.

Under the homoscedastic setting of Model 1, D3, D3op, and CS perform strongly at the moderate and large signal levels. At $\delta=2$, D3 and D3op attain powers of $0.980$ and $0.988$ in the balanced design and $0.904$ and $0.890$ in the unbalanced design. CS is also competitive, with corresponding powers of $0.959$ and $0.770$, whereas HDT is considerably less powerful. At $\delta=3$, D3, D3op, and CS attain power one. The relatively large rejection probabilities of HDT under weak signals in the unbalanced design should be interpreted in light of its inflated null rejection probability.

Under the heteroscedastic setting of Model 2, the proposed procedures clearly outperform the competitors. At $\delta=2$, D3 and D3op have powers of $0.929$ and $0.941$ in the balanced design and $0.840$ and $0.817$ in the unbalanced design, whereas neither CS nor HDT exceeds $0.167$. At $\delta=3$, both proposed procedures attain power one, while the powers of CS and HDT remain below $0.36$. There is no uniform power ordering between D3 and D3op: D3op has a modest advantage in the balanced design, whereas D3 is slightly more powerful at the moderate signal level under sample-size imbalance.

HDT is the fastest procedure, requiring $0.013$--$0.028$ seconds per replication, followed closely by D3 at $0.036$--$0.043$ seconds. D3op requires approximately one second, but remains substantially faster than CS, whose average computation time exceeds four seconds. Overall, D3 offers the strongest combination of speed, size control, and power, while D3op provides comparable statistical performance with the benefit of data-adaptive penalty selection.

\section{Real-data analysis} \label{sec:real}

We apply the proposed multi-sample tests to two breast-cancer gene-expression datasets, GSE1456 and GSE7390, from the Gene Expression Omnibus. These datasets are available in the \texttt{GEOquery} package in \texttt{R} \citep{OMLc2}. In GSE1456 (see \cite{GSE1456c1}), samples are grouped by Elston tumor grade into Grade 1 ($n_1=28$), Grade 2 ($n_2=58$), and Grade 3 ($n_3=61$). In GSE7390 (see \cite{GSE7390c1}), the corresponding group sizes are $n_1=30$, $n_2=83$, and $n_3=83$. Both datasets contain $p=22{,}283$ gene-expression measurements. We compare D3, D3op, CS (with $1{,}000$ resamples), and HDT.

All four procedures yield p-values below $0.001$ for both datasets, providing strong evidence against equality of the three mean vectors. Because these full-sample results do not distinguish between the procedures, we further examine their calibration and rejection stability using the resampling experiments reported in Table~\ref{tab:realdata}.

\begin{table}
\centering
\caption{Rejection proportions (with standard errors in parentheses) and computation times for different methods.}
\label{tab:realdata}
\scriptsize
\begin{tabular}{lcccc}
\toprule
 & D3 & D3op & CS & HDT \\
\midrule
\multicolumn{5}{c}{Rejection proportions over 500 random permutations} \\ \midrule
GSE1456 & 0.052 (0.010) & 0.054 (0.010) & 0.048 (0.009) & 1.000 (0.000) \\
GSE7390 & 0.060 (0.010) & 0.048 (0.009) & 0.046 (0.009) & 1.000 (0.000) \\
\midrule
\multicolumn{5}{c}{Rejection proportions over 500 bootstrap datasets} \\ \midrule
GSE1456 & 1.000 (0.000) & 1.000 (0.000) & 1.000 (0.000) & 1.000 (0.000) \\
GSE7390 & 1.000 (0.000) & 1.000 (0.000) & 1.000 (0.000) & 1.000 (0.000) \\
\midrule
\multicolumn{5}{c}{Average computation time in seconds over 50 iterations} \\ \midrule
GSE1456 & 0.172 (0.012) & 3.268 (0.115) & 10.849 (0.495) & 0.026 (0.001) \\
GSE7390 & 0.219 (0.017) & 4.433 (0.147) & 13.788 (0.356) & 0.027 (0.004) \\
\bottomrule
\end{tabular}
\end{table}

First, we generate $500$ random permutations of the group labels. This produces a permutation-null diagnostic that preserves the pooled covariate structure while removing its association with the observed labels. D3, D3op, and CS have rejection proportions close to $0.05$ for both datasets. In contrast, HDT rejects every permuted dataset. Its full-sample and resampling rejection results should therefore not be interpreted as evidence of high power in these applications. The permutation experiment is intended as a finite-sample calibration diagnostic rather than as an estimate of size over all heterogeneous distributions satisfying equality of means.

Second, we generate 500 bootstrap datasets by resampling independently within each group. D3, D3op, and CS reject all bootstrap samples from both datasets, indicating that the observed group separation is highly stable under within-group resampling. Because this experiment is saturated at the full sample sizes, we also consider smaller effective samples. For each retained proportion $l\in[0.25,1]$, we draw 300 subsamples independently across the three groups and report the resulting rejection proportions in Figure~\ref{fig:subsampling}.

\begin{figure}
    \centering
    \includegraphics[width=0.9\textwidth]{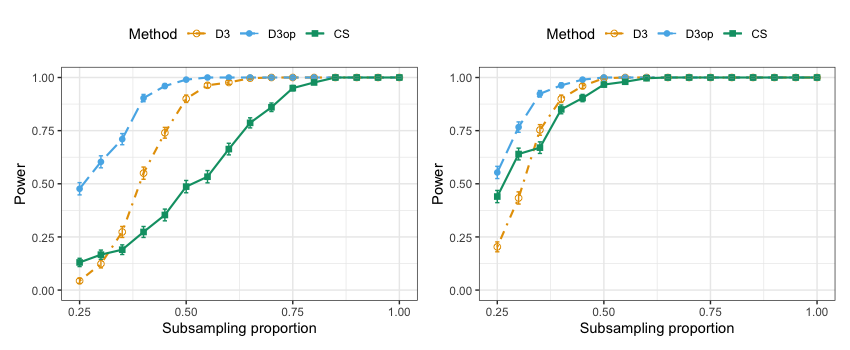}
    \caption{Rejection proportions over $300$ subsamples as a function of the proportion of observations retained from each group; vertical bars show Monte Carlo standard errors. The left panel corresponds to GSE1456 and the right panel to GSE7390.}
\label{fig:subsampling}
\end{figure} 

For GSE1456, D3op has the highest rejection proportion in the smaller-sample regime, particularly when no more than half of the observations are retained. D3 improves rapidly as the retained proportion increases, while CS requires larger subsamples to attain comparable rejection rates. For GSE7390, the differences are smaller: D3op performs best at the lowest retained proportions, but D3 becomes comparable once approximately $40\%$--$50\%$ of the observations are retained. All three procedures approach a rejection proportion of one as the retained sample size increases. Thus, the benefit of the data-adaptive penalty is most visible when the effective sample size is small, although its magnitude varies between datasets.

The computation times in Table~\ref{tab:realdata} reinforce the scalability of the proposed procedures. Among the methods exhibiting satisfactory permutation calibration, D3 is the fastest, requiring $0.172$ seconds for GSE1456 and $0.219$ seconds for GSE7390. D3op requires $3.268$ and $4.433$ seconds, respectively, but remains substantially faster than CS, which requires $10.849$ and $13.788$ seconds despite using only $1{,}000$ resamples. HDT is faster than all three procedures, but its rejection of every permuted dataset makes this computational advantage uninformative here. Overall, D3 provides the most computationally efficient analysis, whereas D3op can improve sensitivity when only a relatively small sample is available.

\section*{Funding}
The second author is partially supported by ANRF--MATRICS grant no. RD/0126-ANRF000-031. The third author is partially supported by the SERB project no. MTR/2023/000884.

\bibliographystyle{plainnat}
\bibliography{reference}

\section*{Appendix A: Preliminary lemmas} \label{sec: lemma}

\begin{lem}\label{kcvl2}
    Suppose $X_1, X_2,\dots,X_n$ are independent mean zero random vectors in $\mathcal{R}^p,$ for any  $p\geq 1$ such that $\|X_{ij}\|_{\zeta_\kappa} < \infty$ for all $1\leq i \leq n, 1\leq j \leq p$ and some $\kappa > 0.$ Additionally, $n^{-1} \sn \mathbb{E} (X_{ij}^2) < \infty$ for all $1\leq j \leq p,$ then there exists a universal constant $C$ that depends only on $\alpha$ such that,
    \begin{multline*}
        \mathbbm{P}\left(\underset{1\leq j \leq p}{\max} \big| \sn X_{ij} \big| \geq C \big( \sqrt{n(t+\log{p})} + (\log{n})^{1/\kappa} (t+\log{p})^{1/\kappa^*} \big) \right)\\ \leq 3 \exp{(-t)}\; \text{ for all } t \geq 0,
    \end{multline*}
    where $\kappa^* = \min\{1,\kappa\}$.
\end{lem}
\begin{proof}
This lemma is a direct consequence of Theorem 3.4 in \cite{kuchi22}.
\end{proof}

\begin{lem} \label{kcvl3}
    If $X_1, X_2,\dots,X_n$ are random variables (possibly dependent) with $\|X_i\|_{\zeta_{\kappa_i}} < \infty$ for some $\kappa_i>0$ and for all $1\leq i \leq n,$ then
    \begin{equation*}
        \left\| \prod_{i=1}^n X_i \right\|_{\zeta_\lambda} \leq \prod_{i=1}^n \big\| X_i \big\|_{\zeta_{\kappa_i}}, \text{ where } \frac{1}{\lambda} := \sn \frac{1}{\kappa_i}.
    \end{equation*}
\end{lem}
\begin{proof}
This lemma is proved in Proposition D.2 in \cite{kuchi22}.
\end{proof}

\begin{lem} \label{vscnew}
(VSC of Logistic Lasso) Define $\tilde{\bm{u}}_n = \sqrt{n} (\tilde{\bm{\beta}}_n - \bm{\beta}_n)$ where $\tilde{\bm{\beta}}_n$ is the Logistic Lasso estimator as defined in Section~2 of the main manuscript. Also assume that $\mathcal{A}_n= \{j:\beta_{j}\neq 0\}= \{0, \dots, p_0\}$ is the set of relevant covariates and $\tilde{\mathcal{A}}_n= \{j:\tilde{\beta}_{j}\neq 0\}$ is the estimator of it based on $\tilde{\bm{\beta}}_n$. Then under the assumptions (A.1)--(A.4), we have as $n\rightarrow \infty$, $$\mathbb{P}\Big(\tilde{\mathcal{A}}_n = \{0,\dots, p_0\}\; \text{and}\; \tilde{\bm{u}}_n = \big(\big(\tilde{\bm{u}}_n^{\tilde{\mathcal{A}}_n}\big)^\top, \bm{0}^\top\big)^\top\; \text{with}\; \|\tilde{\bm{u}}_n^{(1)\prime}\|_\infty \leq C\dfrac{\lambda_{n}}{\sqrt{n}} n^{a}\Big) \rightarrow 1,$$ for some constant $C\geq 1$.
\end{lem}
\begin{proof} The Logistic Lasso estimator is defined as 
\begin{multline*}
    \tilde{\bm{\beta}}_n = \underset{(t_0, \bm{t}^\top)^\top \in \mathcal{R}^{(p+1)}}{\arg\min} \Big[-\sum_{i=1}^n z_i (t_0 + W_i^\top \bm{t})
    + \sum_{i=1}^{n} \log (n_1+n_2\exp(t_0 + W_i^\top\bm{t})) +\lambda_n \sum_{j=0}^{p} |t_j|\Big]. 
\end{multline*} 
Define, $\tilde{\bm{u}}_n = \sqrt{n} (\tilde{\bm{\beta}}_n - \bm{\beta}_n)$. Then, we get (see, e.g., \citet[Section~4.1.3]{boyd2004} for change of variables)
\begin{multline}  \label{2}
\tilde{\bm{u}}_n 
= \underset{\bm{u}\in \mathcal{R}^{(p+1)}}{\arg\min} \left[-\sum_{i=1}^{n}z_i\dot{W}_i^\top\frac{\bm{u}}{\sqrt{n}}+ \sum_{i=1}^n \log \left(n_1+n_2\exp\left(\dot{W}_i^\top\left(\frac{\bm{u}}{\sqrt{n}} + \bm{\beta_n}\right)\right)\right)\right. \\+\left.\lambda_n\sum_{j=0}^{p}\left(\left|\frac{u_j}{\sqrt{n}} + \beta_{nj}\right|-|\beta_{nj}|\right)\right].
\end{multline}
For $j=1,\dots,p,$ \eqref{2} is equivalent to the KKT conditions
\begin{multline} \label{3}
    -\dfrac{1}{\sqrt{n}}\sum_{i=1}^{n}{z_{i}\dot{W}_{ij}+ \sum_{i=1}^{n} \dfrac{\dot{W}_{ij}}{\sqrt{n}}\dfrac{n_2\exp\left(\dot{W}_i^\top\left(\frac{\bm{u}}{\sqrt{n}} + \bm{\beta_n}\right)\right)}{n_1+n_2\exp\left(\dot{W}_i^\top\left(\frac{\bm{u}}{\sqrt{n}} + \bm{\beta_n}\right)\right)}}\\=-\dfrac{\lambda_{n}}{\sqrt{n}} sgn\left(\frac{u_j}{\sqrt{n}} + \beta_{nj}\right), \text{if } \frac{u_j}{\sqrt{n}} + \beta_{nj} \neq 0,
    \end{multline}
    \text{and  }
    \begin{multline} \label{4}
        -\frac{\lambda_{n}}{\sqrt{n}}\leq -\dfrac{1}{\sqrt{n}}\sum_{i=1}^{n}{z_{i}\dot{W}_{ij}+ \sum_{i=1}^{n} \dfrac{\dot{W}_{ij}}{\sqrt{n}}\dfrac{n_2\exp\left(\dot{W}_i^\top\left(\frac{\bm{u}}{\sqrt{n}} + \bm{\beta_n}\right)\right)}{n_1+n_2\exp\left(\dot{W}_i^\top\left(\frac{\bm{u}}{\sqrt{n}} + \bm{\beta_n}\right)\right)}}\\ \leq \dfrac{\lambda_{n}}{\sqrt{n}},\text{if } \frac{u_j}{\sqrt{n}} + \beta_{nj} = 0.
    \end{multline}

Note that, VSC holds if $\tilde{\mathcal{A}}_n = \mathcal{A}_n.$ Without loss of generality, let $\mathcal{A}_n = \{0,1,\dots,p_0\}$. Then VSC holds if $\tilde{\mathcal{A}}_n = \{0,1,\dots,p_0\}$. This is equivalent to $\Big(\frac{\tilde{u}_j}{\sqrt{n}} + \beta_{nj}\Big) \neq 0, \text{for } j=0,1,\dots,p_0$ and $\Big(\frac{\tilde{u}_j}{\sqrt{n}} + \beta_{nj}\Big) = 0, \text{ for } j=(p_0+1),\dots,p,$ where $\tilde{\bm{u}}_n = \sqrt{n} (\tilde{\bm{\beta}}_n - \bm{\beta}_n)$ is the solution of \eqref{3}. Thus, due to equation \eqref{3} and \eqref{4},
\begin{equation} \label{5}
    -\dfrac{1}{\sqrt{n}}\sum_{i=1}^{n}{z_{i}\dot{W}_{i}^{(1)}+ \sum_{i=1}^{n} \dfrac{\dot{W}_{i}^{(1)}}{\sqrt{n}}\dfrac{n_2\exp\left(\dot{W}_i^{(1)\prime}\left(\frac{\bm{u}_n^{(1)}}{\sqrt{n}} + \bm{\beta_n}^{(1)}\right)\right)}{n_1+n_2\exp\left(\dot{W}_i^{(1)\prime}\left(\frac{\bm{u}_n^{(1)}}{\sqrt{n}} + \bm{\beta_n}^{(1)}\right)\right)}}=-\dfrac{\lambda_{n}}{\sqrt{n}} \hat{\bm{s}}_n^{(1)},
    \end{equation}
    \text{and  }
    \begin{multline} \label{6}
        -\frac{\lambda_{n}}{\sqrt{n}}\leq -\dfrac{1}{\sqrt{n}}\sum_{i=1}^{n}{z_{i}\dot{W}_{ij}+ \sum_{i=1}^{n} \dfrac{\dot{W}_{ij}}{\sqrt{n}}\dfrac{n_2\exp\left(\dot{W}_i^{(1)\prime}\left(\frac{\bm{u}_n^{(1)}}{\sqrt{n}} + \bm{\beta_n}^{(1)}\right)\right)}{n_1+n_2\exp\left(\dot{W}_i^{(1)\prime}\left(\frac{\bm{u}_n^{(1)}}{\sqrt{n}} + \bm{\beta_n}^{(1)}\right)\right)}}\\ \leq \dfrac{\lambda_{n}}{\sqrt{n}},\;j=(p_0+1),\dots,p,
    \end{multline}
    where $\hat{s}_{nj} = sgn\left(\frac{u_j}{\sqrt{n}} + \beta_{nj}\right),\;j=0,\dots,p_0.$ Now, applying Taylor's theorem on \eqref{5} we have,
    \begin{multline*}
        -\sum_{i=1}^{n} \left( z_{i} - \dfrac{n_2\exp\left(\dot{W}_i^{(1)\prime}\bm{\beta_n}^{(1)}\right)}{n_1+n_2\exp\left(\dot{W}_i^{(1)\prime} \bm{\beta_n}^{(1)}\right)}\right) \dfrac{\dot{W}_{i}^{(1)}}{\sqrt{n}}\\
        + \left[ \sn \dot{W}_i^{(1)}\dot{W}_i^{(1)\prime} \dfrac{n_1 n_2\exp\left(\dot{W}_i^{(1)\prime}\bm{\beta_n}^{(1)}\right)}{\left(n_1+n_2\exp\left(\dot{W}_i^{(1)\prime} \bm{\beta_n}^{(1)}\right)\right)^2} \right] \frac{\bm{u}_n^{(1)}}{\sqrt{n}}\\ + \frac{1}{2} \sn \frac{n_1 n_2\exp(\xi_i)(n_1-n_2\exp(\xi_i))}{(n_1+n_2\exp(\xi_i))^3} \left(\dot{W}_i^{(1)\prime}\frac{\bm{u}_n^{(1)}}{\sqrt{n}}\right)^2  \dfrac{\dot{W}_{i}^{(1)}}{\sqrt{n}} = -\frac{\lambda_n}{\sqrt{n}} \hat{\bm{s}}_n^{(1)},
    \end{multline*}
where $\xi_i$ is on the line joining $\dot{W}_i^{(1)\prime}\bm{\beta_n}^{(1)}$ and $\dot{W}_i^{(1)\prime}\left(\frac{\bm{u}_n^{(1)}}{\sqrt{n}} + \bm{\beta_n}^{(1)}\right),\; i=1,\dots n.$
Solving for $\bm{u}_{n}^\on,$ we get
\begin{equation} \label{7}
    \bm{u}_{n}^\on = \tilde{\bm{L}}_{11,n}^{-1} \left[ \bm\Lambda_n^\on - \bm\zeta_n^\on -\frac{\lambda_n}{\sqrt{n}} \hat{\bm{s}}_n^{(1)} \right] = f(\bm{u}_{n}^\on),
\end{equation}
where,
\begin{align*}
&\tilde{\bm{L}}_{11,n} = n^{-1}\sn \dot{W}_i^{(1)}\dot{W}_i^{(1)\prime} \hat{p}(\bm\beta|\dot{W}_i^\on) (1-\hat{p}(\bm\beta|\dot{W}_i^\on)),\\
   &\bm\Lambda_n^\on = \frac{1}{\sqrt{n}} \sum_{i=1}^{n} \left( z_{i} - \hat{p}(\bm\beta|\dot{W}_i^\on) \right) \dot{W}_{i}^{(1)},\\
    &\bm\zeta_n^\on = \frac{1}{2n^{3/2}} \sn h(\xi_i) \left(\dot{W}_i^{(1)\prime}\bm{u_n}^{(1)}\right)^2 \dot{W}_{i}^{(1)},\\
    \text{and } &h(\xi_i) = \frac{n_1 n_2\exp(\xi_i)(n_1-n_2\exp(\xi_i))}{(n_1+n_2\exp(\xi_i))^3}.
\end{align*}
Next, We will determine stochastic bounds for each of the quantities separately. First, we determine the stochastic bound for the operator norm of the random matrix $\tilde{\bm{L}}_{11,n}^{-1}$. Using triangle inequality, we can write
\begin{equation} \label{7.1}
    \begin{split}
        \|\tilde{\bm{L}}_{11,n}^{-1}\|_\infty & \leq \|\tilde{\bm{L}}_{11,n}^{-1}  - (\mathbb{E}\tilde{\bm{L}}_{11,n})^{-1} \|_\infty + \| (\mathbb{E}\bm{L}_{11,n})^{-1} \|_\infty \\& \hspace{5cm}+ \|(\mathbb{E}\tilde{\bm{L}}_{11,n})^{-1} - (\mathbb{E}\bm{L}_{11,n})^{-1} \|_\infty  \\
        & \leq \| \tilde{\bm{L}}_{11,n}^{-1} \| _\infty\|\tilde{\bm{L}}_{11,n}  - (\mathbb{E}\tilde{\bm{L}}_{11,n}) \|_\infty \|(\mathbb{E}\tilde{\bm{L}}_{11,n})^{-1}\|_\infty + \| (\mathbb{E}\bm{L}_{11,n})^{-1} \|_\infty \\& \qquad \qquad+ \| (\mathbb{E}\tilde{\bm{L}}_{11,n})^{-1} \|_\infty \|(\mathbb{E}\tilde{\bm{L}}_{11,n}) - (\mathbb{E}\bm{L}_{11,n}) \|_\infty \| (\mathbb{E}\bm{L}_{11,n})^{-1} \|_\infty
    \end{split}
\end{equation}
From (A.2)(i), we have $\| (\mathbb{E}\bm{L}_{11,n})^{-1} \|_\infty \leq C n^{a_1}.$ To get the bound for \eqref{7.1}, we get with probability tending to one,
\begin{equation} \label{8.1}
\begin{split}
    \|\mathbb{E}\tilde{\bm{L}}_{11,n} - \mathbb{E}\bm{L}_{11,n}\|_\infty &\leq C p_0 \;\bigg|\frac{n_1}{n_2} - \frac{\pi_1}{\pi_2}\bigg|\; \max_{0\leq l,l'\leq p_0} \frac{1}{n} \sn \mathbb{E} \Big| \dot{W}_{il} \dot{W}_{il'}\Big|\\
    & = o(p_0 \;n^{-1/2-{a_1}}),
\end{split}
\end{equation}
where the second inequality follows from Taylor's theorem and the final bound follows from (A.3). 
From \eqref{8.1}, using Lemma 4.2 in \cite{lahiri2003} and assumption (A.2)(i), we can say $(\mathbb{E}\tilde{\bm{L}}_{11,n})^{-1}$ exists and with probability tending to one,
\begin{equation}\label{8.2}
    \begin{split}
    \|(\mathbb{E}\tilde{\bm{L}}_{11,n})^{-1}\|_\infty \leq Cn^{a_1}, \text{  and  }
    \|(\mathbb{E}\tilde{\bm{L}}_{11,n})^{-1} - (\mathbb{E}{\bm{L}}_{11,n})^{-1}\|_\infty \leq C \sqrt{p_0}\;n^{-1/2+{a_1}}.
\end{split}
\end{equation}
Furthermore, we can write
\begin{align*}
    \| \tilde{\bm{L}}_{11,n} - \mathbb{E}(\tilde{\bm{L}}_{11,n}) \|_\infty
    \leq p_0\; \underset{0\leq l, \leq p_0}{\max} |V_n^{(l,l')}| .
\end{align*}
where for $0\leq l, \leq p_0,$
\begin{multline*}
    V_n^{(l,l')} = \frac{1}{n} \sn \left\{ \dot{W}_{il} \dot{W}_{il'}\;  \hat{p}(\bm\beta|\dot{W}_i^\on) (1- \hat{p}(\bm\beta|\dot{W}_i^\on)) \right. \\ \left. - \mathbb{E} (\dot{W}_{il} \dot{W}_{il'}\;  \hat{p}(\bm\beta|\dot{W}_i^\on) (1- \hat{p}(\bm\beta|\dot{W}_i^\on))) \right\}.
\end{multline*}
Applying Lemma \ref{kcvl2} and Lemma \ref{kcvl3} together with assumption (A.4) we get with probability tending to one, 
\begin{align*}
    \underset{0\leq l,l' \leq p_0}{\max} |V_n^{(l,l')}| \leq C \sqrt{{\frac{\log(np_0)}{n}}}.
\end{align*}
Therefore, with probability tending to one,
\begin{equation} \label{8}
    \| \tilde{\bm{L}}_{11,n} - \mathbb{E}(\tilde{\bm{L}}_{11,n}) \|_\infty \leq C p_0 \sqrt{\frac{\log(np_0)}{n}}.
\end{equation}

\noindent By similar argument as in \eqref{8.2}, using assumption (A.2)(i) together with \eqref{8.2} and \eqref{8}, we get with probability tending to one,
\begin{equation} \label{9}
\begin{split}
    \| \tilde{\bm{L}}_{11,n}^{-1} \|_\infty \leq C n^{{a_1}}, \text{  and  }
    \| \tilde{\bm{L}}_{11,n}^{-1} -  (\mathbb{E}\tilde{\bm{L}}_{11,n})^{-1}\|_\infty \leq C n^{2a} p_0 \sqrt{\frac{\log(np_0)}{n}}.
\end{split}
\end{equation}
Towards obtaining a stochastic bound for $\bm\Lambda_n^{(1)}$, define 
\begin{align*}
    \bm\Lambda_n^{(2)} &= \frac{1}{\sqrt{n}}\sn\left( z_{i} - p(\bm\beta|\dot{W}_i^\on) \right) \dot{W}_{i}^\on,\\
    \label{9.5} \bm\Lambda_n^{(3)} &= \frac{1}{\sqrt{n}}\sn\left( \hat{p}(\bm\beta|\dot{W}_i^\on) - p(\bm\beta|\dot{W}_i^\on) \right) \dot{W}_{i}^\on.
\end{align*}
Applying Lemma \ref{kcvl2} together with the assumptions (A.4) and (A.3), with probability tending to one, we get
\begin{equation} \label{9.2}
    \| \bm\Lambda_n^{(2)} \|_\infty \leq C \sqrt{ \log{(np_0)}},
\end{equation}
and
\begin{equation} \label{9.25} 
\begin{split}
    \|\bm\Lambda_n^{(3)}\|_\infty &\leq C \bigg|\frac{n_1}{n_2} - \frac{\pi_1}{\pi_2}\bigg| \underset{0\leq l \leq p_0}{\max} \Bigg\{ \Big|\frac{1}{\sqrt{n}} \sn \Big(| \dot{W}_{il} | - \mathbb{E}|\dot{W}_{il} | \Big)\Big| \\& \hspace{7.5cm} + \frac{1}{\sqrt{n}} \sn \mathbb{E}|\dot{W}_{il} | \Big)\Bigg\}\\& = o(n^{-a_1}).
\end{split}
\end{equation}
Therefore, from \eqref{9.2} and \eqref{9.25}, we get with probability tending to one,
\begin{equation} \label{10}
   \| \bm\Lambda_n^{(1)} \|_\infty \leq \| \bm\Lambda_n^{(2)} \|_\infty + \| \bm\Lambda_n^{(3)} \|_\infty \leq C \sqrt{\log{(np_0)}}.
\end{equation}
As $\| \tilde{\bm{L}}_{11,n}^{-1}\bm\Lambda_n^\on \|_\infty \leq \|\tilde{\bm{L}}_{11,n}^{-1} \|_\infty \| \bm\Lambda_n^\on \|_\infty,$ we have from \eqref{9} and \eqref{10}
\begin{equation} \label{11}
    \| \tilde{\bm{L}}_{11,n}^{-1}\bm\Lambda_n^\on \|_\infty \leq C n^{a_1} \sqrt{\log{(np_0)}},
\end{equation}
with probability tending to one. Next, we bound the term $\bm\zeta_n^\on.$ Note that, 
\begin{align}
    \|\bm\zeta_n^\on\|_\infty &\leq \frac{1}{2n^{3/2}}  \max_{0\leq j\leq p_0} \Bigg| \sn h(\xi_i) \left(\dot{W}_i^{(1)\prime}\bm{u_n}^{(1)}\right)^2 \dot{W}_{ij}\Bigg| \nonumber\\
    &\leq \frac{1}{n^{3/2}} \max_{0\leq j\leq p_0} \sn \Big|\dot{W}_{ij}\Big| \Big(\dot{W}_i^{(1)\prime}\bm{u_n}^{(1)}\Big)^2 \nonumber\\
    & \leq \frac{\|\bm{u_n}^{(1)}\|^2 }{n^{3/2}} \max_{0\leq j\leq p_0} \bigg\| \sn \Big|\dot{W}_{ij}\Big| \dot{W}_i^{(1)} \dot{W}_i^{(1)\prime}\bigg\| \nonumber\\
    & \leq \frac{\|\bm{u_n}^{(1)}\|^2 }{n^{3/2}} \max_{0\leq j\leq p_0} \bigg\| \sn\bigg\{ \Big|\dot{W}_{ij}\Big| \dot{W}_i^{(1)} \dot{W}_i^{(1)\prime} - \E \Big( \Big|\dot{W}_{ij}\Big| \dot{W}_i^{(1)} \dot{W}_i^{(1)\prime}\Big) \bigg\}\bigg\|_F \nonumber\\
    & \hspace{3cm} + \frac{\|\bm{u_n}^{(1)}\|^2 }{n^{3/2}} \max_{0\leq j\leq p_0} \bigg\| \sn \E\bigg(\Big|\dot{W}_{ij}\Big| \dot{W}_i^{(1)} \dot{W}_i^{(1)\prime} \bigg)\bigg\| \nonumber\\
    & \leq \frac{p_0 \|\bm{u_n}^{(1)}\|^2 }{n^{3/2}} \max_{0\leq j,l,l'\leq p_0} \bigg| \sn\bigg\{ \Big|\dot{W}_{ij}\Big| \dot{W}_{il} \dot{W}_{il'} - \E \Big( \Big|\dot{W}_{ij}\Big| \dot{W}_{il}^{(1)} \dot{W}_{il}^{(1)}\Big) \bigg\}\bigg| \nonumber\\
    & \hspace{3cm} + \frac{\|\bm{u_n}^{(1)}\|^2 }{n^{3/2}} \max_{0\leq j\leq p} \bigg\| \sn \E\bigg(\Big|\dot{W}_{ij}\Big| \dot{W}_i^{(1)\prime} \dot{W}_i^{(1)\prime} \bigg)\bigg\| \nonumber\\
    &\leq C \;\frac{\|\bm{u_n}^{(1)}\|^2 }{\sqrt{n}} \left( p_0\sqrt{\frac{\log(np_0)}{n}} + 1 \right) \nonumber\\
    & \leq C \;\frac{\|\bm{u_n}^{(1)}\|^2 }{\sqrt{n}}, \label{12.0}
\end{align}
    
where the first and fifth inequality follows from taking maximum over all components, the second and fourth inequality follows from triangle inequality and the fact that $0\leq h(\xi_i) \leq 1,$ almost surely. The third inequality follows from H\"older's inequality. The last two inequalities follows from Lemma \ref{kcvl2}, Lemma \ref{kcvl3}, and assumptions (A.4).

Together with \eqref{9}, we get
\begin{equation} \label{12}
    \| \tilde{\bm{L}}_{11,n}^{-1}\bm\zeta^\on \|_\infty \leq \| \tilde{\bm{L}}_{11,n}^{-1}\|_\infty \|\bm\zeta^\on \|_\infty \leq \frac{C}{\sqrt{n}} n^{a_1} \|\bm{u_n}^{(1)}\|^2, 
\end{equation}
 with probability tending to one. Therefore,
\begin{equation} \label{13}
    \| \tilde{\bm{L}}_{11,n}^{-1}\bm\zeta^\on \|_\infty \leq \|\bm{u_n}^{(1)}\|_\infty,
\end{equation}
with probability tending to one, provided $\Big\{\frac{C\sqrt{p_0}n^{a_1}}{\sqrt{n}} \|\bm{u_n}^{(1)}\| \leq 1 \Big\}$, with probability tending to one, which we later show is true by assumption (A.4)(ii).

Finally, by \eqref{9}, we have with probability tending to one,
\begin{equation} \label{14}
    \bigg\| \frac{\lambda_n}{\sqrt{n}} \tilde{\bm{L}}_{11,n}^{-1} \hat{\bm{s}}_n^{(1)} \bigg\|_\infty \leq C\frac{\lambda_n}{\sqrt{n}} n^{a_1}.
\end{equation}

Combining \eqref{11}, \eqref{13} and \eqref{14} with the assumption (A.4), we get $\bm{u}_{n}^\on = f(\bm{u}_{n}^\on)$ with \begin{equation*}
     \| f(\bm{u}_{n}^\on) \|_\infty \leq C \frac{\lambda_n}{\sqrt{n}} n^{a_1},
\end{equation*}
with probability tending to one. Therefore by Brouwer's fixed point theorem, \eqref{5} has a solution $\tilde{\bm{u}}_n^{(1)}$ such that
\begin{equation} \label{15}
\left\| \tilde{\bm{u}}_n^{(1)} \right\|_\infty \leq C\dfrac{\lambda_{n}}{\sqrt{n}} n^{a_1},
\end{equation}
with probability tending to one. It follows that $\Big\{\frac{C\sqrt{p_0}n^a_1}{\sqrt{n}} \|\tilde{\bm{u}}_n^{(1)}\| \leq 1 \Big\}$, with probability tending to one, by assumption (A.4)(ii) and we get that  $\{\hat{s}_{nj}= s_{nj},$ $\;j=0,\dots,p_0\}$ with probability tending to one due to \eqref{15} and the beta-min condition. 

Now, let us consider \eqref{6}. Note that, for $j=(p_0+1),\dots,p,$ by Taylor's theorem
\begin{align*}
    &-\dfrac{1}{\sqrt{n}}\sum_{i=1}^{n}{z_{i}\dot{W}_{ij}+ \dfrac{1}{\sqrt{n}}\sum_{i=1}^{n} \dot{W}_{ij}\dfrac{n_2\exp\left(\dot{W}_i^{(1)\prime}\left(\frac{\bm{u}_n^{(1)}}{\sqrt{n}} + \bm{\beta_n}^{(1)}\right)\right)}{n_1+n_2\exp\left(\dot{W}_i^{(1)\prime}\left(\frac{\bm{u}_n^{(1)}}{\sqrt{n}} + \bm{\beta_n}^{(1)}\right)\right)}}\\
    &= - \dfrac{1}{\sqrt{n}} \sn (z_i - \hat{p}(\bm\beta|\dot{W}_i^\on)) \dot{W}_{ij} \\& \qquad \qquad+ \left[ \frac{1}{n} \sn \hat{p}(\bm\beta|\dot{W}_i^\on) (1 - \hat{p}(\bm\beta|\dot{W}_i^\on)) \dot{W}_i^\on \dot{W}_{ij} \right]^\top \bm{u}_n^\on
    \\& \qquad \qquad + \frac{1}{2n^{3/2}} \sn \frac{n_1 n_2\exp(\xi_i)(n_1-n_2\exp(\xi_i))}{(n_1+n_2\exp(\xi_i))^3} \left(\dot{W}_i^{(1)\prime}\bm{u_n}^{(1)}\right)^2 \dot{W}_{ij}.
\end{align*}
where $\xi_i$ is between $\dot{W}_i^{(1)\prime}\bm{\beta_n}^{(1)}$ and $\dot{W}_i^{(1)\prime}\left(\frac{\bm{u}_n^{(1)}}{\sqrt{n}} + \bm{\beta_n}^{(1)}\right),\; i=1,\dots n.$ Hence, \eqref{6} is equivalent to 
\begin{multline*}
    -\frac{\lambda}{\sqrt{n}} \bm{1} \leq - \dfrac{1}{\sqrt{n}} \sn (z_i - \hat{p}(\bm\beta|\dot{W}_i^\on)) \dot{W}_{i}^{(2)}+ \tilde{\bm{L}}_{21,n} \bm{u_n}^{(1)} \\ + \frac{1}{2n^{3/2}} \sn \frac{n_1 n_2\exp(\xi_i)(n_1-n_2\exp(\xi_i))}{(n_1+n_2\exp(\xi_i))^3} \left(\dot{W}_i^{(1)\prime}\bm{u_n}^{(1)}\right)^2 \dot{W}_{i}^{(2)} \leq \frac{\lambda}{\sqrt{n}} \bm{1}.
    \end{multline*}
Therefore, it is enough to show that for all $j=(p_0+1),\dots,p,$
\begin{multline} \label{16}
    \left|\dfrac{1}{\sqrt{n}} \sum_{i=1}^{n} {(z_{i}-\hat{p}(\bm\beta|\dot{W}_i^\on))} \dot{W}_{ij}\right| + \left|{(\tilde{\bm{L}}_{21,n})}_{j.}^\top \tilde{\bm{L}}_{11,n}^{-1} \left( \bm\Lambda_n^{(1)} - \bm\zeta_n^{(1)}\right)\right| \\+
     \left|\frac{1}{2n^{3/2}} \sn \frac{n_1 n_2\exp(\xi_i)(n_1-n_2\exp(\xi_i))}{(n_1+n_2\exp(\xi_i))^3} \left(\dot{W}_i^{(1)\prime}\bm{u_n}^{(1)}\right)^2 \dot{W}_{ij}\right|\\ + \left|\dfrac{\lambda_{n}}{\sqrt{n}} {(\tilde{\bm{L}}_{21,n})}_{j.}^\top \tilde{\bm{L}}_{11,n}^{-1} {\hat{\bm{s}}}_n^{(1)}\right| \leq \dfrac{\lambda_{n}}{\sqrt{n}}.
\end{multline}
To determine the bounds for each of the quantities in the left-hand side of \eqref{16}, we define
\begin{align*}
    I & = \underset{(p_0+1)\leq j \leq p}{\max}\left|\dfrac{1}{\sqrt{n}} \sum_{i=1}^{n} {(z_{i}-\hat{p}(\bm\beta|\dot{W}_i^\on))} \dot{W}_{ij}\right|,\\
    II & = \underset{(p_0+1)\leq j \leq p}{\max}\left|{(\tilde{\bm{L}}_{21,n})}_{j.}^\top \tilde{\bm{L}}_{11,n}^{-1} \bm\Lambda_n^{(1)}\right|,\\
    III & = \underset{(p_0+1)\leq j \leq p}{\max}\left|{(\tilde{\bm{L}}_{21,n})}_{j.}^\top \tilde{\bm{L}}_{11,n}^{-1} \bm\zeta_n^{(1)}\right|,\\
    IV &= \underset{(p_0+1)\leq j \leq p}{\max}\left|\frac{1}{2n^{3/2}} \sn \frac{n_1 n_2\exp(\xi_i)(n_1-n_2\exp(\xi_i))}{(n_1+n_2\exp(\xi_i))^3} \left(\dot{W}_i^{(1)\prime}\bm{u_n}^{(1)}\right)^2 \dot{W}_{ij}\right|,\\
    \text{and } V &= \underset{(p_0+1)\leq j \leq p}{\max}\left| {(\tilde{\bm{L}}_{21,n})}_{j.}^\top \tilde{\bm{L}}_{11,n}^{-1} {{\bm{s}}}_n^{(1)}\right|.
\end{align*}
Note that by \eqref{15}, we had with probability tending to one, $\{ {\hat{\bm{s}}}_n^{(1)} = {{\bm{s}}}_n^{(1)} \}$. Therefore, we can replace ${\hat{\bm{s}}}_n^{(1)}$ with ${\bm{s}}_n^{(1)}$ in \eqref{16} and rewrite the equation as
\begin{equation} \label{17}
    \dfrac{\lambda_{n}}{\sqrt{n}} (1-V) \geq I + II + III + IV.
\end{equation}
For the first term, following similar calculations as in \eqref{9.2} and \eqref{9.25}, we get
\begin{equation} \label{18}
\begin{split}
    I &\leq \underset{(p_0+1)\leq j \leq p}{\max}\left|\dfrac{1}{\sqrt{n}} \sum_{i=1}^{n} {(z_{i}-p(\bm\beta|\dot{W}_i^\on))} \dot{W}_{ij}\right|\\ & \hspace{3cm}+ \underset{(p_0+1)\leq j \leq p}{\max}\left|\dfrac{1}{\sqrt{n}} \sum_{i=1}^{n} {(\hat{p}(\bm\beta|\dot{W}_i^\on)-p(\bm\beta|\dot{W}_i^\on))}\dot{W}_{ij}\right|\\
    & \leq C \left\{\sqrt{\log(np)} + \frac{(\log(n))^{1/\kappa}(\log(np))^{1/\min\{1,\kappa\}}}{\sqrt{n}} \right\}.
\end{split}
\end{equation}
To determine the bounds for the quantities in $II, III$, and $V$, we need to determine a bound for the quantity $\underset{(p_0+1)\leq j \leq p}{\max} \| {(\tilde{\bm{L}}_{21,n})}_j \|_1.$ Note that,
\begin{multline} \label{19.3}
        \underset{(p_0+1)\leq j \leq p}{\max} \Big\| {(\tilde{\bm{L}}_{21,n})}_j \Big\|_1 \leq \underset{(p_0+1)\leq j \leq p}{\max} \Big\| {(\tilde{\bm{L}}_{21,n} - \mathbb{E}(\tilde{\bm{L}}_{21,n}))}_j \Big\|_1 \\ \hspace{4cm}
        + \underset{(p_0+1)\leq j\leq p}{\max} \Big\|\Big(\mathbb{E}(\tilde{\bm{L}}_{21,n}) - \mathbb{E}(\bm{L}_{21,n})\Big)_j \Big\|_1 \\+ \underset{(p_0+1)\leq j\leq p}{\max} \Big\|\mathbb{E}(\bm{L}_{21,n})_j \Big\|_1
\end{multline}
For the first term in the right-hand side of \eqref{19.3}, using Lemma \ref{kcvl2}, Lemma \ref{kcvl3}, we have with probability tending to one,
\begin{equation} \label{19.1}
\begin{split}
    &\underset{(p_0+1)\leq j \leq p}{\max} \Big\| {(\tilde{\bm{L}}_{21,n} - \mathbb{E}(\tilde{\bm{L}}_{21,n}))}_j \Big\|_1 \\
    &\qquad \leq {p_0}\; \underset{(p_0+1)\leq j \leq p}{\max}\;\underset{0\leq l \leq p_0}{\max} \Bigg| \frac{1}{n} \sn \bigg( \hat{p}(\bm\beta|\dot{W}_i^\on)(1-\hat{p}(\bm\beta|\dot{W}_i^\on))\dot{W}_{il} \dot{W}_{ij} \\&\hspace{5cm} - \mathbb{E}\Big(\hat{p}(\bm\beta|\dot{W}_i^\on) (1-\hat{p}(\bm\beta|\dot{W}_i^\on))\dot{W}_{il} \dot{W}_{ij}\Big) \bigg)\Bigg| \\
    &\qquad \leq C {p_0} \left[\sqrt{\frac{\log(np)}{n}} + \frac{(\log(n))^{2/\kappa}(\log(np))^{1/\min\{1,\kappa/2\}}}{n}\right],
\end{split}
\end{equation}
and using assumption (A.3), we have with probability tending to one,
\begin{equation} \label{19.2}
\begin{split}
    & \underset{(p_0+1)\leq j\leq p}{\max} \Big\|\Big(\mathbb{E}(\tilde{\bm{L}}_{21,n}) - \mathbb{E}(\bm{L}_{21,n})\Big)_{j.} \Big\|_1\\
    &\qquad \leq C {p_0}\bigg|\frac{n_1}{n_2} - \frac{\pi_1}{\pi_2}\bigg|  \underset{(p_0+1)\leq j\leq p}{\max} \; \underset{0\leq l\leq p_0}{\max} \frac{1}{n} \sum_{i=1}^{n} \mathbb{E} |\dot{W}_{ij}\dot{W}_{il}|\\
    & \qquad = o(p_0 n^{-1/2-{a_1}}).
\end{split}
\end{equation}
We have $\underset{(p_0+1) \leq j\leq p}{\max}\|\mathbb{E}({\bm{L}}_{21,n})_{j.}\|_1 = C p_0,$ due to the sub-Weibull assumption. Therefore, using assumption (A.4) combining with \eqref{19.1}, \eqref{19.2} and \eqref{19.3}, we get with probability tending to one
\begin{multline}  \label{19}
    \underset{(p_0+1)\leq j\leq p}{\max} \Big\|\Big(\tilde{\bm{L}}_{21,n} - \mathbb{E}(\bm{L}_{21,n})\Big)_j \Big\|_1 \\ \leq C {p_0} \left[\sqrt{\frac{\log(np)}{n}} + \frac{(\log(n))^{2/\kappa}(\log(np))^{1/\min\{1,\kappa/2\}}}{n}\right],
\end{multline}
and 
\begin{equation} \label{19.4}
    \underset{(p_0+1)\leq j\leq p}{\max} \Big\|(\tilde{\bm{L}}_{21,n})_j \Big\|_1 \leq C{p_0}.
\end{equation}
Note that, we will require the bound in \eqref{19} later to bound $V$. Using H\"older's inequality and \eqref{11}, \eqref{12}, \eqref{15} and \eqref{19.4}, we get
\begin{align}
    II & \leq C p_0 n^{a_1} \sqrt{\log(np_0)}, \label{20}\\
    III & \leq C \frac{\lambda_n^2}{n^{3/2}}p_0^2 n^{3a_1}. \label{21}
\end{align}
Next, following similar arguments as in \eqref{12.0}, we get from assumptions (A.2), with probability tending to one
\begin{equation} \label{22}
    \begin{split}
        IV &= \frac{1}{2n^{3/2}}\underset{(p_0+1)\leq j \leq p}{\max}\left| \sn h(\xi_i) \left(\dot{W}_i^{(1)\prime}\bm{u_n}^{(1)}\right)^2 \dot{W}_{ij}\right|\\
    & \leq C \;\frac{\|\bm{u_n}^{(1)}\|^2 }{\sqrt{n}}\\
    & \leq C\frac{\lambda_n^2}{n^{3/2}} p_0 n^{2a_1}.
    \end{split}
\end{equation}
Finally, with probability tending to one, we can get
\begin{equation} \label{25}
\begin{split} 
    V &= \underset{(p_0+1)\leq j\leq p}{\max}\left| {(\tilde{\bm{L}}_{21,n})}_{j.}^\top \tilde{\bm{L}}_{11,n}^{-1} {{\bm{s}}}_n^{(1)}\right|\\
    & \leq \underset{(p_0+1)\leq j\leq p}{\max}\left| {(\tilde{\bm{L}}_{21,n} - {\mathbb{E}(\bm{L}_{21,n}))}_{j.}^\top} \tilde{\bm{L}}_{11,n}^{-1} {{\bm{s}}}_n^{(1)}\right|\\&\qquad + \underset{(p_0+1)\leq j\leq p}{\max}\left| {(\mathbb{E}(\bm{L}_{21,n}))_{j.}^\top} (\tilde{\bm{L}}_{11,n}^{-1} - {(\mathbb{E}(\bm{L}_{11,n}))^{-1}}){{\bm{s}}}_n^{(1)}\right|\\&\qquad + \underset{(p_0+1)\leq j\leq p}{\max}\left| ({\mathbb{E}(\bm{L}_{21,n}))_{j.}^\top} ({\mathbb{E}(\bm{L}_{11,n}))^{-1}}{{\bm{s}}}_n^{(1)}\right|\\
    & \leq 1-\tau_1/2,
\end{split}
\end{equation}
where the final inequality follows from assumption (A.1), (A.4) and the bounds obtained in \eqref{8.2}, \eqref{9} and \eqref{19} for large enough $n$.
Combining \eqref{18}, \eqref{20}, \eqref{21}, \eqref{22}, \eqref{25} and assumptions (A.1) and (A.4), we can conclude that the solution $\bm{{\tilde{u}}}_n^{(1)}$ to \eqref{2} satisfies \eqref{7}.

Hence, under assumptions (A.1)--(A.4), we can conclude that there exists a solution $\tilde{\bm{u}}_n$ of \eqref{2} which has the form $\tilde{\bm{u}}_n = {( { \bm{{\tilde{u}}}_n^{(1)}}^{\prime} ,\bm{0}^\top)}^\top$, where $\tilde{\bm{u}}_n^{(1)}$ has the property that $$\mathbbm{P}\left( \left\| \tilde{\bm{u}}_n^{(1)} \right\|_\infty \leq C\dfrac{\lambda_{n}}{\sqrt{n}} n^{a} \right)\to 1\; \text{ as } n\to\infty.$$ \end{proof}


\begin{lem} \label{pl}
    (Post LASSO LR) Recall the post-Lasso Logistic estimator $\hat{\bm{\beta}}_n = \Big(\big(\hat{\bm{\beta}}^{\tilde{\mathcal{A}}_n}_n\big)^\top, \bm{0}^\top\Big)^\top$ defined in Section 2 of main manuscript. Define $\hat{u}_n = \sqrt{n}(\hat{\bm{\beta}}_n - \bm\beta_n)$. Then under the assumptions (A.1)--(A.3), as $n\rightarrow \infty$, 
    $$\mathbb{P}\Big(\big\|\bm{\hat{u}}_n^{\tilde{\mathcal{A}}_n}\big\|_\infty \leq C n^{a_1} \sqrt{\log{(np_0)}}\Big),$$
for some constant $C\geq 1$.    
\end{lem}
\begin{proof} WLOG, assume $\mathbf{A}_n = \{0,\dots,p_0\}$ and define $\mathbf{A}_n = \{\tilde{\mathcal{A}}_n = \mathcal{A}_n \}$ i.e. the set where VSC holds. Then $\mathbbm{P}(\mathbf{A}_n) = 1-o(1)$ and from \eqref{7}, it follows that on this set, 
\begin{equation} \label{26}
    \hat{\bm{u}}_{n}^\on = \tilde{\bm{L}}_{11,n}^{-1} \left[ \bm\Lambda_n^\on - \bm\zeta_n^\on\right].
\end{equation}
The rest of the proof follows directly from the proof of Lemma \ref{vscnew}. \end{proof}

\begin{lem}\label{lem:asymnor}
  (Normal approximation of the components of the centered version of $\bm{T}^{(1)}_{n}$)\\ Under assumptions (A.1)--(A.3), we have for any $j=0,\dots, p_0$
   \begin{equation*}
    \underset{x\in\mathbb{R}}{\sup} \;| \mathbbm{P}(T^{(1)}_{1nj} \leq  x) - {\Phi}(\sigma^{-1}_{nj} x)| = o(1),
\end{equation*} 
where $\bm{T}^{(1)}_{1nj} = \sqrt{n}\big(\hat{\bm{\beta}}_{nj}^\on - {\bm{\beta}}_{nj}^\on) + b_1(p,n) Z_{1j}$ and $\sigma^2_{nj} = \Big((\mathbb{E}\bm{L}_{11, n})^{-1}\Big)_{j j}$.
\end{lem}

\begin{proof} For any $j \in \{0,\dots,p_0\}$, the $(j+1)$th component of $\bm{T}_{1n}$ is given by (see equation \eqref{26})
\begin{align*}
    T^{(1)}_{1nj} &= \sqrt{n}\big(\hat{\bm{\beta}}_{nj}^\on - {\bm{\beta}}_{nj}^\on) + b_1(p,n) Z_{1j}\\
    &= (\tilde{\bm{L}}_{11, n}^{-1})^\prime_{j.} \Big(\bm\Lambda_n^\on - \bm\zeta_n^\on \Big) + b_1(p,n) Z_{1j} \\
    &= \Big((\mathbb{E}\bm{L}_{11, n})^{-1}\Big)_{j.}^\top \bm\Lambda_n^{(2)} + b_1(p,n) Z_{1j}\\&\hspace{.75cm} - \Big((\mathbb{E}\bm{L}_{11, n})^{-1}\Big)_{j.}^\top \bm\Lambda_n^{(3)}  + \Big(\tilde{\bm{L}}_{11, n}^{-1} - (\mathbb{E}\tilde{\bm{L}}_{11, n})^{-1}\Big)_{j.}^\top \bm\Lambda_n^\on \\&\hspace{.75cm}  + \Big((\mathbb{E}\tilde{\bm{L}}_{11, n})^{-1} - (\mathbb{E}\bm{L}_{11, n})^{-1}\Big)_{j.}^\top \bm\Lambda_n^\on  - (\tilde{\bm{L}}_{11, n}^{-1})_{j.}^\top \bm\zeta_n^\on \\
    & = S_{nj} + R_{nj},
\end{align*}
where
\begin{align*}
    & S_{nj} = \Big((\mathbb{E}\bm{L}_{11, n})^{-1}\Big)_{j.}^\top \bm\Lambda_n^{(2)},\\
    & R_{nj} = -\Big((\mathbb{E}\bm{L}_{11, n})^{-1}\Big)_{j.}^\top \bm\Lambda_n^{(3)} + \Big(\tilde{\bm{L}}_{11, n}^{-1} - (\mathbb{E}\tilde{\bm{L}}_{11, n})^{-1}\Big)_{j.}^\top \bm\Lambda_n^\on \\&\hspace{1.5cm}  + \Big((\mathbb{E}\tilde{\bm{L}}_{11, n})^{-1} - (\mathbb{E}\bm{L}_{11, n})^{-1}\Big)_{j.}^\top \bm\Lambda_n^\on  - (\tilde{\bm{L}}_{11, n}^{-1})_{j.}^\top \bm\zeta_n^\on  + b_1(p,n) Z_{1j},
\end{align*}
and all the other notations are defined as in the proof of Lemma~\ref{vscnew}.

We will show that the distribution of $T_{1nj}$ can be approximated by $S_{nj}$ and $R_{nj}$ is a remainder term. From assumptions (A.2)(i), (A.3), (A.4) and equations \eqref{8.2}, \eqref{9}, \eqref{9.25}, \eqref{10}, \eqref{12} and \eqref{15}, with probability tending to one,
\begin{align} \label{30}
    |R_{nj}-b_1(p,n) Z_{1j}| \leq \Delta_{n},
\end{align}
where 
\begin{align*}
    \Delta_{n} 
     &= C n^{1/2+a_1} \Big\{p_0 \;{\log(np_0)}\; n^{-1+2a_1} + \Big| n_1/n_2 - \pi_1/\pi_2 \Big|\Big\} = o(1)
\end{align*}
and $b_1(p,n) Z_{1j} = o_p(1)$ as $b_1(p,n)=o(1)$. For the other term, we can write
\begin{align*}
    S_{nj}& = \frac{1}{\sqrt{n}} \sn  \left( z_{i} - p(\bm\beta|\dot{W}_i^\on) \right) 
    (\mathbb{E}{\bm{L}}_{11, n})^{-1})_{j}^\top \dot{W}_{i}^{(1)} \\
    & = \frac{1}{\sqrt{n}} \sn V_{nij},
\end{align*}
where
\begin{equation*}
    V_{nij} = \left( z_{i} - p(\bm\beta|\dot{W}_i^\on) \right) 
    (\mathbb{E}\bm{L}_{11, n})^{-1})_{j.}^\top \dot{W}_{i}^{(1)},\;i=1,\dots,n.
\end{equation*}
Note that
\begin{align*}
    \mathbb{E}(S_{nj}) = 0 \text{ and } \sigma^2_{nj} = Var(S_{nj}) = ((\mathbb{E}\bm{L}_{11, n})^{-1})_{j j},
\end{align*}
for all $j = 0,\dots, p_0$. Using assumption (A.2)(ii), there exists a constant $c>0$ such that $c \leq \sigma_{nj}$ for large enough $n$.

Also, using assumption (A.2)(i) together with H\"older's inequality, we get
\begin{align*}
    \mathbb{E}|V_{nij}|^3 &\leq C \mathbb{E}\Big| ((\mathbb{E} \bm{L}_{11, n})^{-1})_{j.}^\top \dot{W}_{i}^{(1)} \Big|^3\\
    & \leq C \Big(\Big\| (\mathbb{E} \bm{L}_{11, n})^{-1}\Big\|_\infty^3 \mathbb{E}\Big\|\dot{W}_{i}^\on \Big\|_\infty^3 \Big)\\
    &\leq Cn^{3a_1} (\log p_0)^{3/2}.
\end{align*}
Then, using the Berry-Esseen theorem (see \citet[Theorem~12.4]{BR86}) and assumption (A.4), we get
\begin{equation} \label{31}
    \underset{x\in\mathbb{R}}{\sup} |\mathbbm{P}(S_{nj} \leq x) - {\Phi}(\sigma^{-1}_{nj} x)| \leq C \left(\frac{Cn^{3a_1}(\log p_0)^{3/2}}{ c^3\;\sqrt{n}}\right) = o(1).
\end{equation}
Now, using \eqref{30} and \eqref{31},
\begin{align*}
        &|\mathbbm{P}(T_{1nj} \leq  x) - {\Phi}(\sigma^{-1}_{nj} x)|
        \\& \leq 2\;\mathbbm{P}\Big( |R_{nj}| > \Delta_n \Big) + 2\;\mathbbm{P}\Big(x-\Delta_n\leq S_{nj} \leq x+\Delta_n\Big)+ |\mathbbm{P}(S_{nj} \leq  x) - \Phi(\sigma^{-1}_{nj}x)|\\
        &= o(1).
    \end{align*} 
This completes the proof. \end{proof}

\begin{lem} \label{vscm}
(VSC of Multi-class Logistic LASSO) Define $\tilde{\bm{u}}_n = \sqrt{n} (\tilde{\bm{\beta}}_n - \bm{\beta}_n)$ where $\tilde{\bm{\beta}}_n = \big(\tilde{\bm{\beta}}_{n,1}^\top, \dots, \tilde{\bm{\beta}}_{n, (K-1)}^\top\big)^\top$ is the multi-sample Logistic Lasso estimator as defined in Section 4 of the main manuscript. Also assume that $\mathcal{A}_{n} = \cup_{j=1}^{K-1} \mathcal{A}_{n,j}$ is the set of relevant covariates, where $\mathcal{A}_{n,j}= \{k:\beta_{jk}\neq 0\}= \{0, \dots, p_{0j}\},\; j=1,\dots,(K-1)$ and $\tilde{\mathcal{A}}_{n} = \cup_{j=1}^{K-1} \tilde{\mathcal{A}}_{n,j}$ is the estimator of it based on $\tilde{\bm{\beta}}_n$, where $\tilde{\mathcal{A}}_{n,j} = \{k:\tilde{\beta}_{jk}\neq 0\},\; j=1,\dots,(K-1)$. Then under the assumptions (C.1)--(C.4), we have as $n\rightarrow \infty$, 
\begin{equation*}
    \mathbb{P}\Big(\tilde{\mathcal{A}}_{n} = \cup_{j=1}^{K-1} \{0,\dots, p_{0j}\},\\\; \text{and}\; \tilde{\bm{u}}_n = \big(\big(\tilde{\bm{u}}_n^{\tilde{\mathcal{A}}_n}\big)^\top, \bm{0}^\top\big)^\top\; \text{with}\; \|\tilde{\bm{u}}_n^{(1)\prime}\|_\infty\leq C\dfrac{\lambda_{n}}{\sqrt{n}} n^{a}\Big) \rightarrow 1,
\end{equation*}
for some constant $C\geq 1$.
\end{lem}
\begin{proof} This proof closely follows the proof of Lemma \ref{vscnew}. Let the logistic lasso estimator $\tilde{\bm{\beta}}_n$ be defined as in Section~3 in the main manuscript. Define, $\tilde{\bm{u}}_n = \sqrt{n} (\tilde{\bm{\beta}}_n - \bm{\beta}_n)$. Then, we get (see \citet[Section~4.1.3]{boyd2004} for change of variables)
\begin{multline*}
    \tilde{\bm{u}}_n = \underset{\Big(\bm{u}_1^\top, \dots, \bm{u}_{K-1}^\top\Big)^\top \in \mathcal{R}^{(K-1)(p+1)}}{\arg\min} \left[-\sn\scj y_{ij}\dot{U}_i^\top\frac{\bm{u}_j}{\sqrt{n}} \right.\\\left. + \sn\scj \log \left(n_K+\scj n_j\exp\left(\dot{U}_i^\top\left(\frac{\bm{u}_j}{\sqrt{n}} + \bm{\beta}_{j}\right)\right)\right)\right.\\ +\left.\lambda_n\scj\left({\left\|\frac{\bm{u}_j}{\sqrt{n}} + \bm\beta_{j}\right\|}_1-{\|\bm\beta_{j}\|}_1\right)\right]
\end{multline*}
which is equivalent to the KKT conditions
\begin{equation} \label{m2}
    \begin{split}
    -\dfrac{1}{\sqrt{n}}\sn {y_{ij}\dot{U}_{ik}+ \sn \dfrac{\dot{U}_{ik}}{\sqrt{n}}\dfrac{n_j \exp\left(\dot{U}_i^\top\left(\frac{\bm{u_j}}{\sqrt{n}} + \bm{\beta}_j\right)\right)}{n_K+\scl n_l \exp\left(\dot{U}_i^\top\left(\frac{\bm{u}_l}{\sqrt{n}} + \bm{\beta}_l\right)\right)}}\\=-\dfrac{\lambda_{n}}{\sqrt{n}} sgn\left(\frac{u_{jk}}{\sqrt{n}} + \beta_{jk}\right), \text{if } \frac{u_{jk}}{\sqrt{n}} + \beta_{jk} \neq 0,
    \end{split}
\end{equation}
    \text{and  }
    \begin{equation} \label{m3}
        \begin{split}
            -\frac{\lambda_{n}}{\sqrt{n}}\leq -\dfrac{1}{\sqrt{n}}\sn{y_{ij}\dot{U}_{ik}+ \sn \dfrac{\dot{U}_{ik}}{\sqrt{n}}\dfrac{n_j\exp\left(\dot{U}_i^\top\left(\frac{\bm{u}_j}{\sqrt{n}} + \bm{\beta}_j\right)\right)}{n_K+\scl n_l \exp\left(\dot{U}_i^\top\left(\frac{\bm{u}_l}{\sqrt{n}} + \bm{\beta}_l\right)\right)}}\\ \leq \dfrac{\lambda_{n}}{\sqrt{n}},\text{if } \frac{u_{jk}}{\sqrt{n}} + \beta_{jk} = 0.
        \end{split}
    \end{equation} 
Now, VSC holds if $\tilde{\mathcal{A}}_{n,j} = \mathcal{A}_{n,j}, j=1,\dots,(K-1).$ WLOG, let $\mathcal{A}_{n,j} = \{0,1,\dots,(p_{0j}-1)\},j=1,\dots,(K-1).$ Then, VSC holds if for all $j=1,\dots,(K-1)$, $\tilde{\mathcal{A}}_{n,j} = \{0,1,\dots,p_{0j}\}$. This is equivalent to $\frac{\tilde{u}_{jk}}{\sqrt{n}} + \beta_{jk} \neq 0, \text{for } k = k_j = 0,1,\dots,p_{0j}, j=1,\dots,(K-1)$ and $\frac{\tilde{u}_{jk}}{\sqrt{n}} + \beta_{jk} = 0, \text{for } k = k_j = p_{0j}\dots,p, j=1,\dots,(K-1),$ where $\tilde{\bm{u}}_n = \sqrt{n} (\tilde{\bm{\beta}}_n - \bm{\beta}_n)$ is the solution of \eqref{m2}. Thus, due to equation \eqref{m2} and \eqref{m3},
\begin{multline} \label{m4}
    {\left(-\dfrac{1}{\sqrt{n}}\sn {y_{ij}\dot{U}_{i}^\onj + \sn \dfrac{\dot{U}_i^\onj}{\sqrt{n}}\dfrac{n_j \exp\left(\dot{U}_i^{(j)\prime}\left(\frac{\bm{u}_j^\on}{\sqrt{n}} + \bm{\beta}_j^\on\right)\right)}{n_K+\scl n_l \exp\left(\dot{U}_i^{(l)\prime}\left(\frac{\bm{u}^\on_l}{\sqrt{n}} + \bm{\beta}^\on_l\right)\right)}}\right)}_{j=1}^{K-1}\\= {\left(-\dfrac{\lambda_{n}}{\sqrt{n}} \bm{s}_{j}^\on\right)}_{j=1}^{K-1}
\end{multline}
    \text{and  }
    \begin{multline} \label{m5}
        -\frac{\lambda_{n}}{\sqrt{n}} \bm{1} \\ \leq {\left(-\dfrac{1}{\sqrt{n}}\sn {y_{ij}\dot{U}_{i}^\onjc + \sn \dfrac{\dot{U}_i^\onjc}{\sqrt{n}}\dfrac{n_j\exp\left(\dot{U}_i^{(j)\prime}\left(\frac{\bm{u}_j^\on}{\sqrt{n}} + \bm{\beta}_j^\on\right)\right)}{n_K+\scl n_l \exp\left(\dot{U}_i^{(l)\prime}\left(\frac{\bm{u}^\on_l}{\sqrt{n}} + \bm{\beta}^\on_l\right)\right)}}\right)}_{j=1}^{K-1}\\ \leq \dfrac{\lambda_{n}}{\sqrt{n}} \bm{1},
    \end{multline}
    where $\bm{s}_{j}^\on = sgn\left(\frac{u_{jk}}{\sqrt{n}} + \beta_{jk}\right).$ Applying Taylor's theorem on \eqref{m4} and by similar calculation as \eqref{7}, we get
\begin{equation} \label{m6}
    \bm{u}_{n}^\on = \tilde{\bm{L}}_{11,n}^{-1} \left[ \bm\Lambda_n^\on - \bm\zeta_n^\on -\frac{\lambda_n}{\sqrt{n}} \hat{\bm{s}}_{n}^{(1)}  \right] = f\left(\bm{u}_{n}^\on \right)
\end{equation}
with
\begin{equation*}
    \bm{u}_{n}^\on = {\left(\bm{u}_{nj}^\on\right)}_{j=1}^{K-1} \equiv {\left(\bm{u}_{j}^\on\right)}_{j=1}^{K-1}
    \text{ and } \hat{\bm{s}}_{n}^{(1)} = {\left(\hat{\bm{s}}_{nj}^{(1)}\right)}_{j=1}^{K-1} \equiv {\left(\hat{\bm{s}}_j^{(1)}\right)}_{j=1}^{K-1},
\end{equation*}
and $\tilde{\bm{L}}_{11,n}$ is a $(p_0+K-1)\times (p_0+K-1)$ matrix with block structure. Let $\tilde{\bm{L}}_{11,n}^{(j,j')}$ be the $(j,j')$th block of $\tilde{\bm{L}}_{11,n}$ which is a $(p_{0j}+1)\times (p_{0j'}+1)$ matrix. Then
\begin{align*}
& \tilde{\bm{L}}_{11,n}^{(j,j')} = \begin{cases} \frac{1}{n} \sn \dot{U}_i^{(j)} \dot{U}_i^{(j){\prime}}  \dfrac{n_j(n_K+\tilde{S}_{i,-j})\exp\left(\dot{U}_i^{(j)\prime}\bm{\beta}_j^\on\right)}{(n_K+\tilde{S}_i)^2}\text{ if } j=j',\\
-\frac{1}{n} \sn \dot{U}_i^{(j)} \dot{U}_i^{(j^{\prime})\prime} \dfrac{n_j n_{j'}\exp\left(\dot{U}_i^{(j)\prime}\bm{\beta}_j^\on\right)\exp\left(\dot{U}_i^{(j^{\prime})\prime}\bm{\beta}_{j'}^\on\right)}{(n_K+\tilde{S}_i)^2} \text{ if } j\neq j',   \end{cases}
\end{align*}
where 
\begin{equation*}
\tilde{S}_i= \sum_{l=1}^{K-1} n_l\exp\left(\dot{U}_i^{(l)\prime}\bm{\beta}^\on_l\right) \text{ and }
\tilde{S}_{i,-j} = \underset{l\neq j}{\sum_{l=1 }^{K-1}} n_l\exp\left(\dot{U}_i^{(l)\prime}\bm{\beta}^\on_l\right).    
\end{equation*}
 Furthermore,
\begin{align*}
    &\bm\Lambda_n^\on = (\bm\Lambda_{nj}^\on)_{j=1}^{K-1} = {\left(\frac{1}{\sqrt{n}} \sum_{i=1}^{n} \left( y_{ij} - \dfrac{n_j\exp\left(\dot{U}_i^{(j)\prime}\bm{\beta}_j^\on\right)}{n_K+\tilde{S}_i} \right) \dot{U}_{i}^{(j)}\right)}_{j=1}^{K-1},\\
    &\bm\zeta_n^\on = (\bm\zeta_{nj}^\on)_{j=1}^{K-1} = {\left(\frac{1}{2n^{3/2}}\sn\sum_{l=1}^{K-1}\sum_{l^{\prime}=1}^{K-1} h^{(l,l')}_{ij} \left(\dot{U}_i^{(l)\prime}\bm{u}_l^\on\right) \left(\dot{U}_i^{(l')\prime}\bm{u}_{l'}^\on\right) \dot{U}_i^{(j)} \right)}_{j=1}^{K-1},
\end{align*}
where we denote $\xi^e_{ij} := \exp({\xi_{ij}}) $ and 
\begin{align*}
    &h^{(l,l')}_{ij} = \begin{cases}
        \dfrac{n_j \Big(n_K+\sum_{k=1}^{K-1} n_k \xi^e_{ik}\Big) \Big(n_K+\sum_{k=1}^{K-1} n_k \xi^e_{ik} - 2n_j\xi^e_{ij}\Big) \xi^e_{ij}}{\Big(n_K+\sum_{k=1}^{K-1} n_k \xi^e_{ik}\Big)^3 },\\\hspace{8.5cm} \text{ if } j=l=l',\\
        -\dfrac{n_j n_l \Big(n_K+\sum_{k=1}^{K-1} n_k \xi^e_{ik} - 2n_j\xi^e_{ij}\Big) \xi^e_{ij} \xi^e_{il} }{\Big(n_K+\sum_{k=1}^{K-1} n_k \xi^e_{ik}\Big)^3 },\\\hspace{8.5cm} \text{ if }  j\neq l=l',\\
        -\dfrac{ n_j n_l \Big(n_K+\sum_{k=1}^{K-1} n_k \xi^e_{ik} - 2n_l\xi^e_{il}\Big) \xi^e_{ij} \xi^e_{il} }{\Big(n_K+\sum_{k=1}^{K-1} n_k \xi^e_{ik}\Big)^3 },\\\hspace{8.5cm} \text{ if }  j\neq l=l',\\
        \dfrac{2 n_j n_l n_{l'} \xi^e_{ij} \xi^e_{il} \xi^e_{il'} }{\Big(n_K+\sum_{k=1}^{K-1} n_k \xi^e_{ik}\Big)^3 },\hspace{3.9cm} \text{ if }  j\neq l\neq l',\\
    \end{cases}
\end{align*}
for $j,l,l' \in \{1,\dots,(K-1)\}$ and $(\xi_{ij})_{j=1}^{K-1}$ is on the line segment joining $\Big(\dot{U}_i^{(j)\prime}\bm{\beta}_j^{(1)}\Big)_{j=1}^{K-1}$ and $\Big(\dot{U}_i^{(j)\prime}\big(\frac{\bm{u}_j^{(1)}}{\sqrt{n}} + \bm{\beta}_j^{(1)}\big)\Big)_{j=1}^{K-1},$ $i=1,\dots n.$

Next, we will determine stochastic bounds for every quantity in \eqref{m6}. Toward that, we proceed by finding a bound for the operator norm of the random matrix $\tilde{\bm{L}}_{11,n}^{-1}.$ Following similar calculations as in \eqref{8.1} and \eqref{8} and the fact that $K$ is fixed, we get with probability tending to one
\begin{equation*}
    \begin{split}
        \| \tilde{\bm{L}}_{11,n} - \mathbb{E}(\tilde{\bm{L}}_{11,n}) \|_\infty &\leq  K \max_{1\leq l,l'\leq (K-1)} \|\tilde{\bm{L}}_{11,n}^{(l,l')} - \mathbb{E}(\tilde{\bm{L}}_{11,n}^{(l,l')}) \|_\infty\\& \leq C p_0 \sqrt{\frac{\log(np_0)}{n}},
    \end{split}
\end{equation*}
and 
\begin{equation}  \label{m7.1}
\begin{split}
    \| \mathbb{E}(\tilde{\bm{L}}_{11,n}) - \mathbb{E}({\bm{L}}_{11,n}) \|_\infty &\leq  K \max_{1\leq l,l'\leq (K-1)} \|\mathbb{E}(\tilde{\bm{L}}_{11,n}^{(l,l')}) - \mathbb{E}({\bm{L}}_{11,n}^{(l,l')}) \|_\infty\\
    &= o(p_0 \;n^{-1/2-a_2}).
\end{split}
\end{equation}
Then, by similar arguments as in \eqref{9}, it follows that, with probability tending to one
\begin{equation} \label{m8}
\begin{split}
   \| \tilde{\bm{L}}_{11,n}^{-1} \|_\infty &\leq C n^{a_2}.
\end{split}
\end{equation}
Next, to bound the quantity $\bm\Lambda_n^{(1)},$ define
\begin{equation*}
\begin{split}
    &\bm\Lambda_n^{(2)} = (\bm\Lambda_{nj}^{(2)})_{j=1}^{K-1} = {\left(\frac{1}{\sqrt{n}} \sum_{i=1}^{n} \left( y_{ij} - \dfrac{\pi_j\exp\left(\dot{U}_i^{(j)\prime}\bm{\beta}_j^\on\right)}{\pi_K+\tilde{\Sigma}_i} \right) \dot{U}_{i}^{(j)}\right)}_{j=1}^{K-1},\\
    &\bm\Lambda_n^{(3)} = (\bm\Lambda_{nj}^{(3)})_{j=1}^{K-1} \\&\qquad = \Bigg(\frac{1}{\sqrt{n}} \sum_{i=1}^{n} \Bigg( \dfrac{\pi_j\exp\Bigg(\dot{U}_i^{(j)\prime}\bm{\beta}_j^\on\Bigg)}{\pi_K+\tilde{\Sigma}_i} - \dfrac{n_j\exp\Bigg(\dot{U}_i^{(j)\prime}\bm{\beta}_j^\on\Bigg)}{n_K+\tilde{S}_i} \Bigg) \dot{U}_{i}^{(j)}\Bigg)_{j=1}^{K-1}\\
\end{split}
\end{equation*}
with
\begin{equation*}
    \tilde{\Sigma}_i = \sum_{l=1}^{K-1} \pi_l\exp\left(\dot{U}_i^{(l)\prime}\bm{\beta}^\on_l\right).
\end{equation*}
Therefore, following similar arguments as in \eqref{9.2}, \eqref{9.25} and \eqref{10}, we get
\begin{equation} \label{m8.1}
\begin{split}
    \|\bm\Lambda_n^\on\|_\infty &\leq K \max_{1\leq j\leq (K-1)} \|\bm\Lambda_{nj}\|_\infty\\ &\leq K \max_{1\leq j\leq (K-1)} \Big\{\|\bm\Lambda_{nj}^{(2)}\|_\infty + \|\bm\Lambda_{nj}^{(3)}\|_\infty \Big\} \\&\leq C \sqrt{\log{(np_0)}},
\end{split}
\end{equation}
with probability tending to one.
By \eqref{m8} and \eqref{m8.1}, with probability tending to one,
\begin{equation} \label{m9}
    \| \tilde{\bm{L}}_{11,n}^{-1}\bm\Lambda_n^\on \|_\infty \leq C n^{a_2} \sqrt{\log{(np_0)}}.
\end{equation}
Similarly, following similar calculations as in \eqref{12.0}, with probability tending to one, we have $\|\bm\xi_n^\on\|_\infty \leq K \max_{1 \leq j\leq (K-1)}\|\bm\xi_{nj}^\on\|_\infty$ and
\begin{align*}
   \|\bm\xi_n^\on\|_\infty \leq K \max_{1 \leq j\leq (K-1)}\|\bm\xi_{nj_1}^\on\| \leq C \frac{\|\bm{u_n}^{(1)}\|^2}{\sqrt{n}},
\end{align*}
and by similar arguments as in \eqref{12}, together with \eqref{m8}, we have,
\begin{equation} \label{m10}
    \| \tilde{\bm{L}}_{11,n}^{-1}\bm\zeta^\on \|_\infty \leq \frac{C}{\sqrt{n}}  n^{a_2} \|\bm{u_n}^{(1)}\|^2. 
\end{equation}
 Therefore, with probability tending to one,
\begin{equation} \label{m11}
    \| \tilde{\bm{L}}_{11,n}^{-1}\bm\zeta^\on \|_\infty \leq \|\bm{u_n}^{(1)}\|_\infty,
\end{equation}
provided $\Big\{\frac{C\sqrt{p_0}n^{a_2}}{\sqrt{n}} \|\bm{u_n}^{(1)}\| \leq 1 \Big\}$, with probability tending to one, which we later show is true.
Finally, by \eqref{m8}, we have with probability tending to one,
\begin{equation} \label{m12}
    \bigg\| \frac{\lambda_n}{\sqrt{n}} \tilde{\bm{L}}_{11,n}^{-1} \hat{\bm{s}}_n^{(1)} \bigg\|_\infty \leq C\frac{\lambda_n}{\sqrt{n}} n^{a_2}.
\end{equation}
Combining \eqref{m9}, \eqref{m11} and \eqref{m12}, we get $\bm{u}_{n}^\on = f(\bm{u}_{n}^\on)$ with \begin{equation*}
     \| f(\bm{u}_{n}^\on) \|_\infty \leq C \frac{\lambda_n}{\sqrt{n}}n^{a_2},
\end{equation*}
with probability tending to one. Therefore by Brouwer's fixed point theorem, \eqref{m4} has a solution $\tilde{\bm{u}}_n^{(1)}$ such that,
\begin{equation} \label{m13}
\left\| \tilde{\bm{u}}_n^{(1)} \right\|_\infty \leq C\dfrac{\lambda_{n}}{\sqrt{n}} n^{a_2},
\end{equation}
with probability tending to one. It follows that $\Big\{\frac{C\sqrt{p_0}n^{a_2}}{\sqrt{n}} \|\tilde{\bm{u}}_n^{(1)}\| \leq 1 \Big\}$, with probability tending to one, from assumption (A.4)(ii) and $\{\hat{s}_{jl} = s_{jl},\;l=0,\dots,p_{0j},\;j=1,\dots,(K-1)\},$ with probability tending to one, where $s_{jl} = sgn(\beta_{jl}),\;l=0,\dots,p_{0j},\;j=1,\dots,(K-1)$ due to \eqref{m13} and the beta-min assumption.

Next, to show \eqref{m5} holds, it is enough to show that for all $k_j = k = (p_{0j}+1),\dots,p \text{ and }j=1\dots,(K-1),$
\begin{multline} \label{m14}
    \left|\dfrac{1}{\sqrt{n}} \sum_{i=1}^{n} {\left(y_{ij}- \dfrac{n_j\exp\left(\dot{U}_i^{(j)\prime}\bm{\beta}_j^\on\right)}{n_K+\tilde{S}_i}\right)} \dot{U}_{ik}\right| + \left|{(\tilde{\bm{L}}_{21,n})}_{k.}^\top \tilde{\bm{L}}_{11,n}^{-1} \left(\bm\Lambda_n^{(1)} - \bm\zeta_n^{(1)}\right)\right| \\+
     \left|\frac{1}{2n^{3/2}}\sn\sum_{l=1}^{K-1}\sum_{l^{\prime}=1}^{K-1} h^{(l,l')}_{ij} \left(\dot{U}_i^{(l)\prime}\bm{u}_l^\on\right) \left(\dot{U}_i^{(l')\prime}\bm{u}_{l'}^\on\right) \dot{U}_{ik}\right|\\ + \left|\dfrac{\lambda_{n}}{\sqrt{n}} {(\tilde{\bm{L}}_{21,n})}_{k.}^\top \tilde{\bm{L}}_{11,n}^{-1} {{\hat{\bm{s}}}}_n^{(1)}\right| \leq \dfrac{\lambda_{n}}{\sqrt{n}},
\end{multline}
where, $\tilde{\bm{L}}_{21,n}$ is a $(p(K-1)-p_0)\times (p_0+K-1)$ block matrix. For $j,j' \in \{1,\dots,(K-1)\},$ let $\tilde{\bm{L}}_{21,n}^{(j,j')}$ be the $(j,j')$th block of $\tilde{\bm{L}}_{21,n}$ which is a $(p-p_{0l})\times (p_{0l^\top}+1)$ matrix defined as

\begin{align*}
& \tilde{\bm{L}}_{21,n}^{(j,j')} =\frac{1}{n} \sn g_i^{(j,j')} \dot{U}_i^{(j^c)} \dot{U}_i^{(j'){\prime}},
\end{align*}
where
\begin{align*}
    g_i^{(j,j')}=\begin{cases}
    \dfrac{n_j(n_K+\tilde{S}_{i,-j}) \exp\left(\dot{U}_i^{(j)\prime} \bm{\beta}_j^\on\right)}{(n_K+\tilde{S}_i)^2} \text{ if } j=j',\\
    -\dfrac{n_jn_{j'}\exp\left(\dot{U}_i^{(j)\prime} \bm{\beta}_j^\on\right)\exp\left(\dot{U}_i^{(j^{\prime})\prime}\bm{\beta}_{j'}^\on\right)}{(n_K+\tilde{S}_i)^2} \text{ if } j\neq j',
    \end{cases}
\end{align*}
and ${(\tilde{\bm{L}}_{21,n})}_{k_j.}^\top \equiv {(\tilde{\bm{L}}_{21,n})}_{k.}^\top$ are the rows of $\tilde{\bm{L}}_{21,n}.$
Define
\begin{align*}
    I & = \underset{1\leq j\leq (K-1)}{\max}\;\underset{(p_{0j}+1)\leq k\leq p}{\max}\left|\dfrac{1}{\sqrt{n}} \sum_{i=1}^{n} {\left(y_{ij}- \dfrac{n_j\exp\left(\dot{U}_i^{(j)\prime}\bm{\beta}_j^\on\right)}{n_K+\tilde{S}_i}\right)} \dot{U}_{ik}\right|,\\
    II & = \underset{1\leq j\leq (K-1)}{\max}\;\underset{(p_{0j}+1)\leq k\leq p}{\max}\left|{(\tilde{\bm{L}}_{21,n})}_{k.}^\top \tilde{\bm{L}}_{11,n}^{-1} \bm\Lambda_n^{(1)}\right|,\\
    III & = \underset{1\leq j\leq (K-1)}{\max}\;\underset{(p_{0j}+1)\leq k\leq p}{\max}\left|{(\tilde{\bm{L}}_{21,n})}_{k.}^\top \tilde{\bm{L}}_{11,n}^{-1} \bm\zeta_n^{(1)}\right|,\\
    IV &= \underset{1\leq j\leq (K-1)}{\max}\;\underset{(p_{0j}+1)\leq k\leq p}{\max}\left|\frac{1}{2n^{3/2}}\sn\sum_{l=1}^{K-1}\sum_{l^{\prime}=1}^{K-1} h^{(l,l')}_{ij} \left(\dot{U}_i^{(l)\prime}\bm{u}_l^\on\right) \left(\dot{U}_i^{(l')\prime}\bm{u}_{l'}^\on\right) \dot{U}_{ik}\right|,\\
    \text{and } V &= \underset{1\leq j\leq (K-1)}{\max}\;\underset{(p_{0j}+1)\leq k\leq p}{\max}\left| {(\tilde{\bm{L}}_{21,n})}_{k.}^\top \tilde{\bm{L}}_{11,n}^{-1} {{{\bm{s}}}}_n^{(1)}\right|.
\end{align*}
Therefore, following \eqref{m13} and similar arguments in \eqref{17}, we can rewrite \eqref{m14} as
\begin{equation} \label{m15}
    \dfrac{\lambda_{n}}{\sqrt{n}} (1-V) \geq I + II + III + IV.
\end{equation}
Before, proceeding to show \eqref{m15}, note that using similar arguments as in \eqref{19.1}, with probability going to 1, we get
\begin{align} \label{m16}
    &\underset{1\leq j\leq (K-1)}{\max}\;\underset{(p_{0j}+1)\leq k\leq p}{\max} \| {(\tilde{\bm{L}}_{21,n} - \mathbb{E}(\tilde{\bm{L}}_{21,n}))}_{k.} \|_1 \nonumber\\
    & \leq C p_0\; \underset{1\leq j,j' \leq (K-1)}{\max}\;\underset{(p_{0j'}+1)\leq k'\leq p}{\underset{(p_{0j}+1)\leq k\leq p}{\max}} \left| \frac{1}{n} \sn \left(g_i^{(j,j')}\dot{U}_{ik} \dot{U}_{ik^\top} - \mathbb{E} \left(g_i^{(j,j')}\dot{U}_{il} \dot{U}_{ik^\top}\right) \right)\right| \nonumber\\
    & \leq C p_0 \left[\sqrt{\frac{\log(np)}{n}} + \frac{(\log(n))^{2/\kappa}(\log(np))^{1/\min\{1,\kappa/2\}}}{n}\right].
\end{align}
Again, due to the calculations similar to as in \eqref{19.2} and \eqref{m7.1}, we have (with probability tending to one)
\begin{equation} \label{m16.1}
   \underset{1\leq j\leq (K-1)}{\max}\;\underset{(p_{0j}+1)\leq k\leq p}{\max} \Big\|\Big(\mathbb{E}(\tilde{\bm{L}}_{21,n}) - \mathbb{E}(\bm{L}_{21,n})\Big)_{k.} \Big\|_1 = o(n^{-1/2-a_2}).
\end{equation}
We have $\underset{1\leq j\leq (K-1)}{\max}\;\underset{(p_{0j}+1)\leq k\leq p}{\max}\|\mathbb{E}(\bm{L}_{21,n})_{k.}\| = C p_0,$ due to the sub-Weibull assumption.
Combining \eqref{m16} and \eqref{m16.1}, we get with probability tending to one
\begin{multline} \label{m16.2}
    \underset{1\leq j\leq (K-1)}{\max}\;\underset{(p_{0j}+1)\leq k\leq p}{\max} \Big\|\Big(\tilde{\bm{L}}_{21,n} - \mathbb{E}(\bm{L}_{21,n})\Big)_{k.} \Big\|_1 \\ \leq C p_0 \left[\sqrt{\frac{\log(np)}{n}} + \frac{(\log(n))^{2/\kappa}(\log(np))^{1/\min\{1,\kappa/2\}}}{n}\right]
\end{multline} 
and
\begin{equation*}
    \underset{1\leq j\leq (K-1)}{\max}\;\underset{(p_{0j}+1)\leq k\leq p}{\max}\|(\bm{L}_{21,n})_{k.}\| = Cp_0.
\end{equation*}
Then, \eqref{m15} follows using bounds in \eqref{m9}, \eqref{m10}, \eqref{m12}, \eqref{m13} and \eqref{m16.2} and following similar calculations as in \eqref{18}, \eqref{20}, \eqref{21}, \eqref{22}. Finally, the statement of the theorem is immediate. \end{proof}

\begin{lem} \label{pl:multi}
    (Multi-class Post LASSO Logistic) Recall the post-LASSO Logistic estimator $\hat{\bm{\beta}}_n =  \Big(\big(\hat{\bm{\beta}}_n^{\tilde{\mathcal{A}}_{n}}\big)^\top, \bm{0}^\top\Big) =\Big(\big(\hat{\bm{\beta}}_n^{\tilde{\mathcal{A}}_{n1}}\big)^\top, \dots, \big(\hat{\bm{\beta}}_n^{\tilde{\mathcal{A}}_{n(K-1)}}\big)^\top, \bm{0}^\top\Big)^\top$ defined in Section~3 of the main manuscript. Define $\hat{u}_n = \sqrt{n}(\hat{\bm{\beta}}_n - \bm\beta_n)$. Under assumptions (C.1)--(C.3), as $n\rightarrow \infty$, we have
    $$\mathbb{P}\Big(\big\|\bm{\hat{u}}_n^{\tilde{\mathcal{A}}_n}\big\|_\infty \leq C n^{a_2} \sqrt{\log{(np_0)}}\Big),$$
for some constant $C\geq 1$.    
\end{lem}
\begin{proof} The proof follows from the proof of Lemma~\ref{vscm}. \end{proof}

\begin{lem}\label{lem:asymnor:multi}
  (Normal approximation of the components of the centered version of $\bm{T}^{(2)}_{n}$)\\ Under assumptions (C.1)--(C.4), we have for any $k=k_j=0,1,\dots, p_{0j},\; j=1,\dots,(K-1)$
   \begin{equation*}
    \underset{x\in\mathbb{R}}{\sup} \;| \mathbbm{P}(T^{(2)}_{1nk} \leq  x) - {\Phi}(\sigma^{-1}_{nk} x)| = o(1),
\end{equation*} 
where $\bm{T}^{(2)}_{1nk} = \sqrt{n}\big(\hat{\bm{\beta}}_{nk}^\on - {\bm{\beta}}_{nk}^\on) + b_2(p,n) Z_{2k}$ and  $\sigma^2_{nk} = \Big((\mathbb{E}\bm{L}^{(j,j)}_{11, n})^{-1}\Big)_{kk}$.
\end{lem}
\begin{proof} The proof follows similar arguments as in the proof of Lemma~\ref{lem:asymnor}. \end{proof}

\begin{lem} \label{lem:minimax}
Let $\theta$ be the $p$-dimensional parameter of interest and $P_\theta$ denote the law of the data. Consider testing $H^{(0)}:  \theta = 0$ vs $H^{(1)}: \theta \in \mathcal{H}^{(1)} = \{ \bm b \in \mathbb{R}^p: \|\bm b\|_0 = p_0\; \text{and}\; \min_{j:b_{j} \neq 0}|b_{j}|\geq \rho\}$, for some $0< p_0 \leq p$ and $\rho > 0$. Let $\mu_\rho$ ($\equiv \mu_{p,p_0, \rho}$) be some probability measure on $\mathcal{H}^{(1)}$ and $\mathbbm{P}_{\mu_\rho}(A) = \int \mathbbm{P}_{\bm b}(A) \mu_\rho(\bm b)$ denote the posterior probability of any Borel set $A$. 
If $\Phi_{\alpha}$ is a level $\alpha$ test for testing $H^{(0)}$ vs $H^{(1)}$ then \[ \inf_{\Phi_\alpha} \sup_{ \theta\in\mathcal{H}^{(1)}} \mathbbm{P}_{\theta}(\Phi_\alpha = 0 ) \geq \inf_{\Phi_\alpha} \mathbbm{P}_{\mu_\rho}(\Phi_\alpha = 0) \geq 1-\alpha - \sqrt{\chi^2(\mathbbm{P}_{\mu_\rho}, \mathbbm{P}_0)},\]
where $\chi^2(\mathbbm{P}_{\mu_{\rho}}, \mathbbm{P}_0)$ denotes the chi squared divergence of $P_{\mu_{\rho}}$ from $P_0$.
\end{lem}

\begin{proof} This lemma essentially follows from the arguments given in Section 7.1 of \cite{Baraud2002minimax} and the fact that $2\text{TV}(\mathbbm{P}_{\mu_\rho}, \mathbbm{P}_0) \leq \sqrt{\chi^2(\mathbbm{P}_{\mu_\rho}, \mathbbm{P}_0)}$, where $\text{TV}(\mathbbm{P}_{\mu_\rho}, \mathbbm{P}_0)$ is the total variation distance between $\mathbbm{P}_{\mu_\rho}$ and $\mathbbm{P}_0$. A version of the lemma is stated as Lemma 5 in \cite{Ma2021}. \end{proof}

\section*{Appendix B: Proof of the main results} \label{sec: theorem}
\subsection{Proof of Proposition~\ref{prop:transformation}}
This statement is a special case of Proposition~\ref{prop:multitransformation}; hence, we omit the proof. \qed

\subsection{Proof of Theorem~\ref{theo:size}}
Recall that the Logistic Lasso estimator is defined as 
  \begin{align*}
    \tilde{\bm{\beta}}_n = \underset{(t_0, \bm{t}^\top)^\top \in \mathcal{R}^{(p+1)}}{\arg\min}& \Big[-\sum_{i=1}^n z_i (t_0 + W_i^\top \bm{t})\nonumber \\
    & + \sum_{i=1}^{n} \log (n_1+n_2\exp(t_0 + W_i^\top\bm{t})) +\lambda_n \sum_{i=0}^{p} |t_j|\Big]. 
\end{align*}
Now define, $\tilde{\bm{u}}_n = \sqrt{n} (\tilde{\bm{\beta}}_n - \bm{\beta}_n)$. Then by the change of variables for convex objective functions (see, e.g., \citet[Section~4.1.3]{boyd2004} for change of variables), we have
\begin{align}
\tilde{\bm{u}}_n 
&= \underset{\bm{u}\in \mathcal{R}^{(p+1)}}{\arg\min} \left[-\sum_{i=1}^{n}z_i\dot{W}_i^\top\frac{\bm{u}}{\sqrt{n}}+ \sum_{i=1}^n \log \left(n_1+n_2\exp\left(\dot{W}_i^\top\left(\frac{\bm{u}}{\sqrt{n}} + \bm{\beta_n}\right)\right)\right)\right.\nonumber\\
& \qquad \qquad +\left.\lambda_n\sum_{j=0}^{p}\left(\left|\frac{u_j}{\sqrt{n}} + \beta_{nj}\right|-|\beta_{nj}|\right)\right]. \label{eq2}
\end{align}
Here, \eqref{eq2} is equivalent to the KKT conditions
\begin{multline} \label{eq3}
    -\dfrac{1}{\sqrt{n}}\sum_{i=1}^{n}{z_{i}\dot{W}_{ij}+ \sum_{i=1}^{n} \dfrac{\dot{W}_{ij}}{\sqrt{n}}\dfrac{n_2\exp\left(\dot{W}_i^\top\left(\frac{\bm{u}}{\sqrt{n}} + \bm{\beta_n}\right)\right)}{n_1+n_2\exp\left(\dot{W}_i^\top\left(\frac{\bm{u}}{\sqrt{n}} + \bm{\beta_n}\right)\right)}}\\=-\dfrac{\lambda_{n}}{\sqrt{n}} sgn\left(\frac{u_j}{\sqrt{n}} + \beta_{nj}\right), \text{if } \frac{u_j}{\sqrt{n}} + \beta_{nj} \neq 0,
    \end{multline}
    \text{and  }
    \begin{multline} \label{eq4}
        -\frac{\lambda_{n}}{\sqrt{n}}\leq -\dfrac{1}{\sqrt{n}}\sum_{i=1}^{n}{z_{i}\dot{W}_{ij}+ \sum_{i=1}^{n} \dfrac{\dot{W}_{ij}}{\sqrt{n}}\dfrac{n_2\exp\left(\dot{W}_i^\top\left(\frac{\bm{u}}{\sqrt{n}} + \bm{\beta_n}\right)\right)}{n_1+n_2\exp\left(\dot{W}_i^\top\left(\frac{\bm{u}}{\sqrt{n}} + \bm{\beta_n}\right)\right)}}\\ \leq \dfrac{\lambda_{n}}{\sqrt{n}},\text{if } \frac{u_j}{\sqrt{n}} + \beta_{nj} = 0.
    \end{multline}
First, we would like to show that with probability goes to 1, $\tilde{\mathcal{A}}_n = \mathcal{A}_n = \emptyset.$ under $H_0^\prime$. This is equivalent to establishing $\frac{\tilde{\bm{u}}_n}{\sqrt{n}} + \bm\beta_{n} = 0$. Thus, due to equation (\ref{eq3}) and (\ref{eq4}), for all $j=0,1,\dots,p$, we need to establish
\begin{equation} \label{eq29}
    -\frac{\lambda_{n}}{\sqrt{n}} \leq \dfrac{1}{\sqrt{n}}\sum_{i=1}^{n}\left(z_{i}-\pi_2\right)\dot{W}_{ij} + \dfrac{1}{\sqrt{n}}\sum_{i=1}^{n}\left(\pi_2-\frac{n_2}{n}\right)\dot{W}_{ij} \leq \dfrac{\lambda_{n}}{\sqrt{n}},
\end{equation} 
Now we have, under $H_0^\prime$, with probability tending to one,
\begin{equation*}
     \underset{0\leq j\leq p}{\max} \left| \dfrac{1}{\sqrt{n}}\sum_{i=1}^{n}\left(z_{i}-\pi_2\right)\dot{W}_{ij} \right| \leq C \left[\sqrt{\log(np)} + \frac{(\log(n))^{1/\kappa}(\log(np))^{1/\min\{1,\kappa\}}}{\sqrt{n}} \right],
\end{equation*}
due to Lemma~\ref{kcvl2}. Again, due to the assumptions (A.3), and Lemma~\ref{kcvl2}, we have with probability tending to one,
\begin{equation*}
\begin{split}
    &\quad \underset{0\leq j \leq p}{\max}\bigg|\dfrac{1}{\sqrt{n}}\sn \left(\pi_2-\frac{n_2}{n}\right)\dot{W}_{ij}\bigg| \\
    & \leq C \bigg|\pi_2-\frac{n_2}{n}\bigg| \underset{0\leq j \leq p}{\max} \bigg\{\bigg|\dfrac{1}{\sqrt{n}} \sn \Big(\dot{W}_{ij} - \mathbb{E}(\dot{W}_{ij}) \Big)\bigg| + \bigg|\dfrac{1}{\sqrt{n}} \sn \mathbb{E}(\dot{W}_{ij})\bigg| \bigg\}\\
    &= o(1).
\end{split}
\end{equation*}
Therefore, equation (\ref{eq29}) is true due to the assumption (A.4)(i). Hence the VSC holds under $H_0'$

Now, we would like to show that the constructed test has asymptotic size $\alpha$. In that regard, note that
\begin{align*}
    &\quad \mathbb{E}_{H_0^\prime}\left(\psi_{n, \alpha}^{(1)}\right)= \mathbbm{P}_{H_0^\prime}\left( \underset{0\leq j \leq p}{\max}\; |T_{nj}| \leq m_{n,p,1-\alpha}^{(1)} \right)\\
    &= \mathbbm{P}_{H_0^\prime}\left( \Big\{\underset{0\leq j \leq p}{\max}\; |T_{nj}| \leq m_{n,p,1-\alpha}^{(1)}\Big\} \cap \Big\{\tilde{\mathcal{A}}_n = \phi\Big\} \right)\\&\hspace{3.5cm} + \mathbbm{P}_{H_0^\prime}\left( \Big\{\underset{0\leq j \leq p}{\max}\; |T_{nj}| \leq m_{n,p,1-\alpha}^{(1)}\Big\} \cap \Big\{\tilde{\mathcal{A}}_n \neq \phi\Big\} \right)\\
    &= \mathbbm{P}_{H_0^\prime}\left( \Big\{\underset{0\leq j \leq p}{\max}\; b_1(p,n)\;|Z_{1j}| \leq m_{n,p,1-\alpha}^{(1)}\Big\} \cap \Big\{\tilde{\mathcal{A}}_n = \phi\Big\} \right)\\&\hspace{1.5cm} + \mathbbm{P}_{H_0^\prime}\left( \Big\{\underset{0\leq j \leq p}{\max}\; b_1(p,n) \;|Z_{1j}| \leq m_{n,p,1-\alpha}^{(1)}\Big\} \cap \Big\{\tilde{\mathcal{A}}_n \neq \phi\Big\} \right) \\& \hspace{2cm} + \mathbbm{P}_{H_0^\prime}\left( \Big\{\underset{0\leq j \leq p}{\max}\; |T_{nj}| \leq m_{n,p,1-\alpha}^{(1)}\Big\} \cap \Big\{\tilde{\mathcal{A}}_n \neq \phi\Big\} \right)\\&\hspace{2.5cm} - \mathbbm{P}_{H_0^\prime}\left( \Big\{\underset{0\leq j \leq p}{\max}\; b_1(p,n) \;|Z_{1j}| \leq m_{n,p,1-\alpha}^{(1)}\Big\} \cap \Big\{\tilde{\mathcal{A}}_n \neq \phi\Big\} \right)\\
    &= \mathbbm{P}_{H_0^\prime}\left( \underset{0\leq j \leq p}{\max}\; b_1(p,n) \;|Z_{1j}| \leq m_{n,p,1-\alpha}^{(1)} \right)\\& \hspace{1.25cm} + \mathbbm{P}_{H_0^\prime}\left( \Big\{\underset{0\leq j \leq p}{\max}\; |T_{nj}| \leq m_{n,p,1-\alpha}^{(1)}\Big\} \cap \Big\{\tilde{\mathcal{A}}_n \neq \phi\Big\} \right)\\&\hspace{2.5cm} - \mathbbm{P}_{H_0^\prime}\left( \Big\{\underset{0\leq j \leq p}{\max}\; b_1(p,n) \;|Z_{1j}| \leq m_{n,p,1-\alpha}^{(1)}\Big\} \cap \Big\{\tilde{\mathcal{A}}_n \neq \phi\Big\} \right)\\
    \end{align*}
Therefore,
    \begin{align*}
    &\quad\; \bigg| \mathbbm{P}_{H_0^\prime}\left( \underset{0\leq j \leq p}{\max}\; |T_{nj}| > m_{n,p,1-\alpha}^{(1)} \right) - \alpha \bigg|\\
    &= \bigg| \mathbbm{P}_{H_0^\prime}\left( \underset{0\leq j \leq p}{\max}\; |T_{nj}| \leq m_{n,p,1-\alpha}^{(1)} \right) - (1-\alpha) \bigg|\\
    &= \bigg| \mathbbm{P}_{H_0^\prime}\left( \Big\{\underset{0\leq j \leq p}{\max}\; |T_{nj}| \leq m_{n,p,1-\alpha}^{(1)}\Big\} \cap \Big\{\tilde{\mathcal{A}}_n \neq \phi\Big\} \right)\\&\hspace{2.5cm} - \mathbbm{P}_{H_0^\prime}\left( \Big\{\underset{0\leq j \leq p}{\max}\; b_1(p,n) \;|Z_{1j}| \leq m_{n,p,1-\alpha}^{(1)}\Big\} \cap \Big\{\tilde{\mathcal{A}}_n \neq \phi\Big\} \right) \bigg|\\
    & \leq 2\mathbbm{P}_{H_0^\prime}\left(\Big\{\tilde{\mathcal{A}}_n \neq \phi\Big\}\right).
\end{align*}
Hence, due to VSC, we get
\begin{equation*}
    \left|\mathbb{E}_{H_0^\prime}\left(\psi_{n, \alpha}^{(1)}\right) - \alpha\right| 
    = o(1).
\end{equation*} 
This completes the proof of Theorem~\ref{theo:size}.
\qed

\subsection{Proof of Theorem~\ref{theo:power}}
The test statistic can be written as $$ \bm{T}^\on_n = \sqrt{n}\hat{\bm{\beta}}_n + b_1(p,n)\bm{Z}_{1} = \bm{T}^\on_{1n} + \sqrt{n} {\bm\beta}_n,$$ where under VSC, we have $\bm{T}^\on_{1n} =  \sqrt{n}\Big((\hat{\bm{\beta}}_n^{(1)} - {\bm{\beta}}_n^{(1)})^\top, \bm{0}^\top \Big) +  b_1(p,n)\bm{Z}_1 $ and ${\bm\beta}_n = \Big({{\bm{\beta}}_n^{(1)}}^\top, \bm{0}^\top\Big)^\top.$ 
Now, by Lemma~\ref{lem:asymnor}, we have for any $j=0,\dots, p_0$,
   \begin{equation*}
    \underset{x\in\mathbb{R}}{\sup} \;| \mathbbm{P}(T_{1nj} \leq  x) - {\Phi}(\sigma^{-1}_{nj} x)| = o(1),
\end{equation*} 
where $\sigma^2_{nj} = Var(V_{nij}) = ((\mathbb{E}\bm{L}_{11, n})^{-1})_{j j}$. Note that, 
\begin{equation*}
    \sigma^2_{nj} \leq \|(\mathbb{E}\bm{L}_{11, n})^{-1}\|_\infty \leq C n^{a_1},
\end{equation*}
and due to Lemma 6 of \cite{tony2014two}, $m_{n,p,1-\alpha}^{(1)} \leq C \sqrt{\log(p) - \log\log(p)},$ as $b_1(p,n) = o(1)$. Therefore, for some $j' = 0,\dots,p,$ we have
\begin{align*}
    & \quad \mathbb{E}_{H_1^\prime}\Big(\psi_{n, \alpha}^{(1)}\Big)  =\mathbbm{P}_{H_1^\prime} \Big( \underset{0\leq j \leq p}{\max} |T_{nj}| > m_{n,p,1-\alpha}^{(1)} \Big)\\
    & = 1 - \mathbbm{P}_{H_1^\prime} \Big( \underset{0\leq j\leq p}{\max} |T_{nj}| \leq m_{n,p,1-\alpha}^{(1)} \Big)\\
    & \geq 1 - \mathbbm{P}_{H_1^\prime} \Big( \Big\{\underset{0\leq j\leq p}{\max} |T_{nj}| \leq m_{n,p,1-\alpha}^{(1)}\Big\} \cap \Big\{\tilde{\mathcal{A}}_n = \{0,\dots,p_0\}\Big\} \Big) \\&\hspace{7cm} - \mathbbm{P}_{H_1^\prime} \Big(\tilde{\mathcal{A}}_n = \{0,\dots,p_0\}^c \Big)\\
    & \geq 1 - \mathbbm{P}_{H_1^\prime} \Big( \Big\{\underset{0\leq j\leq p}{\max} |T_{nj}| \leq m_{n,p,1-\alpha}^{(1)}\Big\} \cap \Big\{\tilde{\mathcal{A}}_n = \{0,\dots,p_0\}\Big\} \Big) - o(1)\\
    & \geq 1 - \mathbbm{P}_{H_1^\prime} \Big( \Big\{ |T_{nj'}| \leq m_{n,p,1-\alpha}^{(1)}\Big\} \cap \Big\{\tilde{\mathcal{A}}_n = \{0,\dots,p_0\}\Big\} \Big) - o(1)\\
    & \geq 1 - \mathbbm{P}_{H_1^\prime} \Big( \Big\{ |T_{1nj'}| \geq \sqrt{n} |\beta_{nj'}| - m_{n,p,1-\alpha}^{(1)}\Big\} \cap \Big\{\tilde{\mathcal{A}}_n = \{0,\dots,p_0\}\Big\} \Big) - o(1)\\
    & \geq 1 - \Big| \mathbbm{P}_{H_1^\prime} \Big( T_{1nj'} \geq \sqrt{n} |\beta_{nj'}| - m_{n,p,1-\alpha}^{(1)} \Big) \\&\hspace{4cm} - (1-\Phi({\sigma_{nj'}^{-1}}[\sqrt{n} |\beta_{nj'}| - m_{n,p,1-\alpha}^{(1)}])) \Big|\\
    & \hspace{2cm} -\Big| \mathbbm{P}_{H_1^\prime} \Big( T_{1nj'} \leq -  \sqrt{n} |\beta_{nj'}| + m_{n,p,1-\alpha}^{(1)} \Big)\\&\hspace{4cm} - \Phi( {\sigma_{nj'}^{-1}}[- \sqrt{n}|\beta_{nj'}| + m_{n,p,1-\alpha}^{(1)}]) \Big| \\&\hspace{2cm} - \Big( 1-\Phi({\sigma_{nj'}^{-1}}[\sqrt{n}  |\beta_{nj'}| - m_{n,p,1-\alpha}^{(1)}]) \Big) \\&\hspace{4cm} - \Phi( {\sigma_{nj'}^{-1}}[- \sqrt{n}|\beta_{nj'}| + m_{n,p,1-\alpha}^{(1)}])- o(1)\\
    &\geq  1- 
     2\Phi( {\sigma_{nj'}^{-1}}[- \sqrt{n}|\beta_{nj'}| + m_{n,p,1-\alpha}^{(1)}])- o(1)\\
    &\geq 1 - o(1).
\end{align*} 
The second inequality is due to VSC. The fourth inequality is due to the triangle inequality. The sixth inequality is due to Lemma 2.5 in the supplementary condition, and the last inequality is due to the assumption of the beta-min condition.
  \qed


\subsection{Proof of Theorem~\ref{theo:minimax}}
Recall that the $(\alpha, \delta)$-minimax separation distance for testing $H_0: \mu_1 = \mu_2$ vs $H_1: \mu = (\mu_1^\top, \mu_2^\top)^\top \in \Theta_1(p_0)$ is defined as $$\rho^* \equiv \rho^*(\alpha, \delta) = \inf_{\phi_\alpha} \rho(\phi_\alpha, \delta),$$ where
\begin{align*}
\rho(\phi_\alpha, \delta) = \inf\Big\{ \rho>0: \sup_{\mu\in \Theta_1(p_0): \min_{j: \beta_j \neq 0} |\beta_j| \ge \rho}  \mathbbm{P}_{\bm \beta} (\phi_\alpha=0) \le \delta  \Big\}.
\end{align*}
Therefore, note that $\rho^*\ge \rho_1$ if for all $\rho < \rho_1,$
\begin{equation} \label{minimax1}
    \inf_{\phi_\alpha}\bigg[\sup_{\mu\in \Theta_1(p_0): \mu_1 = 0\; \& \min_{j: \beta_j \neq 0} |\beta_j| \ge \rho}  \mathbbm{P}_{\bm \beta} (\phi_\alpha=0)\bigg] > \delta. 
\end{equation}
Now note that 
\begin{itemize}
    \item[(a)] Under Gaussianity and $\mu_1 = 0,$ the relation between $\mu$ and $\bm \beta = (\beta_0, {\bm\beta^{(1)}}^\top)^\top$ (stated as equation (2) in the main manuscript) implies that $H_{0}: \mu_1 = \mu_2 \iff H_{0}^\prime: {\bm\beta^{(1)}} = 0.$
    \item[(b)] Under Gaussianity with $\mu_1 = 0,$ $\beta_0 = -\frac{1}{2} \mu_2'\Sigma^{-1}\mu_2  = -\frac{1}{2} {\bm\beta^{(1)}}^\top\Sigma{\bm\beta^{(1)}}$. Hence when $\|{\bm\beta^{(1)}}\|_0 = p_0$ and $\min_{j:\beta_{1j}\neq 0} |\beta_{1j}| \ge \rho$ then $|\beta_0| \ge \Sigma_{\min} \|{\bm\beta^{(1)}}\|_2^2 \ge \Sigma_{\min}\;\rho \; p_0 \ge \rho$, since $\Sigma_{\min}$, the minimum eigen value of $\Sigma$, is more than $p_0^{-1}$. Therefore, $|\beta_0|\geq \rho$ whenever ${\bm\beta^{(1)}}\in\mathcal{H}^{(1)} = \{\bm{b}\in \mathcal{R}^p: \|\bm{b}\|_0 = p_0\; \text{and}\; \min_{j: b_j \neq 0}|b_j|\geq \rho\}.$
\end{itemize}
Due to (a){}--{}(b) and Lemma \ref{lem:minimax}, equation (\ref{minimax1})will be established if we can show that
\begin{equation} \label{minimax10}
 \chi^2(\mathbbm{P}_{\mu_\rho}, \mathbbm{P}_0) < (1-\alpha-\delta)^2,\; \text{ for all } \rho < c\sqrt{\log p/n},
\end{equation} 
 for some measure $\mu_\rho$ on $\mathcal{H}^{(1)}$, for sufficiently large $n$ and $p$. Now note that for any ${\bm\beta^{(1)}} \in \mathcal{H}^{(1)}$, due to Gaussianity with $\mu_1 = 0$, the joint density of two samples is
\begin{align} \label{minimax3}
   \mathbbm{P}_{{\bm\beta^{(1)}}} &= \prod_{i=1}^{n}\mathbbm{P}_{{\bm\beta^{(1)}}}(x_i,y_i) \\& = \prod_{i=1}^{n}(2\pi)^{-p}|\Sigma|^{-1} \exp\Big\{ -\frac{1}{2} x_i^\top\Sigma^{-1}x_i - \frac{1}{2} (y_i-\Sigma{\bm\beta^{(1)}})^\top\Sigma^{-1}(y_i-\Sigma{\bm\beta^{(1)}}) \Big\}.
\end{align}

We now follow the arguments of the proof of Theorem~3 of \cite{tony2014two} and Theorem~2 of \cite{Ma2021} to establish (\ref{minimax10}). As in the proof of Theorem~2 of \cite{Ma2021}, consider $\mu_\rho$ to be the probability measure on $\mathcal{H}_1 = \{ {\bm\beta^{(1)}} \in \mathbb{R}^p: \|{\bm\beta^{(1)}}\|_0 = p_0 \text{ and } \min_{j:\beta_{1j} \neq 0}|\beta_{1j}|\geq \rho\}$ such that $\mu_\rho$ gives mass uniformly on all the $p_0$-dimensional sub-vectors of ${\bm\beta^{(1)}}$ with each component being $\rho$. More precisely if $\ell(\{1,\cdots,p\}, p_0)$ denotes all the subsets of $\{1,\cdots,p\}$ having cardinality $p_0$ and if $\tilde{\mathcal{H}}^{(1)} = \{ {\bm\beta^{(1)}} \in \mathbb{R}^p: \beta_{1j}=\rho\;\mathbb{I}(j\in \text{I}),\;\text{I}\in\ell(\{1,\cdots,p\}, p_0)\},$ then $\mu_\rho$ is uniformly distributed over $\tilde{\mathcal{H}}^{(1)}.$
Now recall that
\begin{equation*}
    \chi^2(\mathbbm{P}_{\mu_\rho}, \mathbbm{P}_0) = \int \left(\frac{\mathbbm{P}_{\mu_\rho}}{\mathbbm{P}_0}\right)^2 \mathbbm{P}_0 d\lambda - 1,
\end{equation*}
where $$\mathbbm{P}_0 = \prod_{i=1}^n (2\pi)^{-p}|\Sigma|^{-1} \exp\Big\{ -\frac{1}{2} x_i^\top\Sigma^{-1}x_i - \frac{1}{2} y_i^\top\Sigma^{-1}y \Big\} = \prod_{i=1}^n \mathbbm{P}_0(x_i,y_i),$$ and $$\mathbbm{P}_{\mu_\rho} = \frac{1}{\binom{p}{p_0}} \sum_{{\bm\beta^{(1)}}\in\tilde{\mathcal{H}}^{(1)}} \prod_{i=1}^n \mathbbm{P}_{{\bm\beta^{(1)}}}(x_i,y_i).$$ 
Therefore, 
\begin{equation*}
\begin{split}
    \chi^2(\mathbbm{P}_{\mu_\rho}, \mathbbm{P}_0) 
    = \frac{1}{\binom{p}{p_0}^2} \sum_{{\bm\beta^{(1)}}\in\tilde{\mathcal{H}}_1} \sum_{\dot{{\bm\beta^{(1)}}}\in\tilde{\mathcal{H}}_1} \prod_{i=1}^n  \int \frac{ \mathbbm{P}_{{\bm\beta^{(1)}}}(x_i,y_i) \mathbbm{P}_{\dot{{\bm\beta^{(1)}}}}(x_i,y_i) }{\mathbbm{P}_0(x_i,y_i)}  d\lambda(x_i,y_i) - 1.
\end{split}
\end{equation*}
Now,
\begin{equation*}
    \frac{ \mathbbm{P}_{{\bm\beta^{(1)}}}(x_i,y_i) \mathbbm{P}_{\dot{{\bm\beta^{(1)}}}}(x_i,y_i) }{(\mathbbm{P}_0(x_i,y_i))^2} = \exp\Big\{-\frac{1}{2} {\bm\beta^{(1)}}^\top \Sigma {\bm\beta^{(1)}} -\frac{1}{2} \dot{{\bm\beta^{(1)}}}^\top \Sigma \dot{{\bm\beta^{(1)}}} + {({\bm\beta^{(1)}}+\dot{{\bm\beta^{(1)}}})}^\top y_i \Big\}.
\end{equation*}
Therefore,
\begin{equation*}
    \begin{split}
        & \int \frac{ \mathbbm{P}_{{\bm\beta^{(1)}}}(x_i,y_i) \mathbbm{P}_{\dot{{\bm\beta^{(1)}}}}(x_i,y_i) }{(\mathbbm{P}_0(x_i,y_i))^2} d\lambda(x,y)\\
        & = \exp\Big\{-\frac{1}{2} \bm\beta^\top \Sigma {\bm\beta^{(1)}} -\frac{1}{2} \dot{\bm\beta}^\top \Sigma \dot{{\bm\beta^{(1)}}}\Big\} \mathbb{E}_{Z\sim N(0,\Sigma)}(\exp\{{({\bm\beta^{(1)}}+\dot{{\bm\beta^{(1)}}})}^\top Z\}) \\
        & = \exp\{{\bm\beta^{(1)}}^\top\Sigma\dot{{\bm\beta^{(1)}}}\}.
    \end{split}
\end{equation*}

Therefore, whenever $\rho < c\sqrt{\log p/n}$,
\begin{equation*}
\begin{split}
    \chi^2(\mathbbm{P}_{\mu_\rho}, \mathbbm{P}_0)
    & = \frac{1}{\binom{p}{p_0}^2} \sum_{{\bm\beta^{(1)}}\in\tilde{\mathcal{H}}_1} \sum_{\dot{{\bm\beta^{(1)}}}\in\tilde{\mathcal{H}}_1} \exp\{n{\bm\beta^{(1)}}^\top\Sigma\dot{{\bm\beta^{(1)}}}\} - 1\\
    & = \frac{1}{\binom{p}{p_0}^2} \sum_{\text{I}\in\ell(\{1,\cdots,p\})} \sum_{\dot{\text{I}}\in\ell(\{1,\cdots,p\}, p_0)} \prod_{k\in\text{I}, l\in\dot{\text{I}}} \exp{(n\rho^2 \sigma_{kl})} - 1\\
    & \le \frac{1}{\binom{p}{p_0}^2} \sum_{\text{I}\in\ell(\{1,\cdots,p\})} \sum_{\dot{\text{I}}\in\ell(\{1,\cdots,p\}, p_0)} \prod_{k\in\text{I}, l\in\dot{\text{I}}} \exp{(c^2\log p |\sigma_{kl}|)} - 1,\\
\end{split}
\end{equation*}
where $\Sigma = ((\sigma_{k,l}))_{p\times p}.$ Now following \cite{tony2014two}, for every $\text{I} \in \ell(\{1,\cdots,p\}, p_0)$, define $B_{\text{I}} = \{l: |\sigma_{kl}|\geq M/d, k \in \text{I}\}$ where $d = (p/p_0^2)^{1-\iota}$ with $1>\iota \geq 2Mc^2$ (since $c$ can be taken sufficiently small) and the positive constant $M$ is defined in the statement of Theorem~\ref{theo:minimax}. To that end, due to assumption $\| \Sigma \|_{L_1} \leq M$, the calculations done at page 369-370 of \cite{tony2014two} imply that
\begin{align}\label{minimax4}
    \chi^2(\mathbbm{P}_{\mu_\rho}, \mathbbm{P}_0) \leq o(1)\exp\Big\{\dfrac{dp_0^2p^{Mc^2}}{p} + \dfrac{Mp_0^2\log p}{d}\Big\},
\end{align}
where $o(1)$ is as $n, p \rightarrow \infty$. Since $p_0 \leq Mp^{1/4}$, (\ref{minimax4}) implies (\ref{minimax10}) and the proof is complete.\qed

\subsection{Proof of Proposition~\ref{prop:multitransformation}}
Recall that $\big\{X^{(j)}_1,\dots,X^{(j)}_{n_j}\big\}_{j=1}^{K}$ \big(with $n = \sum_{j=1}^{K}n_j$\big) are observed samples where for all $j \in \{1,\dots, K\}$, $X^{(j)},X^{(j)}_1,\dots,X^{(j)}_{n_j} \overset{iid}{\sim} F_j$, for some cumulative distribution function $F_j$ with mean $\mu_j$. We would like to test $H_{00}: \mu_1=\mu_2=\dots=\mu_K$ vs $H_{11}: \mu_i \neq \mu_j \;\text{for some}\; i \neq j$. Also recall that for any $j=1,2,\dots,K$, $U$ is defined as a random vector such that $\mathbb{P}(U = X^{(j)}) = \pi_j,$, with $\sum_{j=1}^{K}\pi_j = 1$. Moreover, $\bm{y}=(y_1\dots,y_{K-1})^\top$ is a $(K-1)$ dimensional random vector such that $y_{j}$ is 1 and rests are 0 when $U = X^{(j)}$ (for $j=1,2,\dots,K-1$), and $\bm{y} = \bm{0}$ for $U = X^{(K)}$. Then the regression parameter $\bm\beta = ({\bm\beta^{(1)}}^\top,\bm\beta_2^\top,\dots,\bm\beta_{K-1}^\top)^\top$ is defined as the solution of
\begin{equation}
    \mathbb{E}\left[ \left(y_j - \frac{ \pi_j\exp(\beta_{0j}+\bm\beta_{1j}^\top U) }{\pi_K+\sum_{l=1}^{K-1}\pi_l\exp(\beta_{0l}+\bm\beta_{1l}^\top U)}\right)  \begin{pmatrix} 1 \\ U \end{pmatrix} \right] = 0,\;j=1,\dots,K-1, \label{me1}
\end{equation}
where $\bm\beta_j = (\beta_{0j},\bm\beta_{1j}^\top)^\top$. We would like to establish that ``$\mu_1=\mu_2= \cdots = \mu_K$'' is equivalent to ``$\bm{\beta} = \bm{0}$''. First let us assume that $\bm{\beta} = \bm{0}$. Then from equation (\ref{me1}) we have for any $j=1,2,\dots,K-1$,
\begin{align*}
    &\quad\;\; \mathbb{E}\left[ \left(y_j - \pi_j \right) \begin{pmatrix} 1 \\ U \end{pmatrix} \right] = \bm{0}\\
    &\Rightarrow \sum_{l=1}^K \mathbb{E}\left[ \left.\left(y_j - \pi_j \right) \begin{pmatrix} 1 \\ U \end{pmatrix} \right| U = \bm{X}^{(l)} \right]\;\mathbb{P}(U = \bm{X}^{(l)}) = \bm{0}\\
    &\Rightarrow \pi_j\mathbb{E}\left[ \left(1 - \pi_j \right) \begin{pmatrix} 1 \\ \bm{X}^{(j)} \end{pmatrix}\right] + \sum_{l\neq j} \pi_l\mathbb{E}\left[ \left(- \pi_j \right) \begin{pmatrix} 1 \\ \bm{X}^{(l)} \end{pmatrix}\right] = \bm{0}\\
    &\Rightarrow \mathbb{E}\left[ \left(1 - \pi_j \right)\bm{X}^{(j)}\right] + \sum_{l\neq j} \pi_l\mathbb{E}\left[ - \bm{X}^{(l)}\right] = \bm{0}\\
    &\Rightarrow \mu_j=\pi_1\mu_1+\pi_2\mu_2+\dots+\pi_K\mu_K,
\end{align*}
which implies $\mu_1 = \mu_2 = \dots = \mu_K$.
Now assume that $\mu_1 = \mu_2 = \dots = \mu_K$. 
Then from equation (\ref{me1}) we have for any $j=1,2,\dots,K-1$,
\begin{align}
    & \sum_{i=1}^K \pi_i\mathbb{E}\left[ \left.\left(y_j - \frac{ \pi_j\exp(\beta_{0j}+\bm\beta_{1j}^\top U) }{\pi_K+\sum_{l=1}^{K-1}\pi_l\exp(\beta_{0l}+\bm\beta_{1l}^\top U)} \right) \begin{pmatrix} 1 \\ U \end{pmatrix} \right| U = \bm{X}^{(i)} \right] = \bm{0} \nonumber\\
    &\Rightarrow \pi_j\mathbb{E}\left[ \left(1-\frac{ \pi_j\exp(\beta_{0j}+\bm\beta_{1j}^\top\bm{X}^{(j)}) }{\pi_K+\sum_{l=1}^{K-1}\pi_l\exp(\beta_{0l}+\bm\beta_{1l}^\top\bm{X}^{(j)})}\right) \begin{pmatrix} 1 \\ \bm{X}^{(j)} \end{pmatrix}\right] \nonumber\\&\qquad + \sum_{i\neq j}\pi_i\mathbb{E} \left[ \left(-\frac{ \pi_j\exp(\beta_{0j}+\bm\beta_{1j}^\top\bm{X}^{(i)}) }{\pi_K+\sum_{l=1}^{K-1}\pi_l\exp(\beta_{0l}+\bm\beta_{1l}^\top\bm{X}^{(i)})}\right) \begin{pmatrix} 1 \\ \bm{X}^{(i)} \end{pmatrix}\right] = \bm{0} \nonumber\\
    &\Rightarrow \mathbb{E}\left[\begin{pmatrix} 1 \\ \bm{X}^{(1)} \end{pmatrix}\right] = \sum_{i=1}^K \pi_i \mathbb{E} \left[ \left(\frac{\exp(\beta_{0j}+\bm\beta_{1j}^\top\bm{X}^{(i)}) }{\pi_K+\sum_{l=1}^{K-1}\pi_l\exp(\beta_{0l}+\bm\beta_{1l}^\top\bm{X}^{(i)})}\right) \begin{pmatrix} 1 \\ \bm{X}^{(i)} \end{pmatrix}\right] \label{me2}
\end{align}
Now, using the identity $\mu_1 = \mu_2 = \dots = \mu_K$ together with multiplying both side of the equation (\ref{me2}) by $\pi_j$ and then summing over $j=1,\dots, (K-1)$ we have
\begin{align}
    &(1-\pi_K)\mathbb{E}\left[\begin{pmatrix} 1 \\ \bm{X}^{(1)} \end{pmatrix}\right] = \sum_{i=1}^K \pi_i \mathbb{E} \left[ \left(\frac{\sum_{j=1}^{K-1}\pi_j \exp(\beta_{0j}+\bm\beta_{1j}^\top\bm{X}^{(i)}) }{\pi_K+\sum_{l=1}^{K-1}\pi_l\exp(\beta_{0l}+\bm\beta_{1l}^\top\bm{X}^{(i)})}\right) \begin{pmatrix} 1 \\ \bm{X}^{(i)} \end{pmatrix}\right] \nonumber \\
    &\Rightarrow (1-\pi_K) \mathbb{E}\left[\begin{pmatrix} 1 \\ \bm{X}^{(1)} \end{pmatrix}\right] = \sum_{i=1}^K \pi_i \mathbb{E} \left[ \left( 1 - \frac{\pi_K}{\pi_K+\sum_{l=1}^{K-1}\pi_l\exp(\beta_{0l}+\bm\beta_{1l}^\top\bm{X}^{(i)})}\right) \begin{pmatrix} 1 \\ \bm{X}^{(i)} \end{pmatrix}\right] \nonumber \\
    &\Rightarrow \mathbb{E}\left[\begin{pmatrix} 1 \\ \bm{X}^{(1)} \end{pmatrix}\right] = \sum_{i=1}^K \pi_i \mathbb{E} \left[ \left(\frac{1}{\pi_K+\sum_{l=1}^{K-1}\pi_l\exp(\beta_{0l}+\bm\beta_{1l}^\top\bm{X}^{(i)})}\right) \begin{pmatrix} 1 \\ \bm{X}^{(i)} \end{pmatrix}\right] \label{me3}
\end{align}
Subtracting \eqref{me2} from \eqref{me3}, we get
\begin{align}
    & \sum_{i=1}^{K}\pi_i\mathbb{E} \left[ \left(\frac{1-\exp(\beta_{0j}+\bm\beta_{1j}^\top\bm{X}^{(i)}) }{\pi_K+\sum_{l=1}^{K-1}\pi_l\exp(\beta_{0l}+\bm\beta_{1l}^\top\bm{X}^{(i)})}\right) \begin{pmatrix} 1 \\ \bm{X}^{(i)} \end{pmatrix}\right] = \bm{0} \nonumber \\
    & \Rightarrow \mathbb{E} \left[ \left(\frac{1-\exp(\beta_{0j}+\bm\beta_{1j}^\top U) }{\pi_K+\sum_{l=1}^{K-1}\pi_l\exp(\beta_{0l}+\bm\beta_{1l}^\top U)}\right) \begin{pmatrix} 1 \\  U \end{pmatrix}\right] = \bm{0}. \label{me4}
\end{align}
Clearly if $\bm\beta_{1j}=\bm{0}$ then $\beta_{0j}=0$ from (\ref{me4}) for any $j\in \{1,2,\dots,(K-1)\}$ and then we are done. Hence  let us fix $j \in \{1,2,\dots,K-1\}$ and assume that $\bm\beta_{1j}\neq\bm{0}$, if possible. Then there are two possibilities:  $\beta_{0j}\geq 0$ and $\beta_{0j}<0$. First, Let us assume $\beta_{0j} \geq 0$. Then from equation (\ref{me4}) we~have
\begin{align}
    &\mathbb{E} \left[ \frac{1-\exp(\beta_{0j}+\bm\beta_{1j}^\top U) }{\pi_K+\sum_{l=1}^{K-1}\pi_l\exp(\beta_{0l}+\bm\beta_{1l}^\top U)} \right] = 0\nonumber\\
    \Rightarrow\;& \mathbb{E} \left\{\left[ \frac{\exp(\beta_{0j}+\bm\beta_{1j}^\top U) -1 }{\pi_K+\sum_{l=1}^{K-1}\pi_l\exp(\beta_{0l}+\bm\beta_{1l}^\top U)} \right]\mathbb{I}\Big(\bm\beta_{1j}^\top U > -\beta_{0j}\Big)\right\}\nonumber\\& = \mathbb{E} \left\{\left[ \frac{1-\exp(\beta_{0j}+\bm\beta_{1j}^\top U) }{\pi_K+\sum_{l=1}^{K-1}\pi_l\exp(\beta_{0l}+\bm\beta_{1l}^\top U)} \right]\mathbb{I}\Big(\bm\beta_{1j}^\top U< -\beta_{0j}\Big)\right\}, \label{me5}
\end{align}
and
\begin{align} \label{me6}
    & \mathbb{E} \left[ \left(\frac{1-\exp(\beta_{0j}+\bm\beta_{1j}^\top U) }{\pi_K+\sum_{l=1}^{K-1}\pi_l\exp(\beta_{0l}+\bm\beta_{1l}^\top U)}\right) \bm\beta_{1j}^\top U  \right] = 0\nonumber\\
    \Rightarrow\; & -\mathbb{E} \left[ \left(\frac{1-\exp(\beta_{0j}+\bm\beta_{1j}^\top U) }{\pi_K+\sum_{l=1}^{K-1}\pi_l\exp(\beta_{0l}+\bm\beta_{1l}^\top U)}\right) \bm\beta_{1j}^\top U \mathbb{I}\Big(\bm\beta_{1j}^\top U< -\beta_{0j}\Big) \right]\nonumber\\& - \mathbb{E} \left[ \left(\frac{1-\exp(\beta_{0j}+\bm\beta_{1j}^\top U) }{\pi_K+\sum_{l=1}^{K-1}\pi_l\exp(\beta_{0l}+\bm\beta_{1l}^\top U)}\right) \bm\beta_{1j}^\top U \mathbb{I}\Big(\bm\beta_{1j}^\top U> 0\Big) \right]\nonumber\\  & =\mathbb{E} \left[ \left(\frac{1-\exp(\beta_{0j}+\bm\beta_{1j}^\top U) }{\pi_K+\sum_{l=1}^{K-1}\pi_l\exp(\beta_{0l}+\bm\beta_{1l}^\top U)}\right) \bm\beta_{1j}^\top U \mathbb{I}\Big(-\beta_{0j} < \bm\beta_{1j}^\top U < 0\Big) \right].
    \end{align}
Now since the distribution of $ U$ is not degenerate, i.e. does not sit on the lower dimensional space, $\mathbb{P}(\bm\beta_{1j}^\top U = -\beta_{0j})< 1$. This along with (\ref{me5}) imply that $\mathbb{P}(\bm\beta_{1j}^\top U < -\beta_{0j}), \mathbb{P}(\bm\beta_{1j}^\top U > -\beta_{0j}) > 0$. This in turn implies that 
\begin{align}\label{me7}
   &\mathbb{E} \left[ \left(\frac{1-\exp(\beta_{0j}+\bm\beta_{1j}^\top U) }{\pi_K+\sum_{l=1}^{K-1}\pi_l\exp(\beta_{0l}+\bm\beta_{1l}^\top U)}\right) \bm\beta_{1j}^\top U \mathbb{I}\Big(\bm\beta_{1j}^\top U< -\beta_{0j}\Big) \right]\nonumber\\
    & < -\beta_{0j}\mathbb{E} \left[ \left(\frac{1-\exp(\beta_{0j}+\bm\beta_{1j}^\top U) }{\pi_K+\sum_{l=1}^{K-1}\pi_l\exp(\beta_{0l}+\bm\beta_{1l}^\top U)}\right)  \mathbb{I}\Big(\bm\beta_{1j}^\top U< -\beta_{0j}\Big) \right].
\end{align}
Now, it is clear from (\ref{me4})-(\ref{me7}) that
    \begin{align}
    &\beta_{0j}\;\mathbb{E} \left[ \left(\frac{\exp(\beta_{0j}+\bm\beta_{1j}^\top U) - 1 }{\pi_K+\sum_{l=1}^{K-1}\pi_l\exp(\beta_{0l}+\bm\beta_{1l}^\top U)}\right) \mathbb{I}\Big(\bm\beta_{1j}^\top U > -\beta_{0j}\Big) \right]\nonumber\\&+ \mathbb{E} \left[ \left(\frac{\exp(\beta_{0j}+\bm\beta_{1j}^\top U) - 1 }{\pi_K+\sum_{l=1}^{K-1}\pi_l\exp(\beta_{0l}+\bm\beta_{1l}^\top U)}\right) \bm\beta_{1j}^\top U \mathbb{I}\Big(\bm\beta_{1j}^\top U> 0\Big) \right]\nonumber\\ &\ < \;\beta_{0j}\;\mathbb{E} \left[ \left(\frac{\exp(\beta_{0j}+\bm\beta_{1j}^\top U) - 1 }{\pi_K+\sum_{l=1}^{K-1}\pi_l\exp(\beta_{0l}+\bm\beta_{1l}^\top U)}\right) \mathbb{I}\Big(-\beta_{0j} < \bm\beta_{1j}^\top U < 0\Big) \right]\nonumber\\
    \Rightarrow \; & \beta_{0j}\;\mathbb{E} \left[ \left(\frac{\exp(\beta_{0j}+\bm\beta_{1j}^\top U) - 1 }{\pi_K+\sum_{l=1}^{K-1}\pi_l\exp(\beta_{0l}+\bm\beta_{1l}^\top U)}\right) \mathbb{I}\Big(\bm\beta_{1j}^\top U \geq 0\Big) \right]\nonumber\\& + \mathbb{E} \left[ \left(\frac{\exp(\beta_{0j}+\bm\beta_{1j}^\top U) - 1 }{\pi_K+\sum_{l=1}^{K-1}\pi_l\exp(\beta_{0l}+\bm\beta_{1l}^\top U)}\right) \bm\beta_{1j}^\top U \mathbb{I}\Big(\bm\beta_{1j}^\top U> 0\Big) \right] < \; 0.\label{me8}
\end{align}
Clearly the left-hand side of (\ref{me8}) is non-negative, implying that (\ref{me8}) can not be true. Hence $\bm\beta_{1j}\neq \bm{0}$ is not possible when $\beta_{0j}\geq 0$. Through the same line of arguments, it can be shown that $\bm\beta_{1j}\neq \bm{0}$ is not possible when $\beta_{0j} < 0$. Therefore, we must have $\bm\beta_{1j}=\bm{0}$ for all $j\in \{1,2,\dots,(K-1)\}$. Since $j$ is chosen arbitrarily, the proof is now complete.  \qed

\subsection{Proof of Theorem~\ref{theo:multi}}

\noindent Proof of Theorem~\ref{theo:multi} (a): Under $H_{0K},$ the VSC holds if we have with probability tending to one, $\tilde{\mathcal{A}}_{n,j} = \mathcal{A}_{n,j} = \emptyset,\text{ for all } j=1,\dots,(K-1).$ For brevity of notation define the set $\tilde{\mathcal{A}}_n:=\Big\{\tilde{\mathcal{A}}_{n,j} = \phi, j=1\dots,(K-1)\Big\}$. This is equivalent to $\frac{\tilde{u}_{jk}}{\sqrt{n}} + \beta_{jk} = 0,\text{ for all } j=1,\dots,(K-1),\text{ and }k=k_j=0,\dots,p,$ where $\tilde{\bm{u}}_n = \sqrt{n} (\tilde{\bm{\beta}}_n - \bm{\beta}_n)$ is the solution of \eqref{3} in proposition \ref{vscm}. Thus, due to equation \eqref{m2} and \eqref{m3}, it is enough to show that for all $j=1,\dots,(K-1),\text{ and }k=k_j=0,\dots,p$,
\begin{equation} \label{m19}
    -\frac{\lambda_{n}}{\sqrt{n}}\leq  \dfrac{1}{\sqrt{n}}\sum_{i=1}^{n}\left(y_{ik}-\pi_j\right)\dot{U}_{ik} + \dfrac{1}{\sqrt{n}}\sum_{i=1}^{n}\left(\pi_j - \frac{n_j}{n}\right)\dot{U}_{ik} \leq \dfrac{\lambda_{n}}{\sqrt{n}}.
\end{equation} 
Now, due to Lemma \ref{kcvl2}, under $H_{0K},$ with probability tending to one, 
\begin{multline*}
    \underset{1\leq j \leq (K-1)}{\max}\;\underset{0\leq k \leq p}{\max} \left| \dfrac{1}{\sqrt{n}}\sum_{i=1}^{n}\left(y_{ij}-\pi_j\right)\dot{U}_{ij} \right|\\ \leq C \left[\sqrt{\log(np)} + \frac{(\log(n))^{1/\kappa}(\log(np))^{1/\min\{1,\kappa\}}}{\sqrt{n}} \right].
\end{multline*}
Again, due to the assumption (C.3) and Lemma \ref{kcvl2}, we have with probability tending to one,
\begin{equation*} 
\begin{split}
    & \underset{1\leq j \leq (K-1)}{\max}\;\underset{0\leq k \leq p}{\max} \bigg|\dfrac{1}{\sqrt{n}}\sn \left(\pi_j-\frac{n_j}{n}\right)\dot{U}_{ik}\bigg| \\
    & \qquad \leq C \underset{1\leq j \leq (K-1)}{\max} \bigg\{\bigg|\pi_j-\frac{n_j}{n}\bigg|\bigg\} \\&\hspace{3cm} \underset{1\leq j \leq (K-1)}{\max}\;\underset{0\leq k \leq p}{\max} \bigg\{\bigg|\dfrac{1}{\sqrt{n}} \sn \Big(\dot{U}_{ik} - \mathbb{E}(\dot{U}_{ik}) \Big)\bigg| \\&\hspace{8cm} + \bigg|\dfrac{1}{\sqrt{n}} \sn \mathbb{E}(\dot{U}_{ik})\bigg| \bigg\}\\
    & \qquad = o(1).
\end{split}
\end{equation*}
Therefore, equation \eqref{m19} follows due to the assumption (A.4)(i) and we get
\begingroup
\begin{align*}
    &\mathbb{E}_{H_{0K}'}\left(\psi_{n, \alpha}^{(2)}\right)\\
    & = \mathbbm{P}_{H_{0K}'}\left( \underset{k}{\max}\; |T^{(2)}_{nk}| \leq m^{(2)}_{n,p,1-\alpha} \right)\\
    &= \mathbbm{P}_{H_{0K}'}\left( \Big\{\underset{k}{\max}\; |T^{(2)}_{nk}| \leq m^{(2)}_{n,p,1-\alpha}\Big\} \cap \tilde{\mathcal{A}}_n \right)\\&\hspace{1.5cm} + \mathbbm{P}_{H_{0K}'}\left( \Big\{\underset{k}{\max}\; |T^{(2)}_{nk}| \leq m^{(2)}_{n,p,1-\alpha}\Big\} \cap \tilde{\mathcal{A}}_n^c \right)\\
    &= \mathbbm{P}_{H_{0K}'}\left( \Big\{\underset{k}{\max}\; b_2(p, n)\;|Z_{2k}| \leq m^{(2)}_{n,p,1-\alpha}\Big\} \cap \tilde{\mathcal{A}}_n \right)\\&\hspace{0.25cm} + \mathbbm{P}_{H_{0K}'}\left( \Big\{\underset{k}{\max}\; b_2(p, n) \;|Z_{2k}| \leq m^{(2)}_{n,p,1-\alpha}\Big\} \cap \tilde{\mathcal{A}}_n^c \right) \\& \hspace{0.5cm} + \mathbbm{P}_{H_{0K}'}\left( \Big\{\underset{k}{\max}\; |T^{(2)}_{nk}| \leq m^{(2)}_{n,p,1-\alpha}\Big\} \cap \tilde{\mathcal{A}}_n^c \right)\\&\hspace{0.75cm} - \mathbbm{P}_{H_{0K}'}\left( \Big\{\underset{k}{\max}\; b_2(p, n) \;|Z_{2k}| \leq m^{(2)}_{n,p,1-\alpha}\Big\} \cap \tilde{\mathcal{A}}_n^c \right)\\
    &= \mathbbm{P}_{H_{0K}'}\left( \underset{k}{\max}\; b_2(p, n) \;|Z_{2k}| \leq m^{(2)}_{n,p,1-\alpha} \right)\\& \hspace{0.25cm} + \mathbbm{P}_{H_{0K}'}\left( \Big\{\underset{k}{\max}\; |T^{(2)}_{nk}| \leq m^{(2)}_{n,p,1-\alpha}\Big\} \cap \tilde{\mathcal{A}}_n^c \right)\\&\hspace{0.5cm} - \mathbbm{P}_{H_{0K}'}\left( \Big\{\underset{k}{\max}\; b_2(p, n) \;|Z_{2k}| \leq m^{(2)}_{n,p,1-\alpha}\Big\} \cap \tilde{\mathcal{A}}_n^c \right)\\
    \end{align*}
Therefore,
    \begin{align*}
    & \;\quad \bigg| \mathbbm{P}_{H_{0K}'}\left( \underset{k}{\max}\; |T^{(2)}_{nk}| > m^{(2)}_{n,p,1-\alpha} \right) - \alpha \bigg|\\
    &= \bigg|\mathbbm{P}_{H_{0K}'}\left( \Big\{\underset{k}{\max}\; |T^{(2)}_{nk}| \leq m^{(2)}_{n,p,1-\alpha}\Big\} \cap \tilde{\mathcal{A}}_n^c \right)\\&\hspace{0.5cm} - \mathbbm{P}_{H_{0K}'}\left( \Big\{\underset{k}{\max}\; b_2(p, n) \;|Z_{2k}| \leq m^{(2)}_{n,p,1-\alpha}\Big\} \cap \tilde{\mathcal{A}}_n^c \right)\bigg| \\
    & \leq 2\mathbbm{P}_{H_{0K}'}\left(\tilde{\mathcal{A}}_n^c\right).
\end{align*}
\endgroup
Hence, it follows that
\begin{equation*}
    \left|\mathbb{E}_{H_{0K}'}\left(\psi_{n, \alpha}^{(2)}\right) - \alpha\right| 
    = o(1).
\end{equation*} \qed

\noindent Proof of Theorem~\ref{theo:multi} (b):
 Note that, the test statistic can be written as $$ \bm{T}^{(2)}_n = \sqrt{n}\hat{\bm{\beta}}_n + b_2(p, n)\bm{Z}_2 = \bm{T}^{(2)}_{1n} + \sqrt{n} {\bm\beta}_n,$$ where under VSC, we can write $\bm{T}^{(2)}_{1n} =  \sqrt{n}\Big((\hat{\bm{\beta}}_n^{(1)} - {\bm{\beta}}_n^{(1)})^\top, \bm{0}^\top \Big) +  b_2(p, n)\bm{Z}_2 $ and ${\bm\beta}_n = \big({{\bm{\beta}}_n^\on}^\top, \bm{0}^\top\big)^\top = \big(({{\bm{\beta}}_{n,1}^\on})^\top,\dots,({{\bm{\beta}}_{n,(K-1)}^\on})^\top, \bm{0}^\top\big)^\top.$ 
Now following the proof of Lemma \ref{lem:asymnor}, and using Lemma \ref{vscm}, we have for any $j=1,\dots, (K-1)\;, k=k_j= 0,\dots,p_{0j}$,
   \begin{equation} \label{m22}
    \underset{x\in\mathbb{R}}{\sup} \;| \mathbbm{P}(T^{(2)}_{1nl} \leq  x) - {\Phi}(\sigma^{-1}_{nl} x)| = o(1),
\end{equation} 
where $\sigma^2_{nl} = ((\mathbb{E}\bm{L}_{11, n})^{-1})_{ll}$ and $l=\sum_{m=0}^{j-1}(p_{0m}+1)+k+1$. Note that, 
\begin{equation*}
    \sigma^2_{nl} \leq \|(\mathbb{E}\bm{L}_{11, n})^{-1}\|_\infty \leq C n^{a_2},
\end{equation*}
and due to Lemma 6 of \cite{tony2014two}, $m^{(2)}_{n,p,1-\alpha} \leq C \sqrt{\log(p)-\log\log(p)}.$ For brevity of notation, define $\breve{\mathcal{A}}_n = \Big\{\tilde{\mathcal{A}}_{n,j} = 0,\dots,p_{0j}, j=1\dots,(K-1)\Big\}$. Therefore, for some $j' = 1,\dots,(K-1),\; k'= k_j' = 0,\dots,(p_{0j'}-1),$ and $l' = \sum_{m=0}^{j'-1}(p_{0m}+1)+k'+1$, we have
\begingroup
\begin{align*}
    & \mathbb{E}_{H_{1K}'}\Big(\psi_{n, \alpha}^{(2)}\Big)\\
    & = 1 - \mathbbm{P}_{H_{1K}'} \Big( \underset{l}{\max}\; |T^{(2)}_{nl}| \leq m^{(2)}_{n,p,1-\alpha} \Big)\\
    & \geq 1 - \mathbbm{P}_{H_{1K}'} \Big( \Big\{\underset{l}{\max}\; |T^{(2)}_{nl}| \leq m^{(2)}_{n,p,1-\alpha}\Big\} \cap \breve{\mathcal{A}}_n \Big) - \mathbbm{P}_{H_{1K}'}\Big( \breve{\mathcal{A}}_n^c \Big)\\
    & \geq 1 - \mathbbm{P}_{H_{1K}'} \Big( \Big\{\underset{l}{\max}\; |T^{(2)}_{nl}| \leq m^{(2)}_{n,p,1-\alpha}\Big\} \cap \breve{\mathcal{A}}_n \Big) - o(1)\\
    & \geq 1 - \mathbbm{P}_{H_{1K}'} \Big( |T^{(2)}_{nl'}| \leq m^{(2)}_{n,p,1-\alpha}\Big\} \cap \breve{\mathcal{A}}_n \Big) - o(1)\\
    & \geq 1 - \mathbbm{P}_{H_{1K}'} \Big( \Big\{ |T^{(2)}_{1nl'}| \geq \sqrt{n} |\beta_{l'}| - m^{(2)}_{n,p,1-\alpha}\Big\}  \cap \breve{\mathcal{A}}_n \Big) - o(1)\\
    & \geq 1 - \Big| \mathbbm{P}_{H_{1K}'} \Big( T^{(2)}_{1nl'} \geq \sqrt{n} |\beta_{l'}| - m^{(2)}_{n,p,1-\alpha} \Big)\\&\hspace{5cm} - (1-\Phi(\sigma^{-1}_{nl'}[\sqrt{n}|\beta_{l'}| - m^{(2)}_{n,p,1-\alpha}])) \Big|\\
    & \qquad\qquad -\Big| \mathbbm{P}_{H_{1K}'} \Big( T^{(2)}_{1nl'} \leq - \sqrt{n} |\beta_{l'}| + m^{(2)}_{n,p,1-\alpha} \Big)\\&\hspace{5cm} - \Phi( \sigma^{-1}_{nl'}[-\sqrt{n}|\beta_{l'}| + m^{(2)}_{n,p,1-\alpha}]) \Big|\\
    & \qquad\qquad - \Big( 1-\Phi(\sigma^{-1}_{nl'}[\sqrt{n} |\beta_{l'}| - m^{(2)}_{n,p,1-\alpha}]) \Big) \\&\hspace{5cm} - \Phi( \sigma^{-1}_{nl'}[-\sqrt{n}|\beta_{l'}| + m^{(2)}_{n,p,1-\alpha}])- o(1)\\
    &\geq  1 - 2\Phi\Big( \sigma^{-1}_{nl'}[- \sqrt{n}|\beta_{l'}| + m^{(2)}_{n,p,1-\alpha}])- o(1)\\
    &\geq 1 - o(1),
\end{align*}
where the second inequality is due to VSC, which follows from Lemma~\ref{vscm}. The fourth inequality is due to the triangle inequality. The sixth inequality is due to \eqref{m22}. The last inequality is due to the beta-min condition. \qed

\section*{Appendix C: Additional real-data analyses} \label{sec: realdata}
The Leukemia dataset is based on the seminal study of \citet{OMLc1}, where gene expression measurements obtained using Affymetrix microarrays were used to distinguish between acute lymphoblastic leukemia (ALL) and acute myeloid leukemia (AML). In the original study, expression levels of $6{,}817$ genes were measured from bone marrow samples of leukemia patients. The key objective was to classify tumors by molecular profiles rather than morphological features, demonstrating that gene expression patterns alone can reliably distinguish biologically distinct cancer types. In our analysis, we consider an updated dataset comprising $n_1 = 47$ ALL samples and $n_2 = 25$ AML samples, with $p = 7{,}129$ genes available at \href{https://www.openml.org/search?type=data&status=active&id=1104}{OpenML} \citep{OpenML2017}.

The prostate cancer dataset is derived from the study of \citet{OMLc2}, where gene expression profiles from prostate tumors (52 samples) and normal tissues (50 samples) were analyzed using oligonucleotide microarrays containing approximately $12{,}600$ genes. This dataset is also available at \href{https://www.openml.org/search?type=data&status=active&id=45099&sort=runs}{OpenML}.


\begin{table}[t]
\centering
\caption{Rejection proportions (with standard errors in parentheses) and computation times for different methods.}
\label{tab:realdata2}

\vspace{0.1in}
\scriptsize
\begin{tabular}{lcccccc}
\toprule
 & D3 & D3op & CQ & XY & OP & ES \\
\midrule
\multicolumn{7}{c}{Rejection proportions over 500 random permutations} \\ \midrule
Leukemia & 0.052 (0.009) & 0.052 (0.009) & 0.060 (0.010) & 0.044 (0.009) & 0.056 (0.010) & 0.050 (0.009) \\
Prostate & 0.048 (0.009) & 0.044 (0.009) & 0.046 (0.009) & 0.054 (0.010) & 0.050 (0.009) & 0.052 (0.009) \\
\midrule
\multicolumn{7}{c}{Rejection proportions over 500 bootstrap datasets} \\ \midrule
Leukemia & 1.000 (0.000) & 1.000 (0.000) & 1.000 (0.000) & 0.848 (0.016) & 1.000 (0.000) & 0.854 (0.017) \\
Prostate & 1.000 (0.000) & 1.000 (0.000) & 1.000 (0.000) & 0.884 (0.014) & 1.000 (0.000) & 0.874 (0.014) \\
\midrule
\multicolumn{7}{c}{Average computation time in seconds over 50 iterations} \\ \midrule
Leukemia & 0.040 (0.005) & 0.263 (0.013) & 0.107 (0.032) & 0.943 (0.052) & 0.622 (0.050) & 1.514 (0.058) \\
Prostate & 0.081 (0.009) & 0.649 (0.021) & 0.213 (0.023) & 2.598 (0.095) & 1.332 (0.139) & 8.807 (0.366) \\
\bottomrule
\end{tabular}
\end{table}

All procedures reject the null hypothesis for both datasets. Specifically, the $p$-values for D3, D3op, CQ, and OP are below $0.001$ for both the Leukemia and Prostate datasets. The XY procedure yields $p$-values of $0.035$ and $0.014$ for the Leukemia and Prostate datasets, respectively, while the corresponding values for ES are $0.029$ and $0.015$. These results provide clear evidence of differences between the two population mean vectors in both datasets.

To assess empirical size, we randomly permute the sample labels $500$ times. As shown in Table~\ref{tab:realdata2}, all six procedures maintain reasonable size control, with rejection proportions close to the nominal level of $0.05$. For the Leukemia dataset, the empirical sizes range from $0.044$ to $0.060$, where CQ exhibits a slightly inflated size at $0.060$ relative to the nominal level. For the Prostate dataset, the empirical sizes range from $0.044$ to $0.054$. Overall, none of the retained procedures exhibits substantial size distortion in these examples.

For power evaluation, we generate $500$ bootstrap datasets by resampling independently within each population. D3, D3op, CQ, and OP attain empirical power equal to one for both datasets, indicating strong separation between the two groups. The powers of XY and ES are somewhat lower, with values of $0.848$ and $0.854$, respectively, for the Leukemia dataset and $0.884$ and $0.874$ for the Prostate dataset. Since the full-sample bootstrap results provide little separation among D3, D3op, CQ, and OP, we further investigate their performance under reduced sample sizes.

For this purpose, we conduct subsampling analysis by varying the proportion of observations retained from each group. The results are displayed in Figure~\ref{fig:subsampling_real}. For the Leukemia dataset, CQ exhibits the highest power at the smallest subsampling proportions and reaches power close to one by a proportion of approximately $0.3$. OP also performs strongly in the low-sample regime, followed by D3op. In contrast, D3 has relatively low power for the smallest subsamples but increases rapidly once approximately half of the observations are retained. XY and ES improve more gradually and require substantially larger subsamples to attain high power. For the Prostate dataset, D3op and OP provide the strongest power at the smallest proportions, while CQ also performs competitively and approaches power one by a subsampling proportion of approximately $0.6$. D3 again improves sharply after moderate subsample sizes, whereas XY and ES exhibit slower and nearly parallel increases. As the subsampling proportion approaches one, all procedures attain power close to one.

Finally, Table~\ref{tab:realdata2} reports average computation times over $50$ repetitions. D3 is the fastest method, requiring approximately $0.040$ seconds for the Leukemia dataset and $0.081$ seconds for the Prostate dataset. CQ is the second fastest, with corresponding average times of $0.107$ and $0.213$ seconds. D3op incurs additional cost because of its data-adaptive penalty selection, but remains substantially faster than XY, OP, and ES. Among the competing procedures, ES is the most computationally demanding, particularly for the higher-dimensional Prostate dataset, where its average running time is approximately $8.807$ seconds. Overall, D3 provides the greatest computational efficiency, while D3op remains competitive in both power and computation time. The subsampling results also show that the relative performance of the procedures depends on the underlying dataset and the available sample size.

\begin{figure}
\centering
\includegraphics[width=0.8\linewidth]{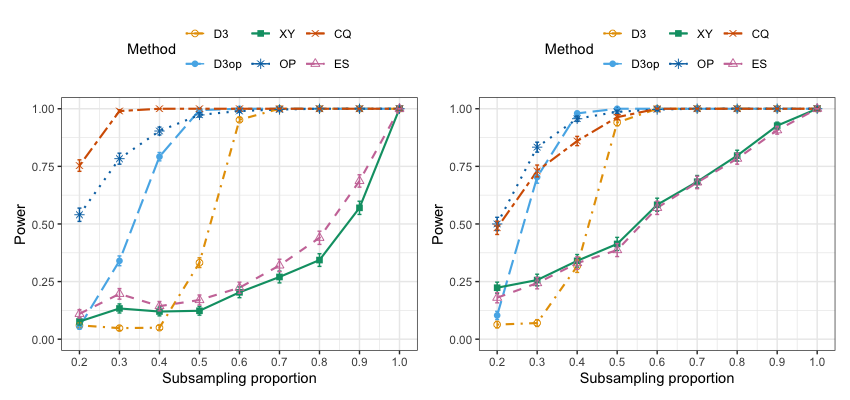}
\caption{Empirical power curves via $300$ subsamples (with vertical bars showing the standard errors). The left panel corresponds to the Leukemia data, while the right panel corresponds to the Prostate data.}
\label{fig:subsampling_real}
\end{figure}

\section*{Appendix D: Additional simulation studies} \label{sec: sim}
\subsection{Two-sample setting}
We compare D3 and D3op with several existing methods introduced earlier. The CQ test is implemented using the \texttt{PEtests} R package. Implementations of XY \citep{Xue2020}, OP \citep{huang2015projection}, and ES \citep{Kong2022} are based on codes provided by the respective authors. Owing to computational considerations, the XY procedure is implemented using $1{,}000$ bootstrap replications rather than the recommended $10{,}000$. For the ES procedure, we employ the non-studentized version of the infinity-norm statistic.

We consider a high-dimensional sparse mean-shift model with dimension $p=5{,}000$. Two sample-size configurations are examined {$(75,75)$, and $(100,50)$. Furthermore, we introduce coordinate-wise variance heterogeneity. The coordinate-wise variance factors are generated from a truncated log-normal distribution. Specifically,
\[
\sigma_j=\min\left\{\max\left(L_j,0.5\right),1.5\right\},
\qquad
L_j\sim \text{Lognormal}(0,0.15^2),
\]
independently for \(j=1,\ldots,p\). The same vector \(\sigma=(\sigma_1,\ldots,\sigma_p)\) is used for both populations, so the covariance structure remains common across groups. We also keep this choice fixed over all Monte Carlo iterations.

The latent observations are generated from one of the following correlation structures:

\begin{enumerate}
\item[(i)] Equicorrelation:
\[
R_{jk}=
\begin{cases}
1, & j=k,\\
\rho, & j\neq k,
\end{cases}
\]

\item[(ii)] AR(1):
\[
R_{jk}=\rho^{|j-k|},
\]

\item[(iii)] Block correlation:
coordinates are partitioned into blocks of size $50$, and within each block
\[
R_{jk}=
\begin{cases}
1, & j=k,\\
\rho, & j\neq k,
\end{cases}
\]
while different blocks are independent.
\end{enumerate}

Throughout the simulations, we set $\rho=0.3$. To assess robustness against deviations from normality, three marginal distributions are considered:

\begin{enumerate}
\item[(i)] standard Gaussian;
\item[(ii)] standardized chi-square distribution with $5$ degrees of freedom;
\item[(iii)] standardized Student's $t$-distribution with $5$ degrees of freedom.
\end{enumerate}

For empirical size, we set $\mu_1 = \mu_2 = 0$. For power, we fix $\mu_1 = 0$ and define $\mu_2$ as follows:
\[
\mu = s_j\delta \sqrt{\frac{\log p}{n}} \text{ and } 
\mu_2 = \big( \mu \cdot \mathbf{1}_{\lfloor 0.2\sqrt{p} \rfloor}^\top, \, 0_{p - \lfloor 0.2\sqrt{p} \rfloor}^\top \big)^\top,
\]
where $n=n_1+n_2$ and $s_j\in\{-1,1\}$ are independent Rademacher random variables. Consequently, positive and negative mean shifts occur with equal probability among the active coordinates. The signal strength parameter is varied over $\delta\in \{0,1,2,3\}.$

The observed samples are then generated according to
\[
X_i=\mu_1+D Z_i^{(1)},
\qquad
Y_j=\mu_2+D Z_j^{(2)},
\]
where $D=\mathrm{diag}(\sigma_1,\ldots,\sigma_p)$ and $Z_i^{(1)}, Z_j^{(2)}$ follow the selected dependence structure. Thus, both populations share the same covariance matrix, while exhibiting coordinate-wise heteroscedasticity arising from the randomly selected noisy coordinates.

All empirical sizes and powers are estimated using $3000$ Monte Carlo replications at significance level $\alpha=0.05$.

\begin{table}[t]
\centering
\caption{Empirical size and power for the two-sample setting with standard Gaussian marginals.}
\label{tab:sim_2sample_normal}
\scriptsize
\begin{tabular}{llcccccc}
\toprule
$(n_1,n_2)$ & Setting & D3 & D3op & CQ & XY & OP & ES \\
\midrule
\multicolumn{8}{c}{Equicorrelation: Empirical size ($\delta=0$)} \\ \midrule
$(75,75)$   &            & 0.060 & 0.058 & 0.077 & 0.050 & 0.049 & 0.053 \\
$(100,50)$  &            & 0.047 & 0.055 & 0.068 & 0.050 & 0.056 & 0.058 \\
\midrule
\multicolumn{8}{c}{Equicorrelation: Empirical power} \\ \midrule
$(75,75)$   & $\delta=1$ & 0.064 & 0.066 & 0.073 & 0.052 & 0.057 & 0.057 \\
            & $\delta=2$ & 0.431 & 0.393 & 0.062 & 0.181 & 0.152 & 0.192 \\
            & $\delta=3$ & 0.999 & 0.999 & 0.076 & 0.961 & 0.538 & 0.973 \\
$(100,50)$  & $\delta=1$ & 0.055 & 0.057 & 0.065 & 0.041 & 0.060 & 0.046 \\
            & $\delta=2$ & 0.189 & 0.294 & 0.072 & 0.136 & 0.146 & 0.151 \\
            & $\delta=3$ & 0.951 & 0.992 & 0.084 & 0.874 & 0.465 & 0.903 \\
\midrule
\multicolumn{8}{c}{AR(1): Empirical size ($\delta=0$)} \\ \midrule
$(75,75)$   &            & 0.065 & 0.055 & 0.048 & 0.041 & 0.044 & 0.046 \\
$(100,50)$  &            & 0.050 & 0.045 & 0.051 & 0.035 & 0.051 & 0.045 \\
\midrule
\multicolumn{8}{c}{AR(1): Empirical power} \\ \midrule
$(75,75)$   & $\delta=1$ & 0.064 & 0.059 & 0.090 & 0.041 & 0.051 & 0.049 \\
            & $\delta=2$ & 0.462 & 0.405 & 0.267 & 0.163 & 0.129 & 0.185 \\
            & $\delta=3$ & 0.999 & 0.998 & 0.727 & 0.938 & 0.438 & 0.951 \\
$(100,50)$  & $\delta=1$ & 0.053 & 0.060 & 0.082 & 0.045 & 0.065 & 0.059 \\
            & $\delta=2$ & 0.191 & 0.295 & 0.230 & 0.118 & 0.110 & 0.151 \\
            & $\delta=3$ & 0.952 & 0.988 & 0.642 & 0.814 & 0.355 & 0.869 \\
\midrule
\multicolumn{8}{c}{Block correlation: Empirical size ($\delta=0$)} \\ \midrule
$(75,75)$   &            & 0.061 & 0.056 & 0.060 & 0.043 & 0.054 & 0.053 \\
$(100,50)$  &            & 0.052 & 0.053 & 0.060 & 0.040 & 0.050 & 0.052 \\
\midrule
\multicolumn{8}{c}{Block correlation: Empirical power} \\ \midrule
$(75,75)$   & $\delta=1$ & 0.061 & 0.062 & 0.060 & 0.045 & 0.050 & 0.053 \\
            & $\delta=2$ & 0.462 & 0.400 & 0.135 & 0.183 & 0.074 & 0.204 \\
            & $\delta=3$ & 0.999 & 0.997 & 0.285 & 0.948 & 0.163 & 0.963 \\
$(100,50)$  & $\delta=1$ & 0.056 & 0.068 & 0.057 & 0.039 & 0.044 & 0.055 \\
            & $\delta=2$ & 0.196 & 0.300 & 0.128 & 0.109 & 0.073 & 0.141 \\
            & $\delta=3$ & 0.955 & 0.988 & 0.235 & 0.819 & 0.140 & 0.873 \\
\bottomrule
\end{tabular}
\end{table}

\begin{table}[t]
\centering
\caption{Empirical size and power for the two-sample setting with standardized $\chi^2_5$ marginals.}
\label{tab:sim_2sample_chisq}
\scriptsize
\begin{tabular}{llcccccc}
\toprule
$(n_1,n_2)$ & Setting & D3 & D3op & CQ & XY & OP & ES \\
\midrule
\multicolumn{8}{c}{Equicorrelation: Empirical size ($\delta=0$)} \\ \midrule
$(75,75)$   &            & 0.046 & 0.050 & 0.061 & 0.036 & 0.052 & 0.049 \\
$(100,50)$  &            & 0.051 & 0.054 & 0.064 & 0.032 & 0.041 & 0.045 \\
\midrule
\multicolumn{8}{c}{Equicorrelation: Empirical power} \\ \midrule
$(75,75)$   & $\delta=1$ & 0.046 & 0.053 & 0.064 & 0.033 & 0.051 & 0.046 \\
            & $\delta=2$ & 0.461 & 0.475 & 0.077 & 0.170 & 0.168 & 0.212 \\
            & $\delta=3$ & 0.999 & 0.999 & 0.081 & 0.922 & 0.550 & 0.954 \\
$(100,50)$  & $\delta=1$ & 0.051 & 0.056 & 0.073 & 0.037 & 0.058 & 0.053 \\
            & $\delta=2$ & 0.201 & 0.311 & 0.072 & 0.098 & 0.143 & 0.144 \\
            & $\delta=3$ & 0.953 & 0.990 & 0.072 & 0.772 & 0.469 & 0.867 \\
\midrule
\multicolumn{8}{c}{AR(1): Empirical size ($\delta=0$)} \\ \midrule
$(75,75)$   &            & 0.054 & 0.050 & 0.049 & 0.030 & 0.052 & 0.049 \\
$(100,50)$  &            & 0.055 & 0.052 & 0.047 & 0.027 & 0.050 & 0.057 \\
\midrule
\multicolumn{8}{c}{AR(1): Empirical power} \\ \midrule
$(75,75)$   & $\delta=1$ & 0.061 & 0.061 & 0.085 & 0.027 & 0.053 & 0.050 \\
            & $\delta=2$ & 0.476 & 0.473 & 0.263 & 0.120 & 0.129 & 0.179 \\
            & $\delta=3$ & 0.999 & 0.999 & 0.717 & 0.892 & 0.451 & 0.946 \\
$(100,50)$  & $\delta=1$ & 0.053 & 0.058 & 0.079 & 0.028 & 0.051 & 0.052 \\
            & $\delta=2$ & 0.198 & 0.307 & 0.240 & 0.070 & 0.108 & 0.126 \\
            & $\delta=3$ & 0.962 & 0.985 & 0.661 & 0.696 & 0.385 & 0.842 \\
\midrule
\multicolumn{8}{c}{Block correlation: Empirical size ($\delta=0$)} \\ \midrule
$(75,75)$   &            & 0.056 & 0.052 & 0.055 & 0.031 & 0.055 & 0.045 \\
$(100,50)$  &            & 0.066 & 0.047 & 0.056 & 0.027 & 0.048 & 0.053 \\
\midrule
\multicolumn{8}{c}{Block correlation: Empirical power} \\ \midrule
$(75,75)$   & $\delta=1$ & 0.067 & 0.062 & 0.076 & 0.036 & 0.049 & 0.059 \\
            & $\delta=2$ & 0.481 & 0.480 & 0.132 & 0.123 & 0.080 & 0.184 \\
            & $\delta=3$ & 0.999 & 0.999 & 0.297 & 0.903 & 0.177 & 0.948 \\
$(100,50)$  & $\delta=1$ & 0.053 & 0.060 & 0.078 & 0.030 & 0.055 & 0.062 \\
            & $\delta=2$ & 0.214 & 0.322 & 0.121 & 0.080 & 0.075 & 0.145 \\
            & $\delta=3$ & 0.965 & 0.986 & 0.247 & 0.694 & 0.149 & 0.840 \\
\bottomrule
\end{tabular}
\end{table}

\begin{table}[t]
\centering
\caption{Empirical size and power for the two-sample setting with standardized $t_5$ marginals.}
\label{tab:sim_2sample_t5}
\scriptsize
\begin{tabular}{llcccccc}
\toprule
$(n_1,n_2)$ & Setting & D3 & D3op & CQ & XY & OP & ES \\
\midrule
\multicolumn{8}{c}{Equicorrelation: Empirical size ($\delta=0$)} \\ \midrule
$(75,75)$   &            & 0.047 & 0.057 & 0.072 & 0.017 & 0.047 & 0.050 \\
$(100,50)$  &            & 0.053 & 0.047 & 0.070 & 0.010 & 0.054 & 0.046 \\
\midrule
\multicolumn{8}{c}{Equicorrelation: Empirical power} \\ \midrule
$(75,75)$   & $\delta=1$ & 0.055 & 0.061 & 0.073 & 0.022 & 0.058 & 0.054 \\
            & $\delta=2$ & 0.476 & 0.512 & 0.079 & 0.080 & 0.165 & 0.176 \\
            & $\delta=3$ & 0.998 & 0.999 & 0.081 & 0.688 & 0.549 & 0.923 \\
$(100,50)$  & $\delta=1$ & 0.055 & 0.064 & 0.072 & 0.013 & 0.059 & 0.054 \\
            & $\delta=2$ & 0.228 & 0.391 & 0.077 & 0.035 & 0.142 & 0.127 \\
            & $\delta=3$ & 0.960 & 0.993 & 0.076 & 0.437 & 0.471 & 0.802 \\
\midrule
\multicolumn{8}{c}{AR(1): Empirical size ($\delta=0$)} \\ \midrule
$(75,75)$   &            & 0.059 & 0.053 & 0.048 & 0.016 & 0.044 & 0.047 \\
$(100,50)$  &            & 0.053 & 0.054 & 0.049 & 0.010 & 0.049 & 0.051 \\
\midrule
\multicolumn{8}{c}{AR(1): Empirical power} \\ \midrule
$(75,75)$   & $\delta=1$ & 0.057 & 0.062 & 0.084 & 0.012 & 0.059 & 0.047 \\
            & $\delta=2$ & 0.498 & 0.511 & 0.273 & 0.063 & 0.142 & 0.161 \\
            & $\delta=3$ & 0.999 & 0.999 & 0.725 & 0.632 & 0.448 & 0.917 \\
$(100,50)$  & $\delta=1$ & 0.059 & 0.060 & 0.087 & 0.006 & 0.066 & 0.054 \\
            & $\delta=2$ & 0.223 & 0.376 & 0.238 & 0.019 & 0.119 & 0.112 \\
            & $\delta=3$ & 0.959 & 0.992 & 0.626 & 0.300 & 0.381 & 0.736 \\
\midrule
\multicolumn{8}{c}{Block correlation: Empirical size ($\delta=0$)} \\ \midrule
$(75,75)$   &            & 0.055 & 0.057 & 0.054 & 0.013 & 0.046 & 0.048 \\
$(100,50)$  &            & 0.053 & 0.054 & 0.055 & 0.008 & 0.047 & 0.056 \\
\midrule
\multicolumn{8}{c}{Block correlation: Empirical power} \\ \midrule
$(75,75)$   & $\delta=1$ & 0.061 & 0.061 & 0.066 & 0.012 & 0.049 & 0.051 \\
            & $\delta=2$ & 0.492 & 0.499 & 0.119 & 0.049 & 0.076 & 0.154 \\
            & $\delta=3$ & 0.999 & 0.999 & 0.301 & 0.635 & 0.173 & 0.904 \\
$(100,50)$  & $\delta=1$ & 0.050 & 0.063 & 0.079 & 0.005 & 0.053 & 0.050 \\
            & $\delta=2$ & 0.230 & 0.375 & 0.125 & 0.022 & 0.074 & 0.122 \\
            & $\delta=3$ & 0.975 & 0.995 & 0.251 & 0.313 & 0.150 & 0.768 \\
\bottomrule
\end{tabular}
\end{table}

The results in Tables~\ref{tab:sim_2sample_normal}, \ref{tab:sim_2sample_chisq}, and \ref{tab:sim_2sample_t5} summarize the empirical size and power under standard Gaussian, standardized $\chi^2_5$, and standardized $t_5$ marginals, respectively. Across the three covariance structures and both sample-size configurations, D3 and D3op generally maintain accurate Type~I error control, with empirical rejection probabilities close to the nominal level of $0.05$. OP and ES also exhibit satisfactory size performance. CQ is reasonably well calibrated overall, although it is slightly liberal under equicorrelation, with empirical sizes ranging from $0.061$ to $0.077$. In contrast, XY becomes increasingly conservative as the marginal distribution departs from Gaussianity. This behavior is most pronounced under the standardized $t_5$ distribution, where its empirical size ranges from $0.008$ to $0.017$.

When the signal is weak ($\delta=1$), the rejection probabilities of all procedures remain close to their corresponding null rejection probabilities, indicating that this signal level is difficult to distinguish from the null when $p=5{,}000$. More substantial differences emerge at the moderate signal level $\delta=2$. For the balanced design $(n_1,n_2)=(75,75)$, D3 and D3op have similar performance and clearly outperform the competing procedures across all marginal distributions and covariance structures. Their empirical powers range from $0.393$ to $0.512$, whereas the largest power among the competitors is $0.273$. D3 is slightly more powerful than D3op in all three Gaussian settings, while the two methods are nearly indistinguishable under the standardized $\chi^2_5$ distribution. Under the standardized $t_5$ distribution, D3op generally has a modest advantage.

The benefit of the data-adaptive penalty is more apparent for the unbalanced design $(100,50)$. At $\delta=2$, D3op uniformly exceeds D3 across the nine combinations of covariance structure and marginal distribution. Its power ranges from $0.294$ to $0.391$, compared with $0.189$ to $0.230$ for D3. The competing procedures remain substantially less powerful, with none attaining power above $0.240$. Thus, although D3 and D3op perform similarly in the balanced setting, the adaptive penalty used by D3op provides a clear advantage when the sample sizes are unequal.

At the stronger signal level $\delta=3$, D3 and D3op achieve power close to one throughout. Under the balanced design, both procedures have power of at least $0.997$, while under the unbalanced design D3op ranges from $0.985$ to $0.995$ and D3 ranges from $0.951$ to $0.975$. Among the competing procedures, performance depends strongly on the dependence structure and marginal distribution. ES and XY perform well under Gaussian marginals, particularly for balanced samples, but both lose power under unequal sample sizes and non-Gaussian observations. CQ performs relatively well under the AR(1) structure, with power between $0.626$ and $0.727$ at $\delta=3$, but is substantially weaker under equicorrelation and block correlation. OP also performs reasonably under equicorrelation but deteriorates markedly under block correlation.

Comparisons across marginal distributions further demonstrate the robustness of D3 and D3op. Their empirical size remains stable, and their power changes only modestly under skewness and heavy tails. In fact, at $\delta=2$, their balanced-sample power is often slightly higher under the standardized $\chi^2_5$ and $t_5$ distributions than under Gaussian marginals. In contrast, XY is highly sensitive to heavy tails: under the standardized $t_5$ distribution, its power at $\delta=2$ ranges only from $0.012$ to $0.080$, and substantial losses remain even at $\delta=3$. CQ and ES also exhibit reductions in several non-Gaussian settings, although the extent of the deterioration depends on the covariance structure.

These findings support the motivation for the proposed methodology. Although the two populations share the same covariance matrix, the coordinate-wise variance factors generated from the truncated log-normal distribution create heterogeneous signal-to-noise ratios across features. The sparse mean shifts are therefore embedded within a large collection of irrelevant coordinates having nonuniform marginal variability. By using penalized classification to identify a smaller set of discriminative variables, D3 and D3op reduce the effect of the heterogeneous background coordinates before constructing the test statistic. The adaptive penalty employed by D3op is particularly beneficial under sample-size imbalance, while retaining reliable Type~I error control across the dependence structures and marginal distributions considered.

\subsection{Three-sample setting}

We extend the numerical study to the three-sample problem. Two sample-size configurations are considered: $(75,75,75)$ and $(100,75,50)$.

The data-generation mechanism, dependence structures, marginal distributions, and active-set construction remain identical to those used in the two-sample setting.

Under the null hypothesis, we set
$\mu_1=\mu_2=\mu_3=0.$ Under the alternative, we
take $\mu_1$ = 0, $\mu_2$ as defined above, and construct $\mu_3$ by randomly permuting (which is kept fixed over the Monte Carlo iterations) the coordinates of $\mu_2$. This preserves both the sparsity level and signal magnitude while allowing the locations of the active coordinates to differ across populations.

\begin{table}[t]
\centering
\caption{Empirical size and power for the three-sample setting with Gaussian marginals.}
\label{tab:sim_3sample_normal}
\scriptsize
\begin{tabular}{llcccc}
\toprule
$(n_1,n_2,n_3)$ & Setting & D3 & D3op & KDCF & HDT \\
\midrule
\multicolumn{6}{c}{Equicorrelation: Empirical size ($\delta=0$)} \\ \midrule
$(75,75,75)$ &  & 0.058 & 0.053 & 0.043 & 0.074 \\
$(100,75,50)$ &  & 0.055 & 0.058 & 0.032 & 0.068 \\
\midrule
\multicolumn{6}{c}{Equicorrelation: Empirical power} \\ \midrule
$(75,75,75)$ & $\delta=1$ & 0.059 & 0.060 & 0.042 & 0.071 \\
 & $\delta=2$ & 0.204 & 0.221 & 0.068 & 0.073 \\
 & $\delta=3$ & 0.956 & 0.954 & 0.432 & 0.073 \\
$(100,75,50)$ & $\delta=1$ & 0.065 & 0.061 & 0.041 & 0.071 \\
 & $\delta=2$ & 0.432 & 0.391 & 0.065 & 0.081 \\
 & $\delta=3$ & 0.999 & 0.998 & 0.407 & 0.073 \\
\midrule
\multicolumn{6}{c}{AR(1): Empirical size ($\delta=0$)} \\ \midrule
$(75,75,75)$ &  & 0.068 & 0.053 & 0.036 & 0.046 \\
$(100,75,50)$ &  & 0.061 & 0.060 & 0.036 & 0.054 \\
\midrule
\multicolumn{6}{c}{AR(1): Empirical power} \\ \midrule
$(75,75,75)$ & $\delta=1$ & 0.057 & 0.063 & 0.043 & 0.069 \\
 & $\delta=2$ & 0.229 & 0.229 & 0.055 & 0.093 \\
 & $\delta=3$ & 0.969 & 0.963 & 0.406 & 0.154 \\
$(100,75,50)$ & $\delta=1$ & 0.069 & 0.068 & 0.034 & 0.047 \\
 & $\delta=2$ & 0.445 & 0.386 & 0.050 & 0.081 \\
 & $\delta=3$ & 0.999 & 0.994 & 0.364 & 0.135 \\
\midrule
\multicolumn{6}{c}{Block correlation: Empirical size ($\delta=0$)} \\ \midrule
$(75,75,75)$ &  & 0.057 & 0.052 & 0.037 & 0.063 \\
$(100,75,50)$ &  & 0.059 & 0.055 & 0.042 & 0.060 \\
\midrule
\multicolumn{6}{c}{Block correlation: Empirical power} \\ \midrule
$(75,75,75)$ & $\delta=1$ & 0.063 & 0.056 & 0.039 & 0.059 \\
 & $\delta=2$ & 0.216 & 0.224 & 0.051 & 0.068 \\
 & $\delta=3$ & 0.961 & 0.957 & 0.406 & 0.103 \\
$(100,75,50)$ & $\delta=1$ & 0.059 & 0.056 & 0.033 & 0.062 \\
 & $\delta=2$ & 0.444 & 0.379 & 0.055 & 0.074 \\
 & $\delta=3$ & 1.000 & 0.997 & 0.369 & 0.087 \\
\bottomrule
\end{tabular}
\end{table}

\begin{table}[t]
\centering
\caption{Empirical size and power for the three-sample setting with standardized $\chi^2_5$ marginals.}
\label{tab:sim_3sample_chisq}
\scriptsize
\begin{tabular}{llcccc}
\toprule
$(n_1,n_2,n_3)$ & Setting & D3 & D3op & KDCF & HDT \\
\midrule
\multicolumn{6}{c}{Equicorrelation: Empirical size ($\delta=0$)} \\ \midrule
$(75,75,75)$ &  & 0.055 & 0.058 & 0.034 & 0.071 \\
$(100,75,50)$ &  & 0.057 & 0.055 & 0.031 & 0.067 \\
\midrule
\multicolumn{6}{c}{Equicorrelation: Empirical power} \\ \midrule
$(75,75,75)$ & $\delta=1$ & 0.065 & 0.054 & 0.032 & 0.067 \\
 & $\delta=2$ & 0.235 & 0.239 & 0.047 & 0.073 \\
 & $\delta=3$ & 0.969 & 0.966 & 0.333 & 0.074 \\
$(100,75,50)$ & $\delta=1$ & 0.065 & 0.060 & 0.031 & 0.064 \\
 & $\delta=2$ & 0.435 & 0.412 & 0.040 & 0.071 \\
 & $\delta=3$ & 0.999 & 0.998 & 0.292 & 0.073 \\
\midrule
\multicolumn{6}{c}{AR(1): Empirical size ($\delta=0$)} \\ \midrule
$(75,75,75)$ &  & 0.054 & 0.057 & 0.027 & 0.054 \\
$(100,75,50)$ &  & 0.060 & 0.059 & 0.026 & 0.048 \\
\midrule
\multicolumn{6}{c}{AR(1): Empirical power} \\ \midrule
$(75,75,75)$ & $\delta=1$ & 0.062 & 0.057 & 0.030 & 0.059 \\
 & $\delta=2$ & 0.229 & 0.226 & 0.039 & 0.087 \\
 & $\delta=3$ & 0.963 & 0.960 & 0.309 & 0.157 \\
$(100,75,50)$ & $\delta=1$ & 0.070 & 0.062 & 0.031 & 0.060 \\
 & $\delta=2$ & 0.457 & 0.425 & 0.036 & 0.086 \\
 & $\delta=3$ & 0.999 & 0.997 & 0.236 & 0.155 \\
\midrule
\multicolumn{6}{c}{Block correlation: Empirical size ($\delta=0$)} \\ \midrule
$(75,75,75)$ &  & 0.053 & 0.051 & 0.028 & 0.055 \\
$(100,75,50)$ &  & 0.060 & 0.056 & 0.029 & 0.056 \\
\midrule
\multicolumn{6}{c}{Block correlation: Empirical power} \\ \midrule
$(75,75,75)$ & $\delta=1$ & 0.059 & 0.060 & 0.026 & 0.062 \\
 & $\delta=2$ & 0.243 & 0.235 & 0.042 & 0.074 \\
 & $\delta=3$ & 0.972 & 0.963 & 0.302 & 0.099 \\
$(100,75,50)$ & $\delta=1$ & 0.067 & 0.061 & 0.029 & 0.062 \\
 & $\delta=2$ & 0.458 & 0.414 & 0.037 & 0.071 \\
 & $\delta=3$ & 0.999 & 0.998 & 0.246 & 0.097 \\
\bottomrule
\end{tabular}
\end{table}

\begin{table}[t]
\centering
\caption{Empirical size and power for the three-sample setting with standardized $t_5$ marginals.}
\label{tab:sim_3sample_t5}
\scriptsize
\begin{tabular}{llcccc}
\toprule
$(n_1,n_2,n_3)$ & Setting & D3 & D3op & KDCF & HDT \\
\midrule
\multicolumn{6}{c}{Equicorrelation: Empirical size ($\delta=0$)} \\ \midrule
$(75,75,75)$ &  & 0.046 & 0.060 & 0.010 & 0.070 \\
$(100,75,50)$ &  & 0.055 & 0.052 & 0.008 & 0.065 \\
\midrule
\multicolumn{6}{c}{Equicorrelation: Empirical power} \\ \midrule
$(75,75,75)$ & $\delta=1$ & 0.057 & 0.065 & 0.008 & 0.072 \\
 & $\delta=2$ & 0.235 & 0.267 & 0.017 & 0.069 \\
 & $\delta=3$ & 0.971 & 0.979 & 0.127 & 0.079 \\
$(100,75,50)$ & $\delta=1$ & 0.060 & 0.064 & 0.008 & 0.076 \\
 & $\delta=2$ & 0.440 & 0.447 & 0.010 & 0.066 \\
 & $\delta=3$ & 0.999 & 0.998 & 0.095 & 0.077 \\
\midrule
\multicolumn{6}{c}{AR(1): Empirical size ($\delta=0$)} \\ \midrule
$(75,75,75)$ &  & 0.056 & 0.054 & 0.008 & 0.050 \\
$(100,75,50)$ &  & 0.063 & 0.061 & 0.003 & 0.049 \\
\midrule
\multicolumn{6}{c}{AR(1): Empirical power} \\ \midrule
$(75,75,75)$ & $\delta=1$ & 0.059 & 0.065 & 0.005 & 0.058 \\
 & $\delta=2$ & 0.249 & 0.282 & 0.012 & 0.086 \\
 & $\delta=3$ & 0.973 & 0.978 & 0.103 & 0.168 \\
$(100,75,50)$ & $\delta=1$ & 0.067 & 0.067 & 0.005 & 0.067 \\
 & $\delta=2$ & 0.480 & 0.459 & 0.006 & 0.084 \\
 & $\delta=3$ & 0.998 & 0.998 & 0.051 & 0.149 \\
\midrule
\multicolumn{6}{c}{Block correlation: Empirical size ($\delta=0$)} \\ \midrule
$(75,75,75)$ &  & 0.053 & 0.054 & 0.010 & 0.055 \\
$(100,75,50)$ &  & 0.057 & 0.058 & 0.006 & 0.051 \\
\midrule
\multicolumn{6}{c}{Block correlation: Empirical power} \\ \midrule
$(75,75,75)$ & $\delta=1$ & 0.056 & 0.058 & 0.005 & 0.068 \\
 & $\delta=2$ & 0.239 & 0.279 & 0.012 & 0.069 \\
 & $\delta=3$ & 0.972 & 0.980 & 0.103 & 0.098 \\
$(100,75,50)$ & $\delta=1$ & 0.062 & 0.065 & 0.004 & 0.062 \\
 & $\delta=2$ & 0.465 & 0.459 & 0.003 & 0.071 \\
 & $\delta=3$ & 0.999 & 0.998 & 0.052 & 0.084 \\
\bottomrule
\end{tabular}
\end{table}

The results in Tables~\ref{tab:sim_3sample_normal},
\ref{tab:sim_3sample_chisq} and \ref{tab:sim_3sample_t5}
summarize the performance of the competing multi-sample procedures
under standard Gaussian, standardized $\chi^2_5$, and standardized
$t_5$ marginals, respectively. Across the covariance structures and
sample-size configurations considered, D3 and D3op generally maintain
accurate Type~I error control, with empirical rejection probabilities
close to the nominal level of $0.05$. KDCF is conservative throughout,
and its conservativeness becomes substantially more pronounced under
the standardized $t_5$ distribution. HDT is reasonably calibrated
under the AR(1) and block-correlation structures but exhibits mild size
inflation under equicorrelation, where its empirical size ranges from
$0.065$ to $0.074$.

At the weakest signal level, $\delta=1$, the rejection probabilities of
all procedures remain close to their corresponding null rejection
rates, indicating that the sparse alternative is difficult to detect
at this signal strength. More substantial differences emerge at
$\delta=2$. Across all marginal distributions, covariance structures,
and sample-size configurations, D3 and D3op attain markedly higher
power than KDCF and HDT. Under the balanced design, the two proposed
procedures perform similarly, with D3op having a modest advantage in
several settings. Under the unbalanced design, D3 is more powerful in
most configurations, although the relative ordering is not uniform.
In either case, the powers of KDCF and HDT remain close to their null
rejection probabilities.

At the stronger signal level, $\delta=3$, D3 and D3op achieve power
close to one in every configuration. Under Gaussian marginals, KDCF
attains moderate power, ranging from approximately $0.36$ to $0.43$,
while HDT remains substantially less powerful. Both competing methods
deteriorate under non-Gaussian marginals. Under the standardized
$\chi^2_5$ distribution, the power of KDCF ranges from $0.236$ to
$0.333$, whereas under the standardized $t_5$ distribution it does not
exceed $0.127$. HDT also remains weak at $\delta=3$, with power below
$0.17$ throughout.

Comparisons across the three marginal distributions further illustrate the robustness of D3 and D3op. Their empirical sizes remain stable under Gaussian, skewed, and heavy-tailed observations, and their powers vary only modestly across the three distributions. In contrast, KDCF is highly sensitive to departures from Gaussianity. Its empirical size becomes increasingly conservative, and its power declines sharply under the standardized $\chi^2_5$ and especially the standardized $t_5$ distributions. HDT is less affected by the marginal distribution than KDCF, but its power remains consistently well below that of D3 and D3op.

These results indicate that the proposed methodology extends effectively from the two-sample to the three-sample setting. Here, the nonzero mean components may occur at different coordinate locations across the populations, while the observations are also subject to dependence and coordinate-wise variance heterogeneity. By using penalized classification to screen the high-dimensional feature space before constructing the test statistic, D3 and D3op retain reliable Type~I error control and achieve substantially higher power than the competing procedures across the dependence structures and marginal distributions
considered.

\end{document}